\documentclass[11pt,reqno]{amsart}
\usepackage{amsfonts,amsmath,amssymb,enumerate}
\usepackage{amsthm}
\usepackage{mathrsfs}  
\usepackage{bm}
\usepackage{color}
\numberwithin{equation}{section}
\theoremstyle{plain}
\newtheorem{thm}{Theorem}[section]
\newtheorem{lem}[thm]{Lemma}
\newtheorem{prop}[thm]{Proposition}
\newtheorem{cor}[thm]{Corollary}
\newtheorem{conj}[thm]{Conjecture}
\theoremstyle{definition}
\newtheorem{defn}[thm]{Definition}

\newtheorem{exmp}[thm]{Example}
\theoremstyle{remark}
\newtheorem{rem}[thm]{Remark}

\usepackage{amsmath}
\usepackage{amsxtra}
\usepackage{amscd}
\usepackage{amsthm}
\usepackage{amsfonts}
\usepackage{amssymb}
\usepackage{eucal}
\usepackage{epsfig}
\usepackage{graphics}
\usepackage{accents}

\usepackage{mathtools}
\usepackage{etex}
\usepackage{tikz}
\usetikzlibrary{decorations.markings,shapes,matrix,calc}
\usepackage{here} 

\usepackage[normalem]{ulem}
\usepackage{cancel}

\usepackage[linktocpage=true,colorlinks=true, linkcolor=blue, citecolor=red, urlcolor=green]{hyperref}

\newcommand{\C}{{\mathbb C}}
\newcommand{\Z}{{\mathbb Z}}

\newcommand{\R}{{\mathbb R}}

\newcommand{\be}{\begin{equation*}}    
\newcommand{\ee}{\end{equation*}}    
\newcommand{\ben}{\begin{equation}}    
\newcommand{\een}{\end{equation}}    
\newcommand{\bea}{\begin{align*}}    
\newcommand{\eea}{\end{align*}}    
\newcommand{\bean}{\begin{align}}    
\newcommand{\eean}{{\end{align}}}

\newcommand{\gl}{\mathfrak{gl}}

\newcommand{\slt}{\mathfrak{sl}_2}
\newcommand{\slth}{\widehat{\mathfrak{sl}}_2}

\newcommand{\e}{\epsilon}

\newcommand{\la}{\lambda}

\newcommand{\res}{{\rm res}}

\newcommand{\on}{\operatorname}

\newcommand{\mc}{\mathcal}

\newcommand{\floori}[1]{ \left\lfloor #1 \right\rfloor}
\newsavebox{\mycases}

\begin{document}

\author{D. Masoero, E. Mukhin, and A. Raimondo }

\address{DM: Grupo de F\'isica Matem\'atica da Universidade de Lisboa \\
Campo Grande Edif\'icio C6, 1749-016 Lisboa, Portugal}
\email{dmasoero@gmail.com}
\address{EM: Department of Mathematical Sciences,
Indiana University Indianapolis,
402 N. Blackford St., LD 270, 
Indianapolis, IN 46202, USA}
\email{emukhin@iu.edu} 
\address{AR: Dipartimento di Ingegneria Gestionale, dell’Informazione e della Produzione, Universit\`a di Bergamo, Italy, and INFN, Sezione di Milano-Bicocca, Piazza della Scienza 3, 20126 Milano, Italy.}
\email{andrea.raimondo@unibg.it}

\title{$Q$-functions for lambda opers}
\begin{abstract}
We consider the Schr{\"o}dinger operators which are constructed from the $\lambda$-opers corresponding to solutions of the $\widehat{\mathfrak{sl}}_2$ Gaudin Bethe Ansatz equations. We define and study the connection coefficients called the $Q$-functions. We conjecture that the $Q$-functions obtained from the $\lambda$-opers
 coincide with the $Q$-functions of the Bazhanov-Lukyanov-Zamolodchikov opers with the monster potential related to the quantum KdV flows. We give supporting evidence for this conjecture.
\end{abstract}

\maketitle

\tableofcontents

\section{Introduction}
In this paper we study the Schr{\"o}dinger operators of the form 
\bean\label{L Gaudin-1}
L^G=\partial_x^2-\left(\frac{l(l+1)}{x^2}+\sum_{i=1}^{d_0}\frac{2}{(x-s_i)^2}+\sum_{i=1}^{d_0}\frac{k+s_i/(s_i-1)}{x(x-s_i)}+x^k(x-1)\la^2\right),
\end{align}
where $k>-2$, $\lambda\in\C \setminus\{0\}$ is a spectral parameter and the  $s_i$  are pairwise distinct non-zero complex numbers such that for every $i=1,\dots,d_0$ the monodromy around $x=s_i$ is trivial for every value of $\la$. The motivation for our study comes from integrable systems. The main results can be summarized as follows.
\begin{itemize}
\item We show that the operators \eqref{L Gaudin-1} are obtained from the Bethe Ansatz equations for the  $\slth$ Gaudin model, and prove that the connection coefficients ($Q$-functions) obtained by considering the central connection problem
for the operator \eqref{L Gaudin-1} satisfy the \lq quantum Wronskian condition\rq\,  (or the $QQ$-system) appearing in the quantum KdV model \cite{BLZ1}. This provides an ODE/IM correspondence between operators \eqref{L Gaudin-1} and the quantum KdV model. We describe some analytical properties of the $Q$-functions obtained from \eqref{L Gaudin-1}.
\item The classical ODE/IM correspondence \cite{DTa,DT,BLZ4} studies 
operators known as \lq monster potential Schr{\"o}dinger operators\rq :
\ben\label{1}
L^{BLZ}=\partial_x^2-\left(\frac{\bar l(\bar l+1)}{x^2}+\frac{1}{x}+\sum_{i=1}^{d}\Big(\frac{2}{(x-\bar{z}_i)^2}+\frac{  \bar k-2}{x(x-\bar{z}_i)}\Big) -x^{\bar k-2}\bar\la\right).
\een
The central connections coefficients of such operators is another family solutions to the $QQ$-system.  The operators \eqref{L Gaudin-1} and the monster potential operators \eqref{1} do not seem to be directly related by a change of coordinates except for some special cases related to $d=0$, see Section \ref{sec:reproductiononopers}. However,  we provide  algebraic and analytic evidence 
that the $Q$-functions obtained from \eqref{L Gaudin-1} coincide with those obtained from the monster potential operators \eqref{1}.
\end{itemize}

The appearance of the monster potential operators \eqref{1} in the quantum KdV model is a mystery from the Bethe Ansatz point of view. On the other hand, operators \eqref{L Gaudin-1} are derived in a rather familiar algebraic way. We hope the connection between the two families of operators will provide a boost to the study of both ODE/IM correspondence and of the quantum KdV/ affine Gaudin models.

We describe our motivation and results in more detail. 

\bigskip

\noindent
{\bf Affine Gaudin models and quantum KdV.} The quantum KdV model is defined by a commuting family of linear operators, Hamiltonians, acting in a generic Virasoro Verma module $M_{c,\Delta}$, where $c$ is the central charge and $\Delta$ the highest weight. These Hamiltonians are given by integrals of local Virasoro currents and are called local integrals of motions, see \cite{BLZ1}. The eigenvalues of the Hamiltonians for a given common eigenvector are expected to be encoded in a pair of entire functions $\bar{Q}_\pm(\bar{\la})$ of the spectral parameter $\bar{\la}$, known as $Q$-functions of the quantum KdV model. The expansion of the $Q$-functions at $\bar{\la}=\infty$ produces eigenvalues of local integral of motions, and one conjectures that the expansions of the $Q$-functions at $\bar{\la}=0$ produce eigenvalues of the so called non-local integrals of motion, \cite{BLZ2}. The $Q$-functions $\bar{Q}_\pm(\bar{\la})$ satisfy the quantum Wronskian condition (or the $QQ$-system):
\be
\gamma \bar{Q}_+(q\bar{\la})\bar{Q}_-(q^{-1}\bar{\la})-\gamma^{-1}\bar{Q}_+(q^{-1}\bar{\la})\bar{Q}_-(q\bar{\la})=1,
\ee
where $\gamma=e^{2\pi ip}$ and $q=e^{\pi i\beta^2}$ are related to $c$ and $\Delta$ by \eqref{cdeltabetap}.

An alternative approach to the quantum KdV model is provided by affine Gaudin models, an affine analogue of Gaudin models associated to simple Lie algebras. It is expected \cite{FF,FJM} that the non-local quantum KdV integrals of motion are related to the  $\slth$ affine Gaudin model acting in the space of singular vectors in a tensor product $L_{2m,2l}\otimes L_{1,0}$ of irreducible $\slth$ Verma module of generic highest weight $(2m,2l)$ with the basic level one module. The subspace of singular vectors of a given $\slt$-weight $2r$ ($r\in\Z$) is naturally a Virasoro Verma module $M_{c,\Delta_r}$ via the Sugawara construction, see \eqref{coset}, where
$$
c=1-\frac{6}{(k+2)(k+3)}, \qquad \Delta_r=\frac{(l-2r-kr)(l+1-2r-kr)}{(k+2)(k+3)},
$$
and $k=2m+2l$ is the level of $L_{2m,2l}$. It is predicted that the quantum KdV integrals of motion commute with affine Gaudin Hamiltonians and, moreover, that non-local integrals of motion of \cite{BLZ1} are images of affine Gaudin Hamiltonians.

Despite many attempts, the affine Gaudin models are not well understood. The first non-trivial Hamiltonian is conjectured in \cite{LVY2}. The affine Gaudin model is expected to be an appropriate limit of the XXZ models associated to quantum toroidal algebras. These XXZ models are described and can be studied by the Bethe Ansatz method, \cite{FJMM0,FJMM1}. However, the limit is difficult to calculate.
 At the moment, the study of affine Gaudin models is effectively reduced to the study of algebraic structures of the Bethe Ansatz equations (BAE) and corresponding opers.   The BAE for the $\slth$-module  $L_{2m,2l}\otimes L_{1,0}$ is an algebraic system on variables $\{s_i\}_{i=1}^{d_0}$, $\{t_j\}_{j=1}^{d_1}$:
 \begin{subequations}\label{Gaudin BAE BLZ intro}
\bean
\frac{1/2}{s_i-1}+ \frac{k/2-l}{s_i}+\sum_{j=1}^{d_1}\frac{1}{s_i-t_j}-\sum_{j=1, j\neq i}^{d_0}\frac{1}{s_i-s_j}=0, \label{Gaudin BAE BLZ 1 intro}\\
\frac{l}{t_i}+\sum_{j=1}^{d_0}\frac{1}{t_i-s_j}-\sum_{j=1, j\neq i}^{d_1}\frac{1}{t_i-t_j}=0.
\label{Gaudin BAE BLZ 2 intro}
\end{align}
\end{subequations}
It is expected that eigenvalues of Hamiltonians of affine Gaudin model are values of explicit polynomials of $s_i,t_j$ symmetric in both $s_i$ and $t_j$ and independent of $m,l,d_0,d_1$.
Set
\begin{equation}\label{eq:rintro}
r=d_0-d_1.
\end{equation}
It is expected that for generic $l,k$, the number of solutions to \eqref{Gaudin BAE BLZ intro}, up to permutation of variables $t_i$ and of variables $s_j$, equals to the number of partitions of $d_0-r^2$ and the corresponding eigenvectors of the Gaudin Hamiltonians form a basis of the vectors of degree $\Delta_r-d_0+r^2$ in the Virasoro Verma module $M_{c,\Delta_r}\subset L_{2m,2l}\otimes L_{1,0}$.  We further set 
\begin{equation}\label{y0y1intro}
y_0(x)=\prod_{i=1}^{d_0}(x-s_i),\qquad y_1(x)=\prod_{i=1}^{d_1}(x-t_i),
\end{equation}
and
\begin{equation}\label{T0T1intro}
T_0(x)=x^{2m}(x-1),\qquad T_1(x)=x^{2l}.
\end{equation}
Then the BAE \eqref{Gaudin BAE BLZ  intro} says that functions $T_0y_1^2/y_0^2$ and  $T_1y_0^2/y_1^2$ have no residues at $s_i$ and $t_i$. The solutions with different $r$ can be connected to the solutions with $r=0$ via an explicit recursive procedure  known  as trigonometric reproduction procedure, \cite{MV2}, therefore, often one can concentrate on the $r=0$ sector.
\bigskip

\noindent
{\bf Affine Gaudin models and $\lambda$-opers.}
The affine opers related to the BAE were first suggested in \cite{FF}, then studied further in \cite{LVY1,LVY2}. Another version of affine opers called $\la$-opers was introduced in \cite{GLVW} motivated by the Kondo problems in products of chiral $\on{SU}(2)$ WZW models. Given polynomials \eqref{y0y1intro} whose roots solve the BAE \eqref{Gaudin BAE BLZ intro}  define the functions
\begin{gather*}
a_+(x)=-\frac 12 \ln'\frac{T_0y_1^2}{y_0^2}, \qquad  a_-(x)=-\frac 12 \ln'\frac{T_1y_0^2}{y_1^2},\\ 
 P(x)=T_0(x)T_1(x)=x^k(x-1),
 \end{gather*}
where $T_0$, $T_1$ as in \eqref{T0T1intro}, and define two operators
\begin{align*}
L_0=(\partial_x+a_+(x))(\partial_x-a_+(x))-P(x)\la^2, \\
L_1=(\partial_x+a_-(x))(\partial_x-a_-(x))-P(x)\la^2,
\end{align*}
see \eqref{L01 gen}. Then $L_1$ has form \eqref{L Gaudin-1}. We note that this construction is similar to the one described in \cite{MV1,MV2}, where given a solution  of the Gaudin BAE associated to any Kac Moody algebra, one naturally assigns a set of $\slt$-opers, one for each Dynkin node, and then uses the reproduction procedure  to produce families of solutions of the BAE.   For the present case of $\slth$ there are two nodes and the two $\slt$-opers are $(\partial_x+a_\pm(x))(\partial_x-a_\pm(x))$. In this case, we can add the same term $\la^2 P=\la^2 T_0T_1$ to the two $\slt$-opers to produce  the same $\la$-oper of \cite{GLVW}.  Thus adding $\la^2 P$, where  $\la$ is an arbitrary non-zero constant, makes those two $\slt$-opers gauge equivalent.

The operators $L_0$ and $L_1$  have poles of order two at $t_i$ and $s_i$, respectively, and the BAE \eqref{Gaudin BAE BLZ intro} imply that monodromy around these poles is trivial for all $\la$.  In particular, $L_1$ does not depend on $t_i$. We note incidentally that for an operator of the form  \eqref{L Gaudin-1} the trivial monodromy at $\lambda=0$ is equivalent to the trivial monodromy for all $\lambda$. We expect that for generic $l,k$ all operators of the form \eqref{L Gaudin-1} are obtained from solutions to BAE \eqref{Gaudin BAE BLZ intro}. Thus going to $L_1$ we effectively eliminate variables $t_i$ from the Bethe Ansatz equations. In this paper we study operators $L^G$ which are operators $L_1$ for some $\la$-oper. The $L^G$ opers are also special cases of more general operators considered in \cite{KL}.

\bigskip

\noindent
{\bf ODE/IM correspondence for $L^G$.}
At $x=0$ the operator $L^G$ has a regular singular point, and for generic $l$ the equation $L^G\Psi(x,\la)=0$ has two  Frobenius solutions $\chi^+,\chi^-$. The functions $\chi^+,\chi^-$ have the form $x^{l+1}u(x,\la)$,  and $x^{ -l}v(x,\la)$, respectively,  where $u,v$ are power series in variables $x$ and $\la^2 x^{k+2}$, see Corollary \ref{cor:frobenius}.
At $x=\infty$, the operator $L^G$ has an irregular singular point, and there is a unique Sibuya solution $\psi^{0}(x,\la)$ which vanishes exponentially fast as $x\to \infty$ in the sector $-\pi+\epsilon <2\on{arg} \la +(k+3)\on{arg} x<\pi+\epsilon$ for all $\epsilon>0$, see Proposition \ref{prop:sibuya}. Since all non-zero singularities are apparent, we can unambiguously write
\be 
\psi^{(0)}(x,\la)=A_-(\la) \chi^+(x,\la) + A_+(\la) \chi^-(x,\la).
\ee 
We define the $Q$-functions
\be
Q_\pm(\la)=\pm e^{-\frac{\pi i}{4}}(l+\tfrac12)^{\frac{1}{2}}\la^{-\frac{k+3}{4}\pm(l+r+\frac12)} A_\pm (\la^{\frac{k+3}{2}}),
\ee
where, as always,  $r=d_0-d_1$, see \eqref{eq:rintro}. We prove in Theorem  \ref{thm:QQsystem} that these functions are holomorphic for $\la\in\C^*$ and that they satisfy the $QQ$-system:
\begin{equation}\label{QQsystemintro-1}
\gamma_{r} Q_+(q\la)Q_-(q^{-1}\la)-\gamma_{r}^{-1}Q_+(q^{-1}\la)Q_-(q\la)=1,
\end{equation}
where
\be
\gamma_{r}=e^{\frac{2l-2r(k+2)+1}{k+3}\pi i}, \qquad q=e^{\frac{k+2}{k+3}\pi i}.
\ee
This is a version of the ODE/IM correspondence, see \cite{DTa, DT, BLZ4}. The  existence of the $QQ$-system \eqref{QQsystemintro-1} for $L^G$ was suggested in \cite{KL}. The limit $\la \to 0$ is a singular limit for the operator $L^G$, since the formal substitution $\la=0$ in \eqref{L Gaudin-1} leads to an operator with a Fuchsian (instead of irregular) singularity $x=\infty$. We study the limit as $\la\to 0$ of the $Q$-functions and we conjecture that if $L^G$ is obtained from a solution to the Gaudin  BAE \eqref{Gaudin BAE BLZ intro} with $d_0-d_1=r$, then the function $Q_{\pm}(\la)$ is entire, namely it extends holomorphically to $0$.
Due to reproduction procedure the general $r$ case, generically, can be deduced from $r=0$, see propositions \ref{prop:R0Qpm} and \ref{prop:R1Qpm} for the explicit action of the reproduction on $Q_{\pm}$.

We also consider lateral connection problems for the operator $L^G$, defining Stokes sectors and Sibuya solutions $
\psi^{(j)}$ ($j\in\Z$) which are subdominant in those sectors, and obtaining entire functions $T_j(\la)$ which satisfy the $TQ$-system and fusion relations, see Theorem \ref{thm:TQ}.

\bigskip

\noindent
{\bf Comparison with the BLZ monster potential oper.}
In the classical ODE/IM correspondence, one considers  Schr{\"o}dinger operators $L^{BLZ}$ of the form \eqref{1}, with $ \bar{k} \in (0,1), \; \Re \bar{l}>-\frac12,$ and where the $\bar{z}_i$ are non-zero  distinct complex numbers such that
 the monodromy about them is trivial for every $\la$, see \cite{BLZ4,MR1}. The operators $L^{BLZ}$ are expected to describe the spectrum of quantum KdV model.  More precisely,  one expects that eigenvectors of weight $\Delta-d$ of the quantum KdV integrals of motion in the Virasoro Verma module $M_{c,\Delta}$ are in a natural bijection with the operators \eqref{1} where $\bar l, \bar k$ are related to the central charge and the highest weight of the Virasoro module, see \eqref{charge and weight general}, and that the eigenvalues of each local integral of motion can be computed as a fixed polynomial evaluated on the positions of the apparent singularities $z_i$ in the corresponding operator, see \cite{BLZ4,Lit}.

Our main conjecture is that the $Q$-functions for the operators $L^G$ and $L^{BLZ}$ essentially coincide. More precisely, for each operator
\eqref{1} and each $r\in \Z$ there exists an operator \eqref{L Gaudin-1} with $d_0=d+r^2$ and 
\begin{equation} \label{eq:parametersintro}
\bar k=\frac{k+2}{k+3}, \qquad \bar{l}+\frac12=\frac{l-(k+2)r+\frac12}{k+3},  
\end{equation}
such that the $Q$-functions of the two operators coincide, see \eqref{eq:Q=barQ}.  Moreover, this correspondence is bijective: for each operator $L^G$ there exists an $r$ and an operator $L^{BLZ}$ with $d=d_0-r^2$ such that the $Q$-functions coincide.

We have the following analytic reasons for this correspondence. The functions $Q_{\pm}$ for
$L^G$ and the function $\bar{Q}_\pm$ for $L^{BLZ}$ in both cases are entire functions of $\la$ which satisfy the same  $QQ$-system. Moreover, the asymptotics of zeros of the functions $Q_{+}$  and  $\bar{Q}_+$ coincide in the limit $l\to\infty$. These properties uniquely determine the $Q$-function, see \cite{CM1}.

Note that the top operators $L^G$ with $d_0=0$ and $L^{BLZ}$ with $d=0$ are related by a simple - though $\la$ dependent - change of variables, see \eqref{eq:phichange}, and therefore have the same $Q$-functions. Moreover, the other "top" operators $L^G$ for each $r\in\Z$ have  $(d_0,d_1)=(r^2,r^2-r)$ and are related to $d_0=d_1=0$ by the reproduction procedure, see Section \ref{sec:reproductiononopers}, and, therefore, can be related to $L^{BLZ}$ with $d=0$ by a change of variables and a gauge transformation. However, we do not expect such a simple relation in the general case.

Also, note that the limit $\la \to 0$ is a singular limit for $L^G$ and a regular limit for $L^{BLZ}$. As a consequence, at $\la=0$, the $L^G$ and $L^{BLZ}$ are very different: at $\la=0$, $L^{G}$ only has Fuchsian singularities and the solutions $L^{G}\Psi=0$ are quasi-polynomials, see \eqref{kernel}, while  at $\la=0$, $L^{BLZ}$ has one irregular singularity at $\infty$ and the solutions $L^{BLZ}\Psi=0$ have an essential singularity at $\infty$. For example for the top oper $d_0=0$ the solutions $L^{BLZ}\Psi=0$ at $\la=0$
are given by confluent hypergeometric functions.
\\

\noindent
{\bf Further perspectives.}
The relation between operators \eqref{L Gaudin-1} and \eqref{1} is a manifestation of conjectural duality between the quantum KdV and affine Gaudin models. Note that we do not expect an algebraic map connecting their poles $s_i$ and $z_i$ since $s_i$ and $z_i$ describe non-local and local integrals of motions respectively, which seem to have very different nature. We also note that for affine Gaudin model we expect a standard $\gl_n\times\gl_m$ type duality which is an affine version of \cite{MTV1,MTV2,MTV3}, and $q\to 1$ limit of the quantum toroidal duality described in \cite{FJM2}. In this context, the case we consider is dual to the affine $\gl_1$ Yangian XXX model, see \cite{FJM}, known as intermediate long wave model \cite{Lit}. While the intermediate long wave model  is anticipated to reproduce local integrals of motion of quantum KdV, the corresponding Bethe Ansatz equations are of XXX type. The corresponding opers are not described yet but they are expected to be difference equations (not differential).

We expect that the picture we painted here holds in a much larger generality. The $\la$-opers can be written at least for ADE types, see \cite{GLVW}. The $QQ$-systems, also known as reproduction equations, are known for all Kac-Moody algebras, \cite{MV2}. For operators of type $L^{BLZ}$ the $Q$-functions and the corresponding $QQ$-systems have also been studied in wide generality, see \cite{MRV,MRV2,FH,ESV}. However, very little is known about zeroes of $Q$-functions and their asymptotics.
\\

\noindent
{\bf Organization of the paper.}
The paper is organized as follows. In Section \ref{BAE chapter} we discuss the Bethe Ansatz equations for the $\slth$ Gaudin model. We briefly describe the various forms of the reproduction procedure, recall the $\la$-opers and explain how they give rise to the Schr{\"o}dinger operators with apparent singularities. In Section \ref{BLZ chapter} we concentrate on the case which corresponds to the quantum KdV. We identify the parameters and give a few examples. In Section \ref{Q chapter} we define the $Q_\pm$-functions, $T$-function, and derive the  $QQ$ and $TQ$ relations. This involves the construction of Frobenius solutions at $x=0$ and of the Sibuya solutions at $x=\infty$. In Section \ref{BAE limit chapter} we consider the $l\to \infty$ limit of solutions of the BAE equations. We discuss the existence of solutions and their behaviour in this limit. In Section \ref{compare chapter} we perform the asymptotic analysis of the $Q$-functions. In particular we study $\la\to 0$ limit of $Q$-functions and the asymptotics of zeroes of $Q$-functions as $\la \to +\infty$ and as $l\to +\infty$, using the complex WKB method.  In Section \ref{comparison section} we compare the $Q$-functions of $L^G$ to the $Q$-functions of $L^{BLZ}$.

\section{Gaudin \texorpdfstring{$\slth$}{sl2} model and \texorpdfstring{$\lambda$}{lambda}-opers}\label{BAE chapter}
\subsection{Bethe Ansatz equations for the Gaudin \texorpdfstring{$\slth$}{sl2} model}\label{periodic sec}
The Lie algebra $\slth$ is generated by Chevalley generators $e_i,f_i,h_i$, $i=0,1,$ and derivation ${\rm d}$ with the standard relations. In particular, $h_0+h_1$ is central and $[{\rm d},f_0]=f_0$, $[{\rm d},e_0]=-e_0$, $[{\rm d},e_1]=[{\rm d},f_1]=0$. 
For $m,l,r\in\C$, let $L_{2m,2l;r}$ be the irreducible  highest weight $\slth$-module generated by the vector $v_{2m,2l;r}$ satisfying
\begin{align*}
&e_0v_{2m,2l;r}=e_1v_{2m,2l;r}=0, \\ &h_0v_{2m,2l;r}=2mv_{2m,2l;r}, \quad h_1v_{2m,2l;r}=2lv_{2m,2l;r},\quad  {\rm d}v_{2m,2l;r}=rv_{2m,2l;r}.
\end{align*}
In particular $L_{2m,2l;r}$ has level (or central charge or the eigenvalue of $h_0+h_1$) given by $k=2m+2l$. 
We set $L_{2m,2l}=L_{2m,2l;0}$ and 
\begin{equation}\label{module L}
\mc L=\mathop{\otimes}_{i=1}^n L_{2m_i,2l_i}.
\end{equation}
The module $\mc L$ has level $\sum_{i=1}^n(2m_i+2l_i)$.
For non-negative integers $d_0,d_1\in\Z_{\geq 0}$,  denote the subspace of $\mc L$ of  vectors of weight $\sum_{i=1}^n (2m_i,2l_i,0)-d_0\alpha_0-d_1\alpha_1$  by $\mc L_{d_0,d_1}$: 
\begin{align*}
\mc L_{d_0,d_1}= \{v\in \mc L\ |\   &h_0v=(\sum_{i=1}^n 2m_i-2d_0+2d_1)v,\\ &h_1v=(\sum_{i=1}^n 2l_i+2d_0-2d_1)v,\ {\rm d}v=d_0v\},
\end{align*}
and the subspace of singular vectors (of the same weight) by $\mc L_{d_0,d_1}^{sing}$:
\be
\mc L_{d_0,d_1}^{sing}= \{v\in \mc L_{d_0,d_1}\ |\  e_0v=e_1v=0\}.
\ee
We have $\dim \mc L_{d_0,d_1}^{sing}\leq \dim \mc L_{d_0,d_1}<\infty$. We also denote
\be
\mc L^{sing}=\mathop{\bigoplus}_{d_0,d_1=0}^\infty \mc L^{sing}_{d_0,d_1}.
\ee
We call the pair $(d_0,d_1)$ the depth pair.
Fix the evaluation parameters $z_1,\dots,z_n\in\C$,  $z_i\neq z_j$. 
The highest weights and evaluation parameters are encoded in functions $T_0,T_1$:
\begin{align}\label{T}
T_0=\prod_{i=1}^n(x-z_i)^{2m_i}, \qquad T_1=\prod_{i=1}^n(x-z_i)^{2l_i}.
\end{align}
We also set
\be
P=T_0T_1=\prod_{i=1}^n(x-z_i)^{k_i}, \qquad k_i=2m_i+2l_i.
\ee
We consider the Gaudin $\slth$ model (or affine Gaudin $\slt$ model) in   $\mc L$. While the Hamiltonians of the affine Gaudin model are not well understood, conjecturally, they commute with $\slth$ and therefore act in the finite-dimensional space $\mc L_{d_0,d_1}^{sing}$. Moreover, generically, their eigenvectors and eigenvalues are expected to be found by the method of Bethe Ansatz based on solutions of a system of algebraic equations which we now describe. For that purpose, consider a pair of polynomials in $x$, 
\begin{equation}\label{y0 y1}
y_0=\prod_{i=1}^{d_0}(x-s_i), \qquad y_1=\prod_{j=1}^{d_1}(x-t_j). 
\end{equation}
We are interested only in zeroes of $y_0,y_1$; sometimes it is convenient to multiply these monic polynomials by non-zero numbers. A pair $(y_0,y_1)$ is called generic if $s_i,t_j$ are all distinct, $T_0(s_i)\not\in\{ 0,\infty\}$, and $T_1(t_j)\not\in\{ 0,\infty\}$.
We say that a pair of polynomials $y=(y_0,y_1)$ is a generalized solution of the Gaudin Bethe Ansatz equations (BAE) if functions
$$
 T_0y_1^2/y_0^2, \qquad T_1y_0^2/y_1^2, 
 $$
have no residues at $x=s_i$ and $x=t_j$ for all $i,j$. We call a pair $y$ a solution of BAE if it is a generalized solution of BAE and generic. 
There is an equivalent formulation of BAE for a generic pair $y$, see \cite{MV1}. A generic pair $y=(y_{0},y_{1})$ is a solution of BAE if and only if zeroes of $y_{0},y_{1}$ satisfy
\begin{subequations}\label{Gaudin BAE}
\bean
\sum_{j=1}^n\frac{m_j}{s_i-z_j}+\sum_{j=1}^{d_1}\frac{1}{s_i-t_j}-\sum_{j=1, j\neq i}^{d_0}\frac{1}{s_i-s_j}=0, \label{Gaudin BAE 1}\\
\sum_{j=1}^n\frac{l_j}{t_i-z_j}+\sum_{j=1}^{d_0}\frac{1}{t_i-s_j}-\sum_{j=1, j\neq i}^{d_1}\frac{1}{t_i-t_j}=0, \label{Gaudin BAE 2}
\end{align}
\end{subequations}
for all $i,j$.
It is expected that eigenvalues of Hamiltonians of affine Gaudin model in 
 $\mc L_{d_0,d_1}$  are given as values of some explicit polynomials of $s_i,t_j$ symmetric in both $s_i$ and $t_j$ and independent of $n,m_i,l_i$.
The BAE \eqref{Gaudin BAE} are one of the central objects of our study. Note that we do not distinguish between solutions which are obtained by permutations of $s_i$ and by permutations of $t_j$.

\subsection{The exponential case}\label{exp sec}
It is expected that there is a generalization of the affine Gaudin system which is called exponential (or quasi-periodic or twisted) Gaudin model. In that terminology the case described in Section \ref{periodic sec} is called periodic Gaudin model.
The Hamiltonians of the exponential Gaudin model should depend on two additional twist parameters $n_0,n_1$. For generic values of $n_0,n_1$, the Hamiltonians are expected to commute with $h_0,h_1$ but not with $\slth$. Therefore, it is expected that Bethe Ansatz produces eigenvectors and eigenvalues in the weight space $\mc{ L}_{d_0,d_1}$.  In the exponential case we have
\be
T_0=e^{2n_0 x}\prod_{i=1}^n(x-z_i)^{2m_i}, \qquad T_1=e^{2n_1 x}\prod_{i=1}^n(x-z_i)^{2l_i}.
\ee
and then the BAE \eqref{Gaudin BAE} take the form
\begin{subequations}\label{Gaudin exp BAE}
\bean
\sum_{j=1}^n\frac{m_j}{s_i-z_j}+\sum_{j=1}^{d_1}\frac{1}{s_i-t_j}-\sum_{j=1, j\neq i}^{d_0}\frac{1}{s_i-s_j}+n_0=0, \label{Gaudin exp BAE 1}\\
\sum_{j=1}^n\frac{l_j}{t_i-z_j}+\sum_{j=1}^{d_0}\frac{1}{t_i-s_j}-\sum_{j=1, j\neq i}^{d_1}\frac{1}{t_i-t_j}+n_1 =0.\label{Gaudin exp BAE 2}
\end{align}
\end{subequations}

\subsection{Reproduction procedure and the trigonometric case}
Let
\ben\label{Wr def}
\on{Wr}\big(y,\tilde y\big)=y\tilde y'-y'\tilde y
\een
denote the Wronskian of two functions.
The reproduction procedure \cite{MV1,MV2} is an effective method, starting from a given  solutions to the Gaudin BAE, to obtain a new solution with new $d_0,d_1$. Iterating the reproduction procedure one gets a family of solutions called a population.
There are three kinds of Gaudin reproduction: reproduction, exponential reproduction, and trigonometric reproduction.  
Consider the periodic case \eqref{Gaudin BAE}. Assume that all $m_i,l_i\in\Z_{\geq 0}$. Then a generic pair $y=(y_0,y_1)$ is a solution if and only if there exist polynomials $R_0[ y_0]$, $R_1[y_1]$ such that
\begin{subequations}\label{rep}
\bean 
&\on{Wr}(y_0,R_0[ y_0])=T_0y_1^2,\label{rep 1}\\
&\on{Wr}(y_1,R_1[ y_1])=T_1y_0^2.\label{rep 2}
\end{align}
\end{subequations}
Moreover, if $(R_0[ y_0],y_1)$ or $(y_0,R_1[ y_1])$ is generic then it is also a solution of BAE 
\eqref{Gaudin BAE}, see \cite{MV1}. One says that these new solutions are obtained by reproduction
from $(y_0,y_1)$ in the $0$-th and the $1$-st direction respectively. Note that one can apply reproduction to these new solutions again.
By abuse of notation, denote $R_0[d_0]=\deg R_0[ y_0]$ and $R_1[ d_1]=\deg R_1[ y_1]$.
In the periodic case, the polynomials $R_0[y_0]$, $R_1[ y_1]$ depend on a parameter (integration constant) and the new weights 
$$(\sum_{i=1}^{n} 2m_i-2R_0[d_0]+2d_1,\sum_{i=1}^{n} 2l_i-2d_1+2R_0[d_0]),$$  
and
$$(\sum_{i=1}^{n} 2m_i-2d_0+2R_1[d_1],\sum_{i=1}^{n} 2l_i-2R_1[d_1]+2 d_0),$$ 
are obtained from the initial weight 
$$(\sum_{i=1}^{n} 2m_i-2 d_0+2d_1,\sum_{i=1}^{n} 2l_i-2d_1+2 d_0)$$ 
by shifted action of simple reflections  $s_0,s_1\in \Z\rtimes\Z/2\Z$  in the  $\slth$ Weyl group. The population  is expected to 
be isomorphic (after an appropriate closure) to the $\slth$ flag variety.

For the exponential reproduction we consider the exponential case \eqref{Gaudin exp BAE}. Assuming that all $m_i,l_i\in\Z_{\geq 0}$, then a generic pair $y=(y_0,y_0)$ is a solution if and only if there exist polynomials $R_0[y_0]$, $R_1[y_1]$ such that
\begin{subequations}\label{exp rep}
\bean 
&\on{Wr}(y_0,e^{2n_0x}R_0[y_0])=T_0y_1^2,\label{exp rep 1}\\
&\on{Wr}(y_1,e^{2n_1x}R_1[y_1])=T_1y_0^2.\label{exp rep 2}
\end{align}
\end{subequations}
Moreover, if $(R_0[y_0],y_1)$ or $(y_0,R_1[y_1])$ is generic then it is also a solution of BAE 
\eqref{Gaudin exp BAE} but with the twist parameters $(n_0,n_1)$ changed respectively to $(R_0[n_0],R_0[n_1])=(-n_0,n_1+n_0)$ and to
$(R_1[n_0],R_1[n_1])=(n_0+n_1,-n_1)$, see \cite{MV2}. In this case (provided that $n_0,n_1$ are non-zero)  $R_0[ y_0]$, $R_1[ y_1]$ are unique and the new weights are obtained by (non-shifted) action of simple reflections  $s_0,s_1\in \Z\rtimes\Z/2\Z$  in the $\slth$ Weyl group. If on each step of the reproduction the new pair of polynomials is generic, then the population is a discrete set identified with the Weyl group (or rather with a free orbit of the Weyl group).

For the trigonometric case, which produces the trigonometric affine Gaudin model, we consider again the periodic case \eqref{Gaudin BAE}. Assume that $m_i,l_i\in\Z_{\geq 0}$ for $i>1$. Note that if $m_1,l_1$ are generic then we can identify
\be
(\mc L)^{sing}=(\mathop{\otimes}_{i=1}^n L_{2m_i,2l_i})^{sing}\simeq \mathop{\otimes}_{i=2}^n L_{2m_i,2l_i},
\ee 
and consider trigonometric affine Gaudin model depending on $m_1,l_1$ as parameters and acting in the module $\mathop{\otimes}_{i=2}^n L_{2m_i,2l_i}$. This model has the same BAE equations \eqref{Gaudin BAE} but the trigonometric reproduction procedure changes parameters $m_1,l_1$ as follows.  A generic pair $y=(y_0,y_1)$ is a solution of BAE if and only if there exist polynomials $R_0[ y_0]$, $R_1[ y_1]$ such that
\begin{subequations}\label{trig rep}
\bean 
&\on{Wr}(y_0, x^{2m_1+1}R_0[ y_0])=T_0y_1^2,\label{trig rep 1}\\
&\on{Wr}(y_1,x^{2l_1+1}R_1[ y_1])=T_1y_0^2.\label{trig rep 2}
\end{align}
\end{subequations}
Moreover, if $(R_0[y_0],y_1)$ is generic then it is also a solution of BAE 
\eqref{Gaudin BAE} but with $(m_1,l_1)$ changed by the  shifted action of simple reflection  $s_0 \in \Z\rtimes\Z/2\Z$  in $\slth$ Weyl group:  
\begin{equation}\label{R_0 m l}
(R_0[m_1],R_0[l_1])=(-m_1-1,l_1+2m_1+1).
\end{equation}
Similarly if   $(y_0,R_1[y_1])$  is generic then it is also a solution of BAE
\eqref{Gaudin BAE} with $(m_1,l_1)$ changed by the shifted action of $s_1$:
\begin{equation}\label{R_1 m l}
(R_1[m_1],R_1[l_1])=(m_1+2l_1+1,-l_1-1),
\end{equation}
see \cite{MV2}. Incidentally, one notes that
$$k_1=2m_1+2l_1=2R_0[m_1]+2R_0[l_1]=2R_1[m_1]+2R_1[l_1],$$
so the central charge $k_1$ is unchanged under trigonometric reproduction. Denote $R_0[d_0]=\deg R_0[y_0]$ and $R_1[d_1]=\deg R_1[y_1]$.
In this case $R_0[y_0]$, $R_1[y_1]$ are again unique and the new weights 
are obtained from the initial weight
by (non-shifted) action of simple reflections  $s_0,s_1\in \Z\rtimes\Z/2\Z$  in $\slth$ Weyl group. If on each step of reproduction the new pair of polynomials is generic, then the population is identified with the Weyl group (or rather with a free orbit of the Weyl group).

\subsection{\texorpdfstring{$\lambda$}.-opers}
Following \cite{GLVW}, we call $\la$-oper a differential operator of the form
\be
\partial_x+A(x,\la),  \qquad A(x,\la)\in\slt,
\ee
considered up to a gauge - up to a conjugation by functions of the form  $\exp(B(x,\la)g)$, where $B(x,\la)$ is a scalar function and $g\in\slt$.\footnote{Such simplistic approach is sufficient for our purposes. For the discussion of admissible gauges for $\la$-opers, see \cite{GLVW}.}. We specify the allowed dependence of $A$ and $B$ on $x,\la$ as needed below. Given a pair of polynomials $y=(y_0,y_1)$, and a pair of functions $T_0$, $T_1$, set 
\be
a_+(x)=-\frac 12 \ln'\frac{T_0y_1^2}{y_0^2}, \qquad a_-(x)=-\frac 12 \ln'\frac{T_1y_0^2}{y_1^2}.
\ee
Note that 
\be
a_+(x)+a_-(x)=-\frac 12 \ln' (T_0T_1) =-\frac 12 \ln' P(x). 
\ee
Consider, see \cite{GLVW}, the operators
\be \label{eq:bL0bL1}
\bar L_0 =\partial_x+\begin{pmatrix} a_+(x) & P(x)\la \\ \la & -a_+(x) \end{pmatrix}, \qquad \bar L_1=\partial_x+\begin{pmatrix} -a_-(x) & \la \\ P(x)\la & a_-(x) \end{pmatrix}.
\ee
It is crucial that $\bar L_0$ and $\bar L_1$ are gauge equivalent (that is they are two representatives of the same oper):
\ben\label{L0 and L1 are equiv}
\bar L_0 = A \,\bar L_1\, A^{-1}, \qquad 
A=\begin{pmatrix} P(x)^{1/2} & 0 \\ 0 & P(x)^{-1/2} \end{pmatrix}.
\een
We consider only $\la$-opers of such form. 
If $\la\neq 0$, the differential operators $\bar L_0$ and $\bar L_1$ can be equivalently written in the scalar forms $L_0$ and $L_1$: the equations
\begin{align*}
\bar L_0 \begin{pmatrix} \psi^{(0)}_1 \\ \psi^{(0)}_2 \end{pmatrix}= \begin{pmatrix} 0 \\ 0 \end{pmatrix},\qquad 
\bar L_1 \begin{pmatrix} \psi^{(1)}_1 \\ \psi^{(1)}_2 \end{pmatrix}= \begin{pmatrix} 0 \\ 0 \end{pmatrix}
\end{align*}
hold if and only if 
\begin{subequations}\label{eq:Ltoscalar}
\begin{align}\label{eq:L1toscalar}
L_0\psi^{(0)}_2=0, \qquad \la \psi_1^{(0)}=-(\psi^{(0)}_2)'+a_+\psi_2^{(0)}, \\
\label{eq:L0toscalar}
L_1\psi^{(1)}_1=0, \qquad \la \psi_2^{(1)}=-(\psi^{(1)}_1)'+a_-\psi_1^{(1)},
\end{align}
\end{subequations}
where
\begin{subequations}\label{L01 gen}
\bean 
L_0=(\partial_x+a_+(x))(\partial_x-a_+(x))-P(x)\la^2, \label{L0 gen}\\
L_1=(\partial_x+a_-(x))(\partial_x-a_-(x))-P(x)\la^2.\label{L1 gen}
\end{align}
\end{subequations} 
The scalar operators $L_0$, $L_1$ admit the following interpretation. Let $y=(y_0,y_1)$ be a solution of the BAE \eqref{Gaudin exp BAE} corresponding to $T_0,T_1$,  and let $(y_0,\tilde{y}_1)$ be the solution obtained from $y$ by exponential reproduction \eqref{exp rep} in the $1$-st direction. Using \eqref{exp rep}, one obtains that the differential operator which has solutions $y_1$ and $e^{2n_1x}\tilde y_1$ has the form
\be
\mc D_1=(\partial - \ln' \frac{T_1y_0^2}{y_1})(\partial -\ln'y_1),
\ee
see \cite{MV2}. This operator is conjugated to $L_1$ at $\la=0$ in such a way that the coefficient of $\partial$ is zero:
\ben\label{lambda=reproduction}
L_1|_{\la=0}= T_1^{-1/2} y_0^{-1}\, \mc D_1\, T_1^{1/2} y_0.
\een
Similarly, $L_0$ at $\lambda=0$ is conjugated to 
the differential operator $\mc D_0$ with solutions $y_0$ and $e^{2n_0x}\tilde y_0$, where $(\tilde{y}_0,y_1)$ is the solution obtained from $y$ by exponential reproduction \eqref{exp rep} in the $0$-th direction. 
Thus operators $\mc D_0$  and $\mc D_1$ after conjugation and addition of the term $-P(x)\la^2$ produce operators $L_0$ and $L_1$ corresponding to the same $\lambda$-oper. In other words after such a change,  operators $\mc D_0$  and $\mc D_1$ become gauge equivalent.

Now we use the BAE to rewrite $L_1,L_0$. 

\begin{prop}\label{prop: L0 L1 from BAE} 
Let $y=(y_0,y_1)$, with $y_0,y_1$ given by \eqref{y0 y1}, be a solution of the BAE \eqref{Gaudin exp BAE}.  Then:
\begin{enumerate}[i)]
\item the operator \eqref{L1 gen} can be written as
\begin{align}\label{L1 first}
L_1= \partial_x^2-& \left(n_1^2+\sum_{i=1}^n\frac{l_i^2+l_i}{(x-z_i)^2}+\sum_{i=1}^n\frac{a_i}{x-z_i}+P(x)\la^2\right.\notag\\
&\left.+\sum_{i=1}^{d_0}\frac{2}{(x-s_i)^2}+\sum_{i=1}^{d_0}\frac{p_i}{x-s_i}\right),
\end{align}
where
\begin{align}
p_i=&\left(\ln' P\right)(s_i),\label{p formula}\\
a_i=&2n_1l_i+\sum_{j=1, j\neq i}^n\frac{2l_il_j}{z_i-z_j}-2l_i\left(\ln'(y_1)-\ln'(y_0)\right)(z_i).
\label{a formula}
\end{align}
In particular, $L_1$ is regular at $x=t_j$  and the monodromy  around $x=s_i$ is trivial for all $\la$.
\item the operator \eqref{L0 gen} can be written as
\begin{align}\label{L0 first}
L_0= \partial_x^2-&\left(n_0^2+\sum_{i=1}^n\frac{m_i^2+m_i}{(x-z_i)^2}+\sum_{i=1}^n\frac{b_i}{x-z_i}+P(x)\la^2\right.\notag\\
&\left.+\sum_{j=1}^{d_1}\frac{2}{(x-t_j)^2}+\sum_{j=1}^{d_1}\frac{q_j}{x-t_j}\right), 
\end{align}
where 
\begin{align*}
q_j&=(\ln'P)(t_j), \\
b_i&=2n_0m_i+\sum_{j=1, j\neq i}^n\frac{2m_im_j}{z_i-z_j}-2m_i\left(\ln'(y_0)-\ln'(y_1)\right)(z_i).
\end{align*}
In particular, $L_0$ is regular at $x=s_i$  and the monodromy  around $x=t_j$ is trivial for all $\la$.
\end{enumerate}
\end{prop}
 \begin{proof}
We have $L_1-\partial^2_x+P(x)\la^2=-a_-^2-a_-'$, and
\begin{align*}
a_-^2+a_-'=&(\ln'T_1)^2/4-(\ln''T_1)/2+(\ln'y_0)^2+(\ln'y_1)^2+\ln'T_1\ln 'y_0\\
&-\ln'T_1\ln'y_1-2\ln'y_0\ln'y_1-\ln''y_0+\ln''y_1\\
=&n_1^2+\sum_{i=1}^n\frac{l_i^2+l_i}{(x-z_i)^2}+\sum_{i=1}^{d_0}\frac{2}{(x-s_i)^2}+\sum_{i=1}^n\frac{a_i}{x-z_i}+\sum_{i=1}^{d_0}\frac{p_i}{x-s_i}.
\end{align*}
Note that $1/(x-t_j)^2+(1/(x-t_j))'=0$ and 
\be
\res_{x=t_j} (a_-^2+a_-')=\left(\ln'y_1'-\ln'(T_1)-2\ln'y_0\right)(t_j)=0,
\ee
due to BAE \eqref{Gaudin exp BAE 1}  for $t_j$. In particular, $L_1$ is regular at $x=t_j$. By a similar computation, one can show that $L_0$ is regular at $x=s_i$, and moreover, that a generic pair $y$ is a solution of BAE if and only if $L_1$ is regular at $x=t_j$ and $L_0$ is regular at $x=s_i$. 
The coefficients $p_i$ and $a_i$ can be computed as follows:
\begin{align*}
p_i=\res_{x=s_i} (a_-^2+a_-')=\left(\ln'(T_1)+\ln'y_0'-2\ln'y_1\right)(s_i)\\ =\left(\ln'(T_1)+\ln'(T_0)\right)(s_i)=\left(\ln' P\right)(s_i),
\end{align*}
where we used the BAE equations \eqref{Gaudin exp BAE 2} for $s_i$, and
\be
a_i=\res_{x=z_i} (a_-^2+a_-')=2n_1l_i+\sum_{j=1, j\neq i}^n\frac{2l_il_j}{z_i-z_j}-2l_i\left(\ln'(y_1)-\ln'(y_0)\right)(z_i).
\ee 
To compute the monodromy of  $L_1$ at $z=s_i$ we recall that the operators $L_0$ and $L_1$ are gauge equivalent by \eqref{L0 and L1 are equiv}. Since $L_0$ is regular around $x=s_i$, it follows that the monodromy of $L_1$ around $x=s_i$ is trivial for all $\la$. 

The coefficients $q_i$, $b_i$ and the monodromy of $L_0$ at $z=t_i$ can be computed in a similar way, exchanging $0 \leftrightarrow 1$, $l_i \leftrightarrow m_i $ in \eqref{L1 first}.
\end{proof}
The operator \eqref{L1 first}  $L_1$ with $n_1=0$ appeared already in \cite[Sec. 6.9]{FF}.
We consider in more detail the following special case of \eqref{L1 first}. If 
\begin{equation}\label{T1 special}
T_1=e^{2n_1x} x^{2l}, 
\end{equation}
or equivalently if $z_1=0$, and $l_2=l_3=\dots=l_n=0$, then the BAE equations \eqref{Gaudin exp BAE} become
\begin{subequations}\label{Gaudin exp BAE special case}
\bean
\sum_{j=1}^n\frac{m_j}{s_i-z_j}+\sum_{j=1}^{d_1}\frac{1}{s_i-t_j}-\sum_{j=1, j\neq i}^{d_0}\frac{1}{s_i-s_j}+n_0=0, \label{Gaudin exp BAE special case 1}\\
\frac{l}{t_1}+\sum_{j=1}^{d_0}\frac{1}{t_i-s_j}-\sum_{j=1, j\neq i}^{d_1}\frac{1}{t_i-t_j}+n_1 =0,\label{Gaudin exp BAE special case 2}
\end{align}
\end{subequations}
and the operator $L_1$ can be written as follows.

\begin{cor}\label{cor: L1 from BAE special case}
If $T_1$ is given by \eqref{T1 special} and $y=(y_0,y_1)$ are a solution of the BAE \eqref{Gaudin exp BAE special case},  then \eqref{L1 first} reads
\begin{align}\label{L1 general}
L_1=\partial_x^2-&\left(n_1^2+\frac{l(l+1)}{x^2}+\frac{2n_1(l+r)}{x}+P(x)\la^2\right.\notag\\
&\left.+\sum_{i=1}^{d_0}\frac{2}{(x-s_i)^2}+\sum_{i=1}^{d_0}\frac{s_ip_i}{x(x-s_i)}\right), 
\end{align}
where $p_i$ is given by \eqref{p formula} and $r=d_0-d_1$. In particular, \eqref{L1 general} does not depend on the variables $t_i$.
\end{cor}
\begin{proof}
Adding all BAE equations \eqref{Gaudin exp BAE} we obtain
\be
2l\left(\ln'(y_1)-\ln'(y_0)\right)(0)=\sum_{i=1}^{d_0}\left(\ln' P \right)(s_i)-2n_1(d_0-d_1)=\sum_{i=1}^{d_0}p_i-2n_1r.
\ee
Substituting in \eqref{a formula}, we obtain the corollary.
\end{proof}
Thus, in this case, $L_1$ is written in terms of $s_i$ and parameters of the system only, while the dependence on $t_i$ disappears. To restore $t_j$ from $L_1$, one should find the polynomial $y_1$ in
$\ker \mc D_1$, where $\mc D_1$ is related to $L_1$  by \eqref{lambda=reproduction}.
Unless $n_1=0$ and $l\in\Z_{\geq 0}$, $y_1$ is the only monic polynomial in  $\ker(\mc D_1)$.
In two cases the full kernel  of $L_1|_{\la=0}$ is described using reproduction as follows.
If $n_1=0$, then from \eqref{lambda=reproduction} we get
$$
\ker(L_1|_{\la=0})=\langle y_1,x^{2l+1}\tilde y_1 \rangle T_1^{-1/2} y_0^{-1}, 
$$
where $\tilde y_1$ is the polynomial obtained by the trigonometric reproduction procedure \eqref{trig rep 2} from $(y_0,y_1)$ in the $1$st direction.
Such a polynomial $\tilde y_1$ exists due to BAE and, $L_1$ does not change if 
one replaces $(y_0,y_1)$ with $(y_0,x^{2l+1}\tilde y_1)$, see \cite{MV2}. 
If  $n_1\neq 0$  and $2l\in\Z_{\geq 0}$, then 
\begin{align}\label{kernel}
\ker(L_1|_{\la=0})=\langle y_1,e^{2n_1x} \tilde y_1 \rangle T_1^{-1/2} y_0^{-1}, 
\end{align}
where $\tilde y_1$ is the polynomial
obtained by the exponential reproduction procedure \eqref{exp rep 2} from $(y_0,y_1)$ in the $1$st direction.

\subsection{Trivial monodromy equations}
By Proposition \ref{prop: L0 L1 from BAE} if the operator of the form \eqref{L1 first} (resp. \eqref{L0 first}) is obtained from a solution to the BAE \eqref{Gaudin exp BAE}, then it has trivial monodromy at $z=s_i$ (resp. $z=t_i$) for each $\lambda$. This in particular applies to the special case of \eqref{L1 general}. We are also interested in the monodromy conditions, without assuming BAE, for a generic differential operator of the form $L=\partial_x^2+a(x)-P(x)\la^2,$ where $a(x)$ is a meromorphic function. We consider the monodromy  in two cases.

\begin{lem}
\begin{enumerate}[i)]
\item An operator of the form
\be
L=\partial_x^2-\frac{3/4}{(x-s)^2}+ \frac{a_{-1}}{x-s}+a_0+O(x-s)
\ee
has trivial monodromy  at $x=s$  if and only if
\ben\label{monodromy eq 3/4}
a_{-1}^2+a_0=0.
\een
\item An operator of the form 
\ben\label{-2 case}
L=\partial_x^2-\frac{2}{(x-s)^2}+ \frac{a_{-1}}{x-s}+a_0+a_1(x-s)+O(x-s)^2
\een
has trivial monodromy  at $x=s$ if and only if
\ben\label{monodromy eq}
a_{-1}^3+4a_0a_{-1}+4a_1=0.
\een
\end{enumerate}
\end{lem}
\begin{proof}
\begin{enumerate}[i)]
\item The indicial equation around $x=s$ has roots $3/2,-1/2$. Thus, we have a solution of $L\psi=0$ holomorphic in the vicinity of $x=s$ given by a series 
$\psi=(x-s)^{-1/2}(1+\sum_{i=1}^\infty b_i(x-s)^{i})$, where coefficients $b_i$ are computed recursively. The second solution in the form of a series 
$\psi=(x-s)^{-1/2}(1+\sum_{i=1}^\infty b_i(x-s)^i)$ exists if an only if \eqref{monodromy eq 3/4} holds.
This is the equation equivalent to the triviality of the monodromy around $x=s$ in this case.
\item The indicial equation around $x=s$ has roots $2,-1$. Thus, we have a solution of $L\psi=0$ holomorphic in the vicinity of $x=s$ given by a series 
$\psi=(x-s)^2(1+\sum_{i=1}^\infty b_i(x-s)^i)$, where coefficients $b_i$ are computed recursively. The second solution in the form of a series 
$\psi=(x-s)^{-1}(1+\sum_{i=1}^\infty b_i(x-s)^{i})$ exists if an only if \eqref{monodromy eq} holds. This is the equation equivalent to the triviality of the monodromy around $x=s$ in the case \eqref{-2 case}.
\end{enumerate}
\end{proof}
We apply the second part of the lemma to the operator 
\begin{align}\label{scalar oper}
L_1=\partial_x^2-&\left(n_1^2+\frac{l(l+1)}{x^2}+\frac{2n_1(l+r)}{x}+P(x)\la^2\right.\notag\\
&\left.+\sum_{i=1}^{d_0}\frac{2}{(x-s_i)^2}+\sum_{i=1}^{d_0}\frac{s_ip_i}{x(x-s_i)}\right), 
\end{align}
where $r\in\Z$ and all $s_i$ are distinct. This operator is of the form \eqref{L1 general} but we do not assume that it comes from the BAE. In this case, we have
\begin{align*}
a_{-1}=&-p_i,  \\
a_0=&-n_1^2-\frac{l(l+1)}{s_i^2}-\frac{2n_1(l+r)-p_i}{s_i}\\
&-\sum_{j\neq i}\frac{2}{(s_i-s_j)^2}-\sum_{j\neq i}\frac{s_jp_j}{s_i(s_i-s_j)}-P(s_i)\la^2, \\
a_1=&\frac{\partial}{\partial s_i}a_0=\frac{2l(l+1)}{s_i^3}+\frac{2n_1(l+r)-p_i}{s_i^2}\\
&+\sum_{j\neq i}\frac{4}{(s_i-s_j)^3}-\sum_{j\neq i}\frac{s_jp_j(2s_i-s_j)}{s_i^2(s_i-s_j)^2}-P'(s_i)\la^2, 
\end{align*}
Then the operator $L_1$ has trivial monodromy at $x=s_i$ for all $\la$ if and only if
\ben\label{residue at si}
p_i=\ln'(P) (s_i)
\een
(the $\la^2$ coefficient in equation \eqref{monodromy eq}), cf. \eqref{p formula} and
\ben\label{gen mon eq}
f(s_i)+\sum_{j=1, j\neq i}^{d_0}g(s_i,s_j)=0
\een
 (the constant coefficient (with respect to $\la$)  in equation \eqref{monodromy eq}), where
\begin{align}
f(s)=&-\frac{p^3}{4}+n_1^2p+\frac{2n_1(l+r)-p^2}{s}\notag\\
&+\frac{p(l(l+1)-1)+2n_1(l+r)}{s^2}+\frac{2l(l+1)}{s^3}\,,\label{f gen}
\\
g(s,\tilde s)=&\frac{2p}{(s-\tilde s)^2}+\frac{4}{(s-\tilde s)^3}+\frac{\tilde s p\tilde p}{s(s-\tilde s)}+\frac{\tilde s\tilde p(2s-\tilde s)}{s^2(s-\tilde s)^2}\,, \label{g gen}
\\
 p=&\frac{P'(s)}{P(s)}, \qquad \tilde p=\frac{P'(\tilde s)}{P(\tilde s)}. \notag
\end{align}
From the above computations and by Corollary \ref{cor: L1 from BAE special case} we obtain that the BAE \eqref{Gaudin exp BAE special case} with  $d_0,d_1$ such that $d_0-d_1=r$  imply the trivial monodromy equations \eqref{residue at si}  and \eqref{gen mon eq}.  

We conjecture that the converse is also true.
\begin{conj}\label{gaudin - no monodromy conj}
For generic $k, l$, all operators of the form \eqref{scalar oper} with distinct non-zero $s_i$, $r\in\Z$, and with trivial monodromy around $x=s_i$ for all $i$ are obtained from the solutions of the exponential Gaudin BAE \eqref{Gaudin exp BAE special case} with $d_0-d_1=r$.   \qed 
\end{conj}

We will see an example in Section \ref{Depth one subsubsection}. 

\section{The BLZ case}\label{BLZ chapter}
\subsection{Identification of parameters} We consider a special case of \eqref{module L}, by fixing $n=2$ and setting $m_1=m$, $l_1=l$, $m_2=1/2$, $l_2=0$. Moreover, we choose evaluations parameters $z_1=0$ and $z_2=1$. This case will be the main focus of our study and we call it the BLZ case.

According to \cite{FJM}, the periodic affine $\slth$ Gaudin model with
\ben\label{BLZ module}
\mc L=L_{2m,2l}\otimes L_{1,0}, \qquad z_1=0,\ z_2=1,
\een
and with $r=d_0-d_1=0$ is expected to be dual to the quantum KdV flows described in \cite{BLZ4}. The case of $r=d_0-d_1\neq 0$ is not discussed in \cite{FJM} but we expect that the duality extends to all $r$.

The module $L_{1,0}$ is called the basic or vacuum $\slth$-module. It is the simplest non-trivial integrable module and it is well understood. A picture of $L_{1,0}$ is given on Figure \ref{basic module pic}. The module  $L_{1,0}$ has a basis represented on the picture by circles. The weights of these vectors form an orbit of the affine Weyl group $\Z\rtimes\Z/2\Z$. The vertical dashed sectors correspond to the decomposition of the module with respect to affine Cartan subalgebra of $\slth$, that is by the Heisenberg algebra generated by the loop generators $h\otimes t^s$, $s\in\Z$, $h\in\slt$. Each dashed sector is an irreducible Heisenberg Fock module. In particular, the dimension of the space of vectors in $L_{1,0}$ of degree $d$ and weight $2r$ equals the number of partitions of $d-r^2$.
The filled circles with weight $2r$ and degree $r^2$, $r\in\Z$, represent the vectors which are annihilated by $h\otimes t^s$ with $s>0$. 
\begin{figure}
        \centering
       {
\begin{tikzpicture}[scale=1.0]
\draw[->] (-11,10)--(-11,3) node[right]{degree};
\draw (-11.1,9)--(-10.9,9) node[right]{$0$};
\draw (-11.1,8)--(-10.9,8) node[right]{$1$};
\draw (-11.1,7)--(-10.9,7) node[right]{$2$};
\draw (-11.1,6)--(-10.9,6) node[right]{$3$};
\draw (-11.1,5)--(-10.9,5) node[right]{$4$};
\draw (-11.1,4)--(-10.9,4) node[right]{$5$};

\draw[->] (-11,10)--(1,10) node[below]{weight};
\draw (-9.5,10.1)--(-9.5,9.9) node[below]{$-6$};
\draw (-8,10.1)--(-8,9.9) node[below]{$-4$};
\draw (-6.5,10.1)--(-6.5,9.9) node[below]{$-2$};
\draw (-5,10.1)--(-5,9.9) node[below]{$0$};
\draw (-3.5,10.1)--(-3.5,9.9) node[below]{$2$};
\draw (-2,10.1)--(-2,9.9) node[below]{$4$};
\draw (-0.5,10.1)--(-0.5,9.9) node[below]{$6$};

\draw[fill] (-5,9) circle [radius=0.07];
\draw (-5,8) circle [radius=0.07];
\draw (-5.1,7) circle [radius=0.07];
\draw (-4.9,7) circle [radius=0.07];

\draw (-5.2,6) circle [radius=0.07];
\draw (-4.8,6) circle [radius=0.07];
\draw (-5,6) circle [radius=0.07];
\draw (-5.1,5.1) circle [radius=0.07];
\draw (-4.9,5.1) circle [radius=0.07];
\draw (-5.2,4.9) circle [radius=0.07];
\draw (-4.8,4.9) circle [radius=0.07];
\draw (-5,4.9) circle [radius=0.07];
\node at (-5,4) {$\vdots$};

\draw[dashed] (-4.9,8.9) -- (-4.4,3);
\draw[dashed] (-5.1,8.9) -- (-5.6,3);
\draw[dashed] (-3.4,7.9) -- (-3.0,3);
\draw[dashed] (-3.6,7.9) -- (-4.0,3);
\draw[dashed] (-6.4,7.9) -- (-6.0,3);
\draw[dashed] (-6.6,7.9) -- (-7.0,3);
\draw[dashed] (-8.1,4.9) -- (-8.3,3);
\draw[dashed] (-7.9,4.9) -- (-7.7,3);
\draw[dashed] (-2.1,4.9) -- (-2.3,3);
\draw[dashed] (-1.9,4.9) -- (-1.7,3);
\node at (-6.5,4) {$\vdots$};

\draw[fill] (-3.5,8) circle [radius=0.07];
\draw (-3.5,7) circle [radius=0.07];
\draw (-3.6,6) circle [radius=0.07];
\draw (-3.4,6) circle [radius=0.07];
\draw (-3.7,5) circle [radius=0.07];
\draw (-3.5,5) circle [radius=0.07];
\draw (-3.3,5) circle [radius=0.07];
\node at (-3.5,4) {$\vdots$};

\draw[fill] (-6.5,8) circle [radius=0.07];
\draw(-6.5,7) circle [radius=0.07];
\draw(-6.6,6) circle [radius=0.07];
\draw (-6.4,6) circle [radius=0.07];
\draw (-6.7,5) circle [radius=0.07];
\draw (-6.5,5) circle [radius=0.07];
\draw (-6.3,5) circle [radius=0.07];

\draw[fill] (-8,5) circle [radius=0.07];
\draw[fill] (-2,5) circle [radius=0.07];
\draw (-8,4) circle [radius=0.07];
\draw (-2,4) circle [radius=0.07];
\node at (-8,3.3) {$\vdots$};
\node at (-2,3.3) {$\vdots$};

\end{tikzpicture}
\caption{The module $L_{1,0}$.}\label{basic module pic}
}
\end{figure}
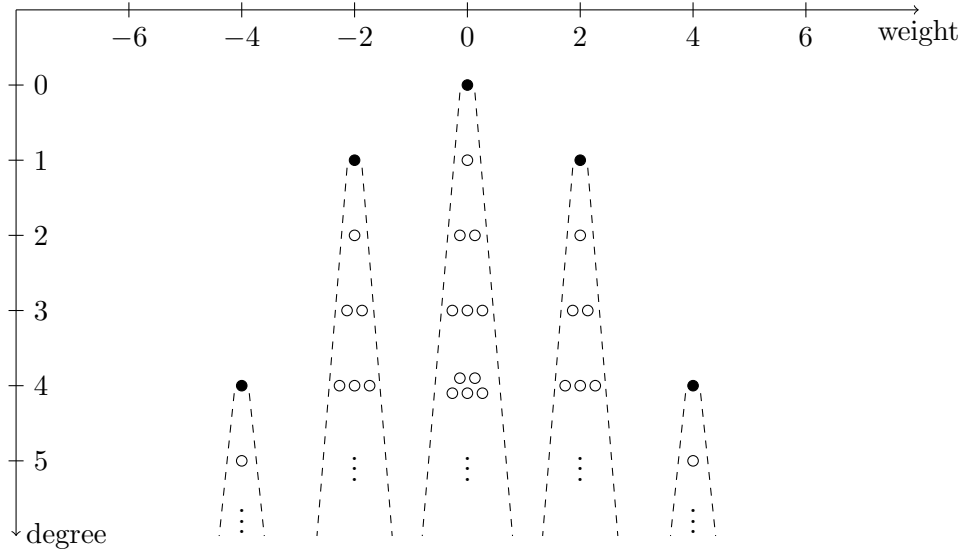
For generic $m,l$, we have  
\begin{equation}\label{Lsing = L10}
\mc L^{sing}\sim L_{1,0}.
\end{equation}
In particular,
\be
\dim \mc L_{d_0,d_0+r}^{sing}=|\{\text{Partitions of } (d_0-r^2)\}|.
\ee
Moreover, as in the GKO coset construction \cite{GKO}, for generic $m,l$ one decomposes the tensor product as
\ben\label{coset}
\mc L=L_{2m,2l}\otimes L_{1,0}=\mathop{\oplus}_{r\in\Z} L_{2m+1-2r,2l+2r;r^2}\otimes M_{c,\Delta_r},
\een
where the multiplicity space $M_{c,\Delta_r}$ is identified with the irreducible Virasoro Verma module of level $c$ and highest weight $\Delta_r$, with:
\ben\label{charge and weight general}
c=1-\frac{6}{(k+2)(k+3)}, \qquad \Delta_r=\frac{(l-2r-kr)(l+1-2r-kr)}{(k+2)(k+3)},
\een
and where, as above, $k=2m+2l$. Here, the Virasoro algebra acts on the multiplicity space by the Sugawara construction \cite{GKO}. On Figure \ref{basic module pic}, the module $M_{c,\Delta_r}$ corresponds to Fock module formed by weight vectors of weight $2r$.
It is expected that the solutions of the Bethe Ansatz, see \eqref{Gaudin BAE BLZ} below,  with 
$$
d_0-d_1=r
$$
correspond to a basis of eigenvectors of the affine Gaudin Hamiltonians in 
the multiplicity space $M_{c,\Delta_r}$. In particular, we should have no solutions unless
\be
d_0\geq r^2
\ee
and a unique solution with $d_0=r^2$, $d_1=r^2-r$ corresponding to the highest weight vector of $M_{c,\Delta_r}$.
For $r=0$, 
\ben\label{charge and weight}
c=1-\frac{6}{(k+2)(k+3)}, \qquad \Delta=\Delta_0=\frac{l(l+1)}{(k+2)(k+3)}.
\een
According to the general rule described in the previous section, for $\mc L$  given by \eqref{BLZ module} we have
\ben\label{BLZ T}
T_0(x)=x^{2m}(x-1), \quad T_1(x)=x^{2l}, \quad P(x)=x^k(x-1),  \quad k=2m+2l.
\een
The Gaudin BAE equations \eqref{Gaudin BAE} read 
\begin{subequations}\label{Gaudin BAE BLZ}
\bean
\frac{1/2}{s_i-1}+ \frac{m}{s_i}+\sum_{j=1}^{d_1}\frac{1}{s_i-t_j}-\sum_{j=1, j\neq i}^{d_0}\frac{1}{s_i-s_j}=0, \label{Gaudin BAE BLZ 1}\\
\frac{l}{t_i}+\sum_{j=1}^{d_0}\frac{1}{t_i-s_j}-\sum_{j=1, j\neq i}^{d_1}\frac{1}{t_i-t_j}=0,
\label{Gaudin BAE BLZ 2}
\end{align}
\end{subequations}
and they are a particular case of \eqref{Gaudin exp BAE special case}. The following lemma will be useful later.
\begin{lem}\label{at inf conjecture}
For the solutions of \eqref{Gaudin BAE BLZ} we have
\bean\label{at inf eq}
\sum_{i=1}^{d_0} \frac{s_i}{s_i-1}=-d_0(k+2) +r(r+1)+2lr,
\end{align}
where, as before, $r=d_0-d_1$. 
\end{lem}
\begin{proof}
Add equations
  \eqref{Gaudin BAE BLZ 1}, multiplied by $s_i$, for all $i$ and equations \eqref{Gaudin BAE BLZ 2}, multiplied  by $t_j$, for all $j$. The result is \eqref{at inf eq}.
\end{proof}

The scalar differential monodromy-free operator $L_1$ corresponding to the solutions to the BAE \eqref{Gaudin BAE BLZ} is obtained by setting $n_1=0$ and $P(x)=x^k(x-1)$ in \eqref{L1 general}. Denoting in this case $L_1$ by $L^G$ we thus obtain
\bean\label{L1 BLZ}
L^G=&L_1=\\
&\partial_x^2-\left(\frac{l(l+1)}{x^2}+\sum_{i=1}^{d_0}\frac{2}{(x-s_i)^2}+\sum_{i=1}^{d_0}\frac{k+s_i/(s_i-1)}{x(x-s_i)}+x^k(x-1)\la^2\right).\notag
\end{align}
We will consider the operator $L_0$ in Section \ref{L0 sec}.  The equation for trivial monodromy for operator \eqref{L1 BLZ} at $x=s_i$, see \eqref{gen mon eq}, takes the form
\bean\label{la monodromy eq}
f(s_i)+\sum_{j=1, j\neq i}^{d_0}g(s_i,s_j)=0
\end{align}
 with functions $f(s)$, $g(s,\tilde s)$, \eqref{f gen}, \eqref{g gen} explicitly given by:
\begin{align*}
f(s)= ((k+3)s-(k+2))\,&((k-2l+1)s-(k-2l))\,\times\\&\times\frac{((k+2l+3)s-(k+2l+2))}{4s^3(s-1)^3},
\end{align*}
\begin{align*}
&g(s,\tilde s)=\frac{1}{s^2(s-\tilde s)^3(s-1)(\tilde s-1)}\times \notag \\
&\times \big((k+2)^2s^2-(k+2)(k+3)s^3-k(2k+5)s\tilde s+(k^2+2k-4)s^2\tilde s +\\
&\hspace{30pt} +(k+3)^2s^3\tilde s+k(k+1)\tilde s^2+(k^2+5k+3)s\tilde s^2-\\
&\hspace{30pt} -(2k+7)(k+1)s^2\tilde s^2-(k+1)^2\tilde s^3+(k+1)(k+2)s\tilde s^3\big).
\end{align*}
According to the general Conjecture \ref{gaudin - no monodromy conj} for any $d_0$ and generic $l$, $k$, we expect that the number of operators \eqref{L1 BLZ} with trivial monodromy at each $z=s_i$ for every value of $\lambda$ is finite and equal to the number vectors of degree $d_0$ in $L_{1,0}$, or, equivalently, the number of circles of degree $d=d_0$ in Figure \ref{basic module pic}. Thus, we expect that the number of such operators is 
\be
\sum_{r\in \Z,\ r^2\leq d_0} \on{Part}(d-r^2),
\ee 
where $\on{Part}(s)$ is the number of partitions of $s$. For $d_0=0$ this claim is trivial: in this case there is only one operator \eqref{L1 BLZ} with corresponds to the highest weight vector of $M_{c,\Delta_0}$ in \eqref{coset} and $r=0$. For $d_0=1$ the situation is already nontrivial, we discuss it in Section \ref{Depth one subsubsection}.

\subsection{Action of the reproduction on opers}\label{sec:reproductiononopers}
We now show how the trigonometric reproduction (\ref{trig rep}) acts on the oper \eqref{L1 BLZ} as the conjugation by a gauge transformation.  Assume that $(y_0,y_1)$ is a generic solution to the Gaudin BAE \eqref{Gaudin BAE BLZ} for fixed parameters, $l$ and $m$. In what follows, we omit the dependence on $k=2l+2m$ (since it is fixed by the reproduction) while we explicitly express the dependence on $l,m$.  In particular, we write
$$T_0(x;m)=x^{2m}(x-1),\qquad T_1(x;l)=x^{2l},$$
and
$$a_-(x;l,y_0,y_1)=-\frac12\ln'\left(\frac{T_1(x;l)y_0^2}{y_1^2}\right).$$
Then the oper $\bar{L}_1$ associated to $(y_0,y_1)$, see \eqref{eq:bL0bL1}, is
\begin{equation}\label{bar L reproduction}
\bar{L}_1(x,\la;l,y_0,y_1)=\partial_x+
\begin{pmatrix}
-a_-(x;l,y_0,y_1) & \la\\
P(x)\la & a_-(x;l,y_0,y_1)
\end{pmatrix}
\end{equation}
Similarly, we denote in this subsection
the corresponding scalar operator \eqref{L1 BLZ} by
\begin{equation}\label{scalar reproduction}
L_1(x,\la;l,y_0)=\partial_x^2-V(x,\la;l,y_0)
\end{equation}
Note that \eqref{scalar reproduction} does not depend on $y_1$, as proved in Corollary \ref{cor: L1 from BAE special case}.\\

Let us first consider the trigonometric reproduction in the $0$-th direction. Since $(y_0,y_1)$ is generic, there exists a polynomial  $R_0[y_0]$ satisfying \eqref{trig rep 1}, that is:
$$\on{Wr}(y_0,x^{2m+1}R_0[y_0])=T_0y_1^2.$$
Assume that $(R_0[y_0],y_1)$ is generic, so that it is a solution to the Gaudin BAE with parameters
$$R_0[m]=-m-1,\qquad R_0[l]=k-l+1,$$
see \eqref{R_0 m l}. Note that we have
$$R_0[d_0]=\deg(R_0[y_0])=2d_1-d_0+1,$$
so that if $d_0-d_1=r$ then
$$R_0[r]=R_0[d_0]-d_1=-r+1.$$
Define the action of $R_0$ on $\bar{L}_1$, and $L_1$ repectively by
\begin{equation}\label{eq:R0bL1}
R_0[\bar{L}_1(x,\la;l,y_0,y_1)]= \bar{L}_1(x,\la;R_0[l],R_0[y_0],y_1),
\end{equation}
and
\begin{equation}\label{eq:R0L1}
R_0[L_1(x,\la;l,y_0)]=\partial_x^2-V(x,\la;R_0[l],R_0[y_0]).
\end{equation}
It is known that $R_0[\bar{L}_1(x,\la;l,y_0,y_1)]$ can also be obtained by the action of a gauge transformation on $\bar{L}_1(x,\la;l,y_0,y_1)$. Namely, we have the following lemma.
\begin{lem}\label{lem:rep0Gauge}
Let
\begin{equation}\label{eq:H0rep}
 H_0(x;l,y_0,y_1)=\frac12\ln'\left(\frac{T_1(x;R_0[l])(R_0[y_0])^2}{T_1(x;l)y_0^2}\right),
 \end{equation}
and
\begin{equation}\label{eq:B0}
B_0=
\begin{pmatrix} 1& \frac{H_0}{P\lambda}\\ 0 &1 \end{pmatrix}.
\end{equation}
We have that
\begin{equation}\label{R0bL1}
B_0\bar{L}_1(x,\la;l,y_0,y_1)B_0^{-1}=\bar{L}_1(x,\la;R_0[l],R_0[y_0],y_1).
\end{equation}

\begin{proof}
 We have the identity
\begin{equation*}
a_-(x;l,y_0,y_1)-H_0(x;l,y_0,y_1)=a_-(x;R_0[l],R_0[y_0],y_1).
\end{equation*}
By \eqref{trig rep 1}  we also have
\begin{equation}\label{H0 eq}
\ln'\Big(\frac{H_0(x;l,y_0,y_1)}{P}\Big)-2a_-(x;l,y_0,y_1)+H_0(x;l,y_0,y_1)=0.
\end{equation}
The lemma follows.
\end{proof}

\end{lem}

We now consider the reproduction in the $1$-st direction. In this case, there exists a polynomial  $R_1[y_1]$ satisfying \eqref{trig rep 2}, that is:
$$\on{Wr}(y_1,x^{2l+1}R_1[y_1])=T_1y_0^2.$$
Assuming that $(y_0,R_1[y_1])$ is generic, so that it is a solution to the Gaudin BAE with parameters
$$R_1[m]=k-m-1,\qquad R_1[l]=-l-1,$$
see \eqref{R_1 m l}.
 Note that we have
$$R_1[d_1]=\deg(R_1[y_1])=2d_0-d_1,$$
so that if $d_0-d_1=r$ then
$$R_1[r]=d_0-R_1[d_1]=-r.$$
Define the action of $R_1$ on $\bar{L}_1$ and $L_1$ respectively by
\begin{equation}\label{eq:R1bL1}
R_1[\bar{L}_1(x,\la;l,y_0,y_1)]= \bar{L}_1(x,\la;R_1[l],y_0,R_1[y_1]),
\end{equation}
and
\begin{equation}\label{eq:R1L1}
R_1[L_1(x,\la;l,y_0)]=\partial_x^2-V(x,\la;R_1[l],y_0).
\end{equation}
Notice that $L_1$ is invariant under $R_1$ since, by \eqref{L1 BLZ},
the potential $V(x,\la;l,y_0)$
is invariant under the transformation $l \to R_0[l]=-l-1$.
It is known that $R_1[\bar{L}_1(x,\la;l,y_0,y_1)]$ can also be obtained by the action of a gauge transformation on $\bar{L}_1(x,\la;l,y_0,y_1)$.
\begin{lem}\label{lem:rep1Gauge}
Let
\begin{align}\label{eq:H1rep}
& H_1(x;l,y_0,y_1)=\frac12\ln'\left(\frac{T_ 1(x;R_1[l])(R_1[y_1])^2}{T_1(x;l)y_1^2}\right),\\
& \label{eq:B1}
B_1=
\begin{pmatrix}
1& 0\\
H_1\lambda^{-1} &1
\end{pmatrix}.
\end{align}
Then
\begin{equation}\label{R1bL1}
B_1\bar{L}_1(x,\la;l,y_0,y_1)B_1^{-1}=\bar{L}_1(x,\la;R_0[l],R_0[y_0],y_1).
\end{equation}

\begin{proof}
We have the identity
\begin{equation*}
a_-(x;l,y_0,y_1)+H_1(x;l,y_0,y_1)=a_-(x;R_1[l],y_0,R_1[y_1]).
\end{equation*}

By \eqref{trig rep 2}  we also have
\begin{equation*}
\ln'(H_1(x;l,y_0,y_1))+2a_-(x;l,y_0,y_1)+H_1(x;l,y_0,y_1)=0.
\end{equation*}
The lemma follows.
\end{proof}
\end{lem}

\begin{exmp}
Let us consider the top oper, namely we fix $(y_0,y_1)=(1,1)$ ($d_0=d_1=0$).
By definition
\begin{equation*}
\bar{L}_1(x,\la;l,1,1)=\partial_x+
\begin{pmatrix}
\frac{l}{x}& \lambda \\
P(x)\lambda& -\frac{l}{x}
\end{pmatrix}.
\end{equation*}
The trigonometric reproduction in the 0-th direction yields
\begin{equation}\label{rep ex}
R_0[y_0](x)=x-s_1,\qquad s_1=\frac{2m+2}{2m+1}.
\end{equation}
With $B_0$ as in Lemma \ref{lem:rep0Gauge}, 
\begin{align*}
 B_0 \bar{L}_1(x,\la;1,1)B_0^{-1} & =\partial_x+
\begin{pmatrix}
\frac{l+2m+1}{x}+\frac{1}{x-s_1}& \lambda \\
P(x)\lambda& -\frac{l+2m+1}{x}-\frac{1}{x-s_1}
\end{pmatrix} \\
& = \bar{L}_1(x,\la;R_0[l],x-s_1,1).
\end{align*}
\end{exmp}

We expect that for generic values of the parameters, the Bethe Ansatz is complete. In that case using the reproduction procedure, one can relate any solutions of BAE with any $r$ to a solution of the BAE with $r=0$.


\subsection{The case of degree \texorpdfstring{$1$}{1}}\label{Depth one subsubsection} Consider the operator \eqref{L1 BLZ} with $d_0=1$, that is
\ben\label{depth 1 oper}
L^G=\partial_x^2-\left(\frac{l(l+1)}{x^2}+\frac{2}{(x-s)^2}+\frac{k+s/(s-1)}{x(x-s)}+x^k(x-1)\la^2\right),
\een
where we put  $s=s_1$. Equation \eqref{gen mon eq}, providing the trivial monodromy conditions at  $x=s$ in this particular case takes the form:
\ben\label{L1 depth 1}
\frac{(s-s^{(1)})(s-s^{(0)})(s-s^{(-1)})}{s^3(s-1)^3}=0, 
\een
where
$$s^{(1)}=\frac{k-2l}{k-2l+1},\ \ s^{(0)}=\frac{k+2}{k+3},\ \ s^{(-1)}=\frac{k+2l+2}{k+2l+3}.$$
Thus there are exactly 3 operators of the form \eqref{depth 1 oper} with trivial monodromy at $x=s$. We claim that these three operators to three vectors of degree one in $L_{1,0}$, see Figure \ref{basic module pic}.

The case $s=s^{(1)}$ is the unique solution  of the Gaudin BAE equations for $\mc L_{1,0}^{sing}$:
\ben\label{1,0 BAE}
\frac{k-2l}{s}+\frac{1}{s-1}=0. 
\een
Therefore, $s=s^{(1)}$ corresponds to the highest weight vector of $M_{c,\Delta_1}$ in \eqref{coset} and $r=1$.  Note that $s^{(1)}$ can be obtained from the trivial solution with $d_0=0$ by the reproduction procedure, see \eqref{rep ex}. Indeed, 
\begin{equation*}
\frac{k-R_0[l]}{k-R_0[l]-1}=\frac{2m+2}{2m+1}. 
\end{equation*}

The case $s=s^{(0)}$ corresponds  to the unique solution  of the Gaudin BAE equations for $\mc L^{sing}_{1,1}$. Indeed, for the case of $d=(1,1)$, the Gaudin BAE equations take the form
\begin{align*}
\frac{k-2l}{s}+\frac{1}{s-1}+\frac{2}{s-t}=0,\\
\frac{2l}{t}+\frac{2}{t-s}=0.
\end{align*}
There is a unique solution
\bean\label{1,1 sol}
t=\frac{l}{l+1}\,\frac{k+2}{k+3}, \qquad s=\frac{k+2}{k+3}.
\end{align}
Therefore, $s=s^{(0)}$ corresponds to the unique degree one vector of $M_{c,\Delta_0}$ in \eqref{coset} and $r=0$.
Finally, the case $s=s^{(-1)}$ corresponds  to the unique solution  of the Gaudin BAE equations for $\mc L^{sing}_{1,2}$. 
In this case we have equations
\begin{align*}
\frac{k-2l}{s}+\frac{1}{s-1}+\frac{2}{s-t_1}+\frac{2}{s-t_2}=0,\\
\frac{l}{t_1}+\frac{1}{t_1-s} -\frac{1}{t_1-t_2}=0,\\
\frac{l}{t_2}+\frac{1}{t_2-s} -\frac{1}{t_2-t_1}=0.
\end{align*}
One finds that there exists a unique (up to exchanging $t_1$ and $t_2$) solution $y=(y_1,y_2)$ given by
\bean
&y_0=x-s_2=x-\frac{k+2l+2}{k+2l+3},\label{12 explicit 2} \\ 
&y_1=(x-t_1)(x-t_2)=x^2-\frac{2l-1}{l}s_2x+\frac{2l-1}{2l+1}s_2^2\notag\\
&\hspace{15pt}=x^2-\frac{(2l-1)(k+2l+2)}{l(k+2l+3)}x+\frac{(2l-1)(k+2l+2)^2}{(2l+1)(k+2l+3)^2}.\label{12 explicit 1}
\end{align}
Thus this operator corresponds to the highest weight vector of $M_{c,\Delta_{-1}}$ in \eqref{coset} and $r=-1$.
Note that that the solution \eqref{12 explicit 2}, \eqref{12 explicit 1}, is obtained from $r=1$ solution $(x-s^{(1)},1)$ by the reproduction in the first direction. Indeed, we have
\be
y_0(x;R_1[l])= x-s^{(1)}
\ee 
and
 \be
 \on{Wr} (1, x^{2l+1} y_1(x;R_1[l]))=(2l+3) x^{2l}(x-s^{(1)})^2.
 \ee 

The Gaudin equations \eqref{Gaudin BAE} corresponding to \eqref{BLZ module} and $d_0=1$, $d_1>2$ have no isolated solutions. Indeed, \eqref{Gaudin BAE 2} is the $\slt$ (non-affine) equations corresponding to tensor product of Verma module of highest weight $2l$ and of a three-dimensional module irreducible module of highest weight $2$. The evaluation parameters are $0$ and $s$. There are no singular vectors of depth more than two and hence no BAE solutions, \cite{MV1}.
Thus, all three operators of the form \eqref{depth 1 oper} with trivial monodromy around $x=s$ are obtained from the $\la$-opers corresponding to solutions of Gaudin BAE according to Conjecture \ref{gaudin - no monodromy conj}.

\subsection{The operator \texorpdfstring{$L_0$}{L0}}\label{L0 sec} 
In the case \eqref{BLZ T}, the operator $L_0$ takes the form
\bean
L_0=\partial_x^2-\left(\frac{m(m+1)}{x^2}+\frac{3}{4(x-1)^2}+\sum_{i=1}^{d_1}\frac{2}{(x-t_i)^2}+\right.\hspace{40pt} \label{L0 BLZ}\\ \left. 
+\frac{b_0}{x}+\frac{b_1}{x-1} +\sum_{i=1}^{d_1}\frac{q_j}{x-t_j}+x^k(x-1)\la^2\right), \notag
\end{align}
where 
\bea
&b_0=-m+\sum_{i=1}^{d_0}\frac{2m}{s_i}-\sum_{j=1}^{d_1}\frac{2m}{t_j}, 
\qquad b_1=m+\sum_{i=1}^{d_0}\frac{1}{s_i-1}-\sum_{j=1}^{d_1}\frac{1}{t_j-1},\\
& q_j=\frac{P'(t_j)}{P(t_j)} =\frac{k}{t_j}+\frac{1}{t_j-1}. \\
\end{align*}
Adding all BAE equations \eqref{Gaudin BAE BLZ 1}, \eqref{Gaudin BAE BLZ 2}, we see that
\ben\label{inf reg}
b_0+b_1+\sum_{i=1}^{d_1} q_i=0.
\een
Thus the expansion at infinity (with generic $k$) contains no $1/x$ term.
As we know, operator $L_0$ is monodromy-free at $x=t_j$ and at $x=1$.
We expect the following to be true, cf. Conjecture \ref{gaudin - no monodromy conj}.
\begin{conj}\label{L0 conj}
For generic $m,k$, all operators of the form \eqref{L0 BLZ} with distinct $t_j$ different from $0,1$, satisfying \eqref{inf reg} and monodromy-free at $x=1$ and $x=t_j$, $j=1,\dots, d_1,$ correspond to isolated solutions of \eqref{Gaudin BAE BLZ}.
\end{conj}
We give a couple of examples.
Let $d_1=0$. Then equation \eqref{monodromy eq 3/4}
and relation \eqref{inf reg} have the form
\be
b_1^2-b_0-m(m+1)=0, \qquad b_1+b_0=0.
\ee
and have the solutions
\be
b_1^{(0,0)}=m, \qquad b_1^{(1,0)}=-m-1.
\ee
Clearly, it corresponds to solutions of 
 \eqref{Gaudin BAE BLZ},  given by $y=(1,1)$ (that is $d=(d_0,d_1)=(0,0)$) and $y=(x-2m/(2m+1),1)$ (that is $d=(1,0)$), see \eqref{1,0 BAE}. Also if $d_1=0$, then equations \eqref{Gaudin BAE BLZ 1} are $\slt$ periodic Gaudin BAE equations for tensor product of a Verma module and a $2$-dimensional vector representation. It is well known that such equations have no isolated solutions with $d_0>1$. Thus Conjecture \ref{L0 conj} holds for $d_1=0$.
For $d_1=1$, one also can do the computation explicitly. One has three equations: \eqref{monodromy eq} at $x=3$, \eqref{monodromy eq 3/4} at $x=1$,  and \eqref{inf reg} for three variables $b_0,b_1,t_1$. This system has two solutions
\begin{align*}
t^{(1,1)}&=\frac{(k+2)(k-2m)}{(k+3)(k-2m+2)}, \\ 
b_1^{(1,1)}&=-\frac{2m^2-3km+2k^2-6m+10k+12}{3k-2m+6},
\end{align*}
and
\begin{align*}
t^{(2,1)}&=\frac{(k+2)(k+2m+2)}{(k+3)(k+2m+4)}, \\  
b_1^{(2,1)}&=-\frac{2m^2+3km+2k^2+10m+13k+20}{3k-2m+8}.
\end{align*}
The former corresponds to the solution \eqref{Gaudin BAE BLZ},  with $d=(1,1)$ given by \eqref{1,1 sol}. 
The latter corresponds  to 
 to the unique (up to interchanging of $s_1$ and $s_2$) solution \eqref{Gaudin BAE BLZ} with $d=(2,1)$,
\be
\begin{cases}\frac{2m}{s_1}+\frac{1}{s_1-1}+\frac{2}{s_1-t}-\frac{2}{s_1-s_2}=0,\\
\frac{2m}{s_2}+\frac{1}{s_2-1}+\frac{2}{s_2-t} -\frac{2}{s_2-s_1}=0,\\
\frac{k-2m}{t}+\frac{2}{t-s_1} +\frac{2}{t-s_2}=0
\end{cases}
\ee
given by
\bea
&y_0=(x-s_1)(x-s_1)=\\
&\hspace{0pt}x^2-\frac{(4km+k+10m+4)(k+2m+2)}{(2m+1)(k+3)(k+2m+4)} \,x +\frac{2m(k+2)(k+2m+2)^2}{(2m+1)(k+3)(k+2m+4)^2},\\
&y_1=x-t^{(1,2)}=x-\frac{(k+2)(k+2m+2)}{(k+3)(k+2m+4)}. 
\end{align*}
It follows that Conjecture \ref{L0 conj} holds for $d_1=1$ as well.

Clearly, $L_0$ and $L_1$ are related by a gauge transformation, see \eqref{L0 and L1 are equiv}. This fact is well-known and crucial, see \cite{GLVW}. It also appears in \cite{KL}, equation (7.66).

\subsection{BLZ opers}\label{sub:BLZopers}
In the ODE/IM literature \cite{DTa,DT,DDT,BLZ4,CM1}, the Schr{\"o}dinger  operators associated to Quantum KdV have the form
\begin{equation}\label{eq:blzeq}
\tilde{L}= \partial_w^2-\left(w^{2\alpha}- E+\frac{\tilde{l}(\tilde{l}+1)}{w^2}- 2 \frac{d^2}{dw^2}\sum_{j=1}^{d} \log{(w^{2\alpha+2}-\tilde{z}_j)} \right),
\end{equation}
with $\alpha>0$, $\Re \tilde{l}>-\frac12$ and where  the $\tilde{z}_i$ are distinct non-zero complex numbers such that
the monodromy at $w=z_i$ is trivial for every $i$ and $E$. These operators are called \textit{BLZ opers} or \textit{monster potentials} and they were introduced in \cite{BLZ4} to generalise the original proposal
\cite{DTa,DT, BLZ5}, which corresponds to the case $d=0$. It was conjectured in \cite{BLZ4},
that the eigenvalues of the Hamiltonians of quantum KdV flows are given by suitable symmetric polynomials of $\tilde{z}_i$. These Hamiltonians are written in terms of Virasoro algebra acting in a Verma module $M_{c,\Delta}$, where $c$ and $\Delta$ are written in terms of  the parameters $\tilde{l}$, $\alpha$ as
\begin{equation}\label{c delta original blz}
c=1-\frac{6\alpha^2}{\alpha+1},\qquad \Delta=\frac{(\tilde{l}+\tfrac12)^2-\alpha^2}{4(\alpha+1)},
\end{equation}
see \cite{BLZ4}. The same Virasoro Verma module is the space for our affine Gaudin system. It is expected that the eigenvectors of quantum KdV flows and of the affine Gaudin Hamiltonians coincide. Therefore we should have a natural correspondence of the BLZ operators \eqref{eq:blzeq} with the operators  \eqref{L1 BLZ}.
One of the main purposes of this paper is to study this correspondence. 

In order to compare the two operators, we first apply a change of variable and parameters, which is essentially the one considered in \cite[Sec. 5.7]{FF}, and which was also used in \cite{MR1,MR2}.  Let $L=\partial_w^2-V(w)$ and let $\tilde L^\phi$ be the operator obtained from $L$ by the change of coordinates $w=\phi(z)$.
Then
\begin{align}\label{L change of variable}
&L^\phi=(\phi')^{-1/2} ((\phi')^2 \tilde L^\phi)(\phi')^{1/2}= \partial_z^2-V^\phi(z), \notag\\ 
&V^\phi(z)= \left(\phi'(z)\right)^2V(\phi(z))-\frac12 \{\phi,z\}, 
\end{align}
where 
$$\{\phi,z\}=\frac{\phi'''(z)}{\phi'(z)}-\frac32 \left(\frac{\phi''(z)}{\phi'(z)}\right)^2,$$
is the Schwarzian derivative of $\phi$. 
Introducing new parameters $\bar{l}$, $\bar{k}$, $\bar{\la}$ and $\bar{z}_i$ through the relations
\begin{equation}\label{eq:ztowvariabele}
\tilde{l}+\frac12= \frac{2}{\bar{k}}\left(\bar{l}+\frac{1}{2}\right),\qquad \alpha=\frac{1}{\bar{k}}-1,\qquad E=\left(\frac{2}{\bar{k}}\right)^{2(1-\bar{k})}\bar{\lambda},\qquad \bar{z}_i=\left(\frac{\bar{k}}{2}\right)^2\tilde{z}_i,
\end{equation}
then under the change of variable
\begin{equation}\label{w -> z}
w=\varphi(z)=\left(\frac{2}{\bar{k}}\right)^{\bar{k}}z^\frac{\bar{k}}{2},
\end{equation}
the operator \eqref{eq:blzeq} takes the form
\ben\label{BLZ oper}
L^{BLZ}=L^\phi= \partial_z^2-\left(\frac{\bar l(\bar l+1)}{z^2}+\frac{1}{z} -z^{\bar k-2}\bar\la+\sum_{i=1}^{d}\frac{2}{(z-\bar{z}_i)^2}+\sum_{i=1}^{d} \frac{\bar k-2}{z(z-\bar{z}_i)}\right),
\een
with $ \bar{k} \in (0,1)$, $\Re \bar{l}>-\frac12$. In addition, the parameters \eqref{c delta original blz} now read
\begin{equation}\label{c delta original blz 1}
c=13-6\left(\frac{1}{\bar{k}}+\bar{k}\right),\qquad \Delta=\frac{(\bar{l}+\tfrac12)^2}{\bar{k}}+\frac{1}{2}-\frac{1}{4}\left(\frac{1}{\bar{k}}+\bar{k}\right).
\end{equation}
\begin{rem}\label{rem:blz parameters}
Putting $\beta=\bar{k}^{\frac12}$, $p=\bar{l}+\tfrac12,$ then  \eqref{c delta original blz 1} can be written as

\begin{equation}\label{cdeltabetap}
c=13-6\left(\frac{1}{\beta^2}+\beta^2\right),\qquad \Delta=\left(\frac{p}{\beta}\right)^2+\frac{c-1}{24},
\end{equation}
These expressions coincide with those appearing in \cite{BLZ1}. \qed
\end{rem}
The operator \eqref{BLZ oper} looks more similar to \eqref{L1 BLZ}, and in the case $d=d_0=0$ the two operators are actually related by a ($\lambda$-dependent) change of variable. In fact, under the change of coordinate
\ben\label{eq:phichange}
z=\phi(x)=\left(\frac{\la}{k+3}\right)^{2} x^{k+3},
\een
the $L^G$ oper \eqref{L1 BLZ} (with $d_0=0$) coincides with the BLZ oper \eqref{BLZ oper} (with $d=0$), provided the corresponding parameters are related by:
\ben\label{parameters change}
\bar k=\frac{k+2}{k+3}, \qquad \bar l +\frac12=\frac{l+\frac12}{k+3},\qquad
\bar\la=\left(\frac{\la}{k+3}\right)^{\frac{2}{k+3}}.
\een
We also notice that under this change of parameters the values \eqref{c delta original blz 1} for the central charge and highest weight of the Virasoro module $M_{c,\Delta}$ coincide with those in \eqref{charge and weight}. The relation with the general values \eqref{charge and weight general} is achieved by shifting $\bar{l}\mapsto \bar{l}-\bar{k}r$ in \eqref{BLZ oper} so that, using \eqref{parameters change} we get
$$\bar{l}-\bar{k}r+\frac12=\frac{l+\frac12}{k+3}-\bar{k}r=\frac{l-kr-2r+\frac12}{k+3}.$$
\begin{rem}\label{rem: shift}
The shift $\bar{l}\mapsto \bar{l}-\bar{k}r$ corresponds to an action of the Weyl group of $\slth$ on the oper \eqref{BLZ oper}, and to explain this fact it is sufficient to consider the case $d_0=0$, which we can write in matrix form as
$$\mc{L}'=\partial_z+\begin{pmatrix}\frac{\bar{l}}{z}&\frac1z+z^{\bar{k}-2}\bar{\lambda}\\1&-\frac{\bar{l}}{z}\end{pmatrix}=\partial_z+f_1+\frac{\bar{l}}{z}h_1+\left(\frac1z+z^{\bar{k}-2}\bar{\lambda}\right)e_1.$$
We consider a realization of $\slth$ as $\slth=\slt[\bar{\la},\bar{\la}^{-1}]\oplus\C K\oplus\C\bar{d}$, with $K$ central and $\bar{d}=-\bar{\la}\partial_{\bar{\la}}$, and with $f_0=\bar{\la}e_1$, $e_0=\bar{\la}^{-1}f_1$, and $h_0=K-h_1$. The oper $\mc{L}'$, see  \cite{MR1,MR2}, can be written as a meromorphic oper with values in $\slth$:
$$z^{\bar{k}\bar{d}}z^{\frac12 h_1}\mc{L}'=\partial_z+\frac{1}{z}\left(\hat{f}+\left(\bar{l}-\tfrac12\right)h_1-\bar{k}\bar{d}+ze_1\right),$$
where $\hat{f}=f_0+f_1$. The element $\left(\bar{l}-\tfrac12\right)h_1-\bar{k}\bar{d}$ belongs to the Cartan subalgebra of $\slth$, and if $s_0$, $s_1$ are simple reflections of the $\slth$ Weyl group we get 
$$(s_0s_1)^{r}\left[\left(\bar{l}-\tfrac12\right)h_1-\bar{k}\bar{d}\right]=\left(\bar{l}-r\bar{k}-\tfrac12\right)h_1-\bar{k}\bar{d}\mod{K},$$
which corresponds to the shift $\bar{l}\mapsto \bar{l}-\bar{k}r$.
\qed
\end{rem}
Combining  \eqref{w -> z}, \eqref{eq:ztowvariabele} with \eqref{eq:phichange}, \eqref{parameters change}, we deduce that in the case $d=d_0=0$, the relation between the ground state oper \eqref{eq:blzeq} and the ground state oper \eqref{L1 BLZ}  is given by the transformation
\begin{equation}\label{eq:xtowparameters}
 w= \left(\frac{2 \, \la}{k+2}\right)^{\frac{k+2}{k+3}}x^{\frac{k+2}{2}},\ \  \alpha=\frac{1}{k+2}, \ \ 
 \tilde{l}+\frac{1}{2}=\frac{2\left(l+\frac{1}{2}\right)}{k+2}, \ \  E=\left(\frac{2\la}{k+2}\right)^{\frac{2}{k+3}}.
\end{equation}
For $d_0>0$ the situation is more involved. 
From \eqref{monodromy eq}, we deduce that the trivial monodromy conditions are equivalent to the following system of algebraic equations \cite{BLZ4}, sometimes called the \textit{BLZ system},
\ben\label{BLZ BAE}
\bar f(\bar{z}_i)+\sum_{j=1, j\neq i}^{d} \bar g(\bar{z}_i,\bar{z}_j)=0
\een
with
\bea
&\bar f(z)= \frac{\bar{l}(\bar{l}+1)}{\bar{k}}+\frac{\bar{k}}{4}-z \frac{1-\bar{k}}{\bar{k}^2},\\
&\bar g(z,\tilde z)=z \frac{\bar{k}^2z^2+(2-\bar{k})(2\bar{k}+1)z\tilde z+(1-\bar{k})(2-\bar{k})\tilde{z}^2}{\bar{k}^2(z-\tilde z)^3}.
\end{align*}
We notice that the BLZ system is similar to the system of equations \eqref{gen mon eq} for the poles
$s_i$ of the oper $L^G$.
For example, BLZ oper with one apparent singularity is unique and has the form
\be
L^{BLZ}=\partial_z^2-\frac{\bar l(\bar l+1)}{z^2}-\frac{2}{(z-\bar{z}_1)^2}-\frac{1}{z}+\frac{\bar k-2}{z(z-\bar{z}_1)}+z^{\bar k-2}\bar\la, 
\ee
with
\be
 \bar{z}_1=  \frac{ \bar{k}(\bar{l}+\frac12)^2}{1-\bar{k}}-\frac{1-\bar{k}}{4}.
\ee
According to the general scheme, this operator  should be compared to operator \eqref{depth 1 oper} with $s=s_0$ as in \eqref{L1 depth 1}.  If $d_0>0$, the change of variables \eqref{eq:phichange} does not map $L^G$ opers to BLZ opers. In fact, since the change of variables is $\la$-dependent, the $\la$-independent additional poles $s_i$ in the $x$ co-ordinates are transformed into $\bar{\la}$-dependent poles in the $z$ coordinate; whence the transformed oper cannot be of the form \eqref{BLZ oper} unless $d_0=0$. 

We discuss the correspondence between $L^G$ opers and BLZ opers in Section \ref{comparison section}. This correspondence makes use of the spectral determinant of the opers which is a non-algebraic object. We do not expect an algebraic map which connects poles $s_i$ of $L^G$ to the poles $\bar z_j$ of $L^{BLZ}$ for all $k,l$.

\section{\texorpdfstring{$Q$}{Q}-functions and their properties}\label{Q chapter}
In this section we define and study the $Q$ and $T$ functions for the equation of the form \eqref{L1 BLZ}, adapting the construction in \cite{DTa}, \cite{DT,BLZ4}. Let $\{s_i\}_{i=1}^{d_0}$, $\{t_j\}_{j=1}^{d_1}$ be a solution of the Gaudin BAE \eqref{Gaudin BAE BLZ}.
Let $L^G$ be the operator given in \eqref{L1 BLZ}, and we consider the equation 
\begin{align}
&  L^G \psi=0,\label{diffeq:L1s}
\\
& L^G=\partial_x^2-\left(\frac{l(l+1)}{x^2}+\lambda^2x^k(x-1)+
 \sum_{i=1}^{d_0}\left(\frac{2}{(x-s_i)^2}+\frac{k+s_i/(s_i-1)}{x(x-s_i)}\right)\right),\label{diffeq:L1sop}
\end{align}
where $k\in\R$, $k>-2$, $\Re{l}\neq -1/2$. This implies that for all $i=1,\dots,d_0$ and each $\la$, the monodromy around $x=s_i$ is trivial. We expect that all operators of the form \eqref{diffeq:L1sop} where all $s_i$ are non-zero distinct numbers and the monodromy around $x=s_i$ is trivial, come from solutions of the Gaudin BAE, see  Conjecture \ref{gaudin - no monodromy conj}.
Let $\widetilde{\C^*}$ be the universal cover of $\C^*=\C\setminus\{0\}$ and $\Pi:\,\widetilde{\C^*}\to \C^*$ be the corresponding projection.
The operator \eqref{diffeq:L1sop}, which is defined on the complex plane with a cut, naturally extends to a linear differential operator on $\widetilde{\C^*}\setminus\left(\bigcup_{i=1}^{d_0} \Pi^{-1} (s_i)\right)$:
 \begin{align}\label{diffeq:L1sopPi}
 L^G=\partial_{\Pi(x)}^2-&\left(\frac{l(l+1)}{\Pi(x^2)}+\Pi(\lambda^2)\Pi(x^k)(\Pi(x)-1)+ \right.\notag\\ +
 &\sum_{i=1}^{d_0}\left.\left(\frac{2}{(\Pi(x)-s_i)^2}+\frac{k+s_i/(s_i-1)}{\Pi(x)(\Pi(x)-s_i)}\right)\right).
 \end{align}
By abuse of notation, we omit from now on the projection $\Pi$, writing for instance \eqref{diffeq:L1sop} in place of \eqref{diffeq:L1sopPi}.
Since we have assumed that the monodromy about each pole $x=s_i$ is trivial, for any fixed $\la$ every local solution to equation \eqref{diffeq:L1s} can be analytically continued to a global single-valued analytic function on $\widetilde{\C^*}\setminus\left(\bigcup_{i=1}^{d_0} \Pi^{-1} (s_i)\right)$. In what follows we will be interested in solutions of (\ref{diffeq:L1s}) which are analytic in $x$ and $\la$ separately
in the domain $\left(\widetilde{\C^*} \setminus \bigcup_{i=1}^{d_0} \Pi^{-1} (s_i) \right) \times \widetilde{\C^*}$.
\subsection{Space of solutions, rotated operators, and Dorey-Tateo symmetry}
For $\theta\in\R$, the multiplication by $e^{\pi i \theta}$ in $\C$ extends to the universal cover $\widetilde{\C^*}$. Note that for 
 $x\in\widetilde{\C^*}$, we have  $e^{2\pi i}x  \neq x$.  
 For $(x,\la)\in \widetilde{\C^*}\times \widetilde{\C^*}$ we consider the transformation
 \ben\label{transformationxla}
 (x,\la)\mapsto (e^{2\pi i\theta}x,e^{-(k+2)\pi i\theta}\la).
 \een
\begin{defn}\label{A}
Given the operator $L^G$ as in \eqref{diffeq:L1sop} we denote by $L^G_{[\theta]}$ the operator induced on $L^G$  by the transformation \eqref{transformationxla}, and for a function $\psi(x,\la)$ we define
\ben\label{psi theta}
\psi_{[\theta]}(x,\la)=e^{-\pi i\theta}\psi(e^{2\pi i\theta} x,e^{-(k+2)\pi i\theta}\la).
\een
We denote by  $\mathcal{A}_{[\theta]}$  the vector space of functions in $\ker L^G_{[\theta]}$ which are analytic for $(x,\lambda)$ in the domain $ \left(\widetilde{\C^*} \setminus \bigcup_{i=1}^{d_0} \Pi^{-1} (s_i) \right) \times \widetilde{\C^*}$.  For $\theta=0$ we also write $\mathcal{A}$ in place of $\mathcal{A}_{[0]}$. 
\end{defn}
We will show below, see \eqref{wronsian j j+1}, that the space  $\mathcal{A}_{[\theta]}$ is a free module of rank $2$ over the ring $\mathcal{O}^*_{\la}$ of analytic functions in the variable
$\la$ with domain $\widetilde{\C^*}$, cf. also \cite{MR1}.
In other words, there exists two functions  $\phi,\chi \in \mathcal{A}_{[\theta]} $ such that  $\lbrace \phi(\cdot,\la) ,\chi(\cdot,\la)\rbrace$ is a basis of solutions of equation (\ref{diffeq:L1s})  for every fixed (non-zero) value of $\la$. 
We thus say that $\lbrace \phi,\chi \rbrace$ is an $\mathcal{O}^*_{\la}$-basis of $\mathcal{A}_{[\theta]}$ and we identify $\mathcal{A}_{[\theta]}$ with $\C^2 \otimes_{\mathbb{C}} \mathcal{O}^*_{\la}$  as vector spaces over $\C$. 
We collect below some important invariant properties of $L^G_{[\theta]}$ and $\mathcal{A}_{[\theta]}$, whose proof is elementary.
\begin{lem}\label{lemmaLA}
i) Consider  the transformation
 \ben\label{transformationxla1}
 (x,\la)\mapsto (e^{2\pi i}x,e^{-(k+2)\pi i}\la),
 \een
 which maps $L^G\mapsto L^G_{[1]}$ and $\psi\mapsto \psi_{[1]}$. The operator (\ref{diffeq:L1sop}) is fixed by \eqref{transformationxla1}, namely $L^G_{[1]}=L^G$. It follows from this that for every $\theta\in\R$ the operator $L^G_{[\theta]}$ admits the following invariance
\ben\label{dorey-tateo}
L^G_{[\theta+j]}=L^G_{[\theta]},\qquad j\in\Z,
\een
and that the space $\mathcal{A}_{[\theta]}$ is invariant under \eqref{transformationxla1}:  $\psi\in\mathcal{A}_\theta$ if and only if  $\psi_1\in \mathcal{A}_\theta$.\\
ii)  For every $\theta\in\R$ the operator $L^G_{[\theta]}$ is fixed by the transformation
\ben\label{e pi i la}
(x,\la)\mapsto(x,e^{i\pi}\la),
\een
and therefore the space $\mathcal{A}_{[\theta]}$ is invariant: $\psi(x,\la)\in\mathcal{A}_{[\theta]}$ if and only if $\psi(x,e^{i\pi}\la)\in\mathcal{A}_{[\theta]}$.
\end{lem}
Property  \eqref{dorey-tateo} is known as Dorey-Tateo symmetry.
Below it will be useful to consider the Wronskian (with respect to $x$) of a pair of functions depending on $(x,\la)$.  We denote by
$$\on{Wr}(\psi,\chi)(x,\lambda)=\psi(x,\lambda)\partial_x\chi(x,\la)-\partial_x\psi(x,\la)\chi(x,\la)$$
the Wronksian with respect to the variable $x$ of the pair of functions $\psi(x,\la)$ and $\chi(x,\la)$, see also \eqref{Wr def}.
Under the map \eqref{psi theta} the Wronskian transforms as a function:
\ben\label{wr theta}
\on{Wr}(\psi_{[\theta]},\chi_{[\theta]})(x,\lambda)=\on{Wr}(\psi,\chi) (e^{2\pi i\theta}x,e^{-(k+2)\pi i\theta}\la).
\een
We also introduce the  \textit{monodromy}  transformation, defined as the ($\C$-linear) map
\begin{equation}\label{eq:monodromyoperator}
 \big(\mathcal{M} \psi \big)(x,\la)= \psi(e^{2\pi i} x,e^{-(k+2) \pi i } \la)=-\psi_{[1]}(x,\la),
\end{equation}
and we extend its action on column vectors by setting $\mathcal{M}(\psi,\chi)^t=(\mathcal{M}\psi,\mathcal{M}\chi)^t$.  The action of the monodromy on a given basis $\{\psi,\chi\}\subset \mathcal{A}$ uniquely defines a  matrix, known as \emph{monodromy matrix} (with respect to the chosen basis) and given by
\ben\label{monodromy matrix}
\mathcal{M}
\begin{pmatrix}
\psi\\
\chi
\end{pmatrix}
=
\begin{pmatrix}
c_{11} & c_{12} \\
c_{21} & c_{22} 
\end{pmatrix}
\begin{pmatrix}
\psi \\
\chi 
\end{pmatrix},
\een
for some $c_{ij}\in \mathcal{O}^*_{\la}$.
In order to define the $Q$ functions, we introduce below a set of distinguished solutions at $0$ and $\infty$.

\subsection{Frobenius solutions at \texorpdfstring{$x=0$}{x=0}}
For every fixed $\la$ the point $x=0$ is formally a  Fuchsian singular point of equation (\ref{diffeq:L1s}) because it admits the expansion
\begin{equation*}
 \psi''(x)= \left( \frac{l(l+1)}{x^2} + o\big(|x|^{-2}\big)\right) \psi(x), \qquad x \sim 0.
\end{equation*}
However, $x=0$ it is actually not a Fuchsian singular point whenever $k$ is not integer,
since in this case $x=0$ is a branch-point of the potential. To overcome this difficulty we consider equation (\ref{diffeq:L1s}) as an equation with values
in the space $\mathcal{O}_{\xi}$ of entire functions of the variable $\xi=\la^2 x^{k+2} $ \cite{MR1,MR2}. More precisely,  we consider the change of variables
$$
\begin{cases}
\xi=\lambda^2x^{k+2},\\
\tau=x,
\end{cases}
$$
and notice that $\xi$ is invariant under the monodromy map \eqref{eq:monodromyoperator}. Introducing the operator
$$D_\xi=(k+2)\xi\partial_\xi,$$
one has
\be
\partial_x=\partial_\tau+\frac{D_\xi}{\tau},\qquad \partial_x^2=\partial_\tau^2+\frac{2D_\xi}{\tau}\partial_\tau+\frac{D_\xi^2-D_\xi}{\tau^2},
\ee
and in the new variables the equation reads
\begin{align}\label{eq:L1partial}
&\tilde{L}\psi=0,\notag\\
&\tilde{L}=-\partial_\tau^2- \frac{2D_{\xi}}{\tau}\partial_\tau-
\\&-\frac{D_{\xi}^2-D_{\xi}-l(l+1)+\xi}{\tau^2}
+\frac{\xi}{\tau}+\sum_{i=1}^{d_0}\left(\frac{2}{(\tau-s_i)^2}+\frac{k+s_i/(s_i-1)}{\tau(\tau-s_i)}\right). \notag
\end{align}
For a function $\psi(x,\la)$ we define the function $\Psi(\tau,\xi)$ by the formula
\be
\Psi(\tau,\xi)=\psi(x(\tau,\xi), \la(\tau,\xi))=\psi(\tau, \xi \tau^{-k-2}),
\ee
so that
$$
 \psi(x,\la)= \Psi(x,\la x^{k+2}).
$$
We look for solutions $\Psi(\tau, \xi) $ to the equation $\tilde{L}\Psi=0$ such that for every $x \in \widetilde{\C^*}$ the function $\xi\mapsto\Psi(\tau, \xi) $ is entire. Within this approach, a (generalised) Frobenius solution is a series of the form
\begin{equation}\label{eq:genfrobenius}
\Psi(\tau,\xi)=\tau^c \sum_{j=0}^\infty g_j(\xi)\tau^j, \qquad  g_j \in \mathcal{O}_\xi,\quad  g_0(0)=1,
\end{equation}
for some $c$.  Given such a $c$, we expect the corresponding series to be convergent for every $|\tau|< \min_{i} |s_i|$ and every $\xi$. 
We remark that given a solution of $\tilde{L}\Psi=0$   of the form \eqref{eq:genfrobenius}, the function 
\be
\psi(x,\lambda)=\Psi(x,\lambda^2x^{k+2})=x^c\sum_{j=0}^{\infty}g_j(\lambda^2x^{k+2})x^j
\ee
is a solution to equation \eqref{diffeq:L1s} and
\begin{align*}
&\lim_{x\to 0} x^{-c} \psi(x,\la)=g_0(0)=1 ,\\
& \mathcal{M} \psi = e^{2\pi ic } \psi ,
\end{align*}
where $\mathcal{M}$ is the monodromy operator defined in equation (\ref{eq:monodromyoperator}). To find the Frobenius solutions we apply the operator \eqref{eq:L1partial} to the Frobenius series (\ref{eq:genfrobenius}) and study the recursion on $g_j$. We expand the $\underline{s}$-dependent part of the potentials at $\tau=0$ as
$$ 
\sum_{i=1}^{d_0}\left(\frac{2}{(\tau-s_i)^2}+\frac{k+s_i/(s_i-1)}{\tau(\tau-s_i)}\right)=\sum_{j=-1}^{\infty}f_j\tau^j  , \qquad 0<|\tau|< \min_{i} |s_i|, 
$$
so that the operator $\tilde{L}$ in \eqref{eq:L1partial} reads
\be
\tilde{L}=-\partial_\tau^2-\frac{2D_{\xi}}{\tau}\partial_\tau-\frac{D_{\xi}^2-D_{\xi}-l(l+1)+\xi}{\tau^2}
+\frac{\xi}{\tau}+ \sum_{j=-1}^{\infty}f_j\tau^j.
\ee
Substituting \eqref{eq:genfrobenius} into the equation $\tilde{L}\psi=0$ and expanding in $\tau$ we  obtain the recursion 
\begin{equation}\label{230119-6}
A_{\alpha,j}g_j=\xi g_{j-1}+\sum_{h=0}^{j-1}f_{j-2-h}g_h , \qquad g_{-1}=0, \; g_0(0)=1.
\end{equation}
where
\begin{equation}\label{Aaj}
A_{\alpha,j}:\mc{O}_\xi\to \mc{O}_\xi,\qquad A_{\alpha,j}=D_{\xi}^2+(2\alpha+2j-1)D_{\xi}+P_l(\alpha+j)+\xi,
\end{equation}
and $P_l(\alpha)$   is the indicial polynomial of the operator $-\partial_{\tau}^2+\frac{l(l+1)}{\tau^2}$, given explicitly by
$$
P_l(\alpha)  =\alpha^2-\alpha-l(l+1)=(\alpha-l-1)(\alpha+l).
$$
The operator $A_{\alpha,j}$ is a linear differential operator of the
second order. It
has a Fuchsian singularity at $\xi=0$, an irregular singularity at
$\xi=\infty$, and no other singularities. It follows from this that
$\ker A_{\alpha,j}\cap \mathcal{O}_\xi \neq 0$ if
and only if
the indicial polynomial of the operator $ A_{\alpha,j}$, given by
\begin{equation} \label{230120-1}
 R_{j}(r;\alpha)=\big((k+2)r + \alpha+j -l-1\big)\big((k+2)r+ \alpha+j +l\big).
\end{equation}
has a non-negative integer root. To further characterise this kernel we notice
that it is related to the Bessel equation: more precisely, $y$ is a solution to the Bessel equation
of order $\nu=\frac{2l+1}{k+2}$, namely
\begin{equation}\label{eq:bessel}
z^2 \frac{d^2y}{dz^2}+z\frac{dy}{dz}+(z^2-\nu^2)y=0, \qquad \nu=\frac{2l+1}{k+2},
\end{equation}
if and only if
\begin{equation}\label{eq:kernelbessel}
 A_{\alpha,j} \left( \xi^{\frac{1-2\alpha-2j}{4+2k}} y\left(\frac{2
\xi^{\frac12}}{k+2}\right) \right)=0.
\end{equation}
It follows from this that $\dim \big( \ker A_{\alpha,j} \cap
\mathcal{O}_\xi \big) \leq 1 $
for every value of $\alpha,j$, and moreover that in the case $\dim \big( \ker
A_{\alpha,j} \cap  \mathcal{O}_\xi  \big) =1$ the holomorphic kernel is spanned
by a solution with a zero of order $r_+$ at $0$,
where $r_+$ is the greatest integer root of $R_{j}(r;\alpha)$.  In fact,
if the indicial polynomial (\ref{230120-1}) admits
two integer roots then $\nu \in \mathbb{Z}_{\geq 0}$ and the general
solution of the Bessel equation (\ref{eq:bessel}) has a logarithmic
term \cite{EMOT}. Finally we recall that the general solution of the Bessel equation (\ref{eq:bessel}) can be expressed in terms of the Bessel function $J_{\nu}(z)$, whose Frobenius expansion at $z=0$ is given by
\ben\label{Jnu}
J_{\nu}\left(z\right)=(\tfrac{1}{2}z)^{\nu}\sum_{j=0}^{\infty}(-1)^{j}\frac{(\tfrac{1}{4}z^{2})^{j}}{j!\,\Gamma\left(\nu+j+1\right)}.
\een
In the proposition below, we analyse the recursion (\ref{230119-6}).
Before that, we recall that the space $\mathcal{O}_{\xi}$ is a Frechet space whose topology is induced by the following family of norms,
parametrised by the real positive number $\sigma \in (0, +\infty)$:
$$
 \| f \|_{\sigma}= \sup_{i} |a_i| \sigma^{i}, \qquad \mbox{ where }\qquad f(\xi)=\sum_{i=0}^{\infty} a_i \xi^i.
$$

\begin{prop}\label{prop:frobenius}
Let $k>-2$ and  $\Re l \neq -\frac12$.
\begin{enumerate}[i)]
\item If the recursion (\ref{230119-6}) has a solution then $\alpha=l+1$ or $\alpha=-l$.
\item If $\Re l >-\frac12$, the recursion (\ref{230119-6}) with $\alpha=l+1$ admits a unique solution $\{g_j^+,j\geq0\}$, where $g_j^+$ are entire functions such that  for every $0<\rho < \min_{i} |s_i|$ and every $\sigma >0$ there exists a constant $C_{\rho,\sigma}$ such that
 \begin{equation}\label{eq:recursionconvergence+}
  \|g^+_j \|_{\sigma} \leq C_{\rho,\sigma} \rho^{-j} .
 \end{equation}
The term $g^+_0$ admits the explicit representation 
 \begin{equation}\label{eq:g0+Bessel}
  g^+_0(\xi)=(k+2)^{\frac{2l+1}{k+2}}\Gamma\left(1+\frac{2l+1}{k+2}\right)\xi^{-\frac{2l+1}{k+2}}
  J_{\frac{2l+1}{k+2}}\left(\frac{2\sqrt{\xi}}{k+2}\right),
 \end{equation}
where $J_{\frac{2l+1}{k+2}}(\cdot)$ is the Bessel function of order $\frac{2l+1}{k+2}$, see \eqref{Jnu}. \\
Moreover, if $l$ is such that $2l+1 \notin \lbrace (k+2) i+j, (i,j)\in \mathbb{Z}_{\geq 0}^2 \rbrace$,
the recursion (\ref{230119-6}) with $\alpha=-l$ admits a unique solution $\{g_j^-,j\geq 0\}$, where $g_j^-$ are entire functions such that
 for every $0<\rho < \min_{i} |s_i|$ and every $\sigma >0$ there exists a  constant $C_{\rho,\sigma}$ such that
 \begin{equation}\label{eq:recursionconvergence-}
  \|g^-_j \|_{\sigma} \leq C_{\rho,\sigma} \rho^{-j} .
 \end{equation}
The term $g^-_0$ admits the explicit representation 
 \begin{equation}\label{eq:g0-Bessel}
  g^-_0(\xi)=(k+2)^{-\frac{2l+1}{k+2}}\Gamma\left(1-\frac{2l+1}{k+2}\right)\xi^{+\frac{2l+1}{k+2}}
  J_{-\frac{2l+1}{k+2}}\left(\frac{2\sqrt{\xi}}{k+2}\right),
 \end{equation}
where $J_{-\frac{2l+1}{k+2}}(\cdot)$ is the Bessel function of order $-\frac{2l+1}{k+2}$, see \eqref{Jnu}.
\item If $\Re l <-\frac12$, the recursion (\ref{230119-6}) with $\alpha=-l$ admits a unique solution $\{g_j^-,j\geq 0\}$, where $g_j^-$ are entire functions such that
 for every $0<\rho < \min_{i} |s_i|$ and every $\sigma >0$ there exists a constant $C_{\rho,\sigma}$ such that $ \|g^-_j \|_{\sigma} \leq C_{\rho,\sigma} \rho^{-j} $.
The term $g^-_0$ admits the explicit representation \eqref{eq:g0-Bessel}.  \\
Moreover, if $l$ is such that
$-(2l+1) \notin \lbrace (k+2) i+j, (i,j)\in \mathbb{Z}_{\geq 0}^2 \rbrace$,
the recursion (\ref{230119-6}) with $\alpha=l+1$ admits a unique solution $\{g_j^+,j\geq0\}$, where $g_j^+$ are entire functions such that for every $0<\rho < \min_{i} |s_i|$ and every $\sigma >0$ there exists a constant $C_{\rho,\sigma}$ such that
$\|g^+_j \|_{\sigma} \leq C_{\rho,\sigma} \rho^{-j}$.
The term $g^+_0$ admits the explicit representation \eqref{eq:g0+Bessel}.
\end{enumerate}
\end{prop}
\begin{proof}
\begin{enumerate}[i)]
\item We study \eqref{230119-6} in the case $j=0$ to deduce the values of $\alpha$ for which the recursion can be solved.
The equation reads
\begin{equation}\label{230119-7}
 A_{\alpha,0}g_0=0,\qquad g_0 \in \mathcal{O}_{\xi}, \, g_0(0)=1 .
 \end{equation}
This equation admits a solution if and only if
$r=0$ is the greatest integer root of the indicial polynomial $R_{0}  (r;\alpha)$, given by \eqref{230120-1}. Hence, equation \eqref{230119-7} admits a solution only if
$P_l(\alpha)=0$, namely $\alpha=l+1$ or $\alpha=-l$. 

\item We introduce the notation
$$
A^+_j=A_{l+1,j}=D_\xi^2+(2j+2l+1)D_\xi+j(j+2l+1)+\xi,\qquad j\geq 0.
$$
for operator \eqref{Aaj} with $\alpha=l+1$, with indicial polynomials,  see \eqref{230120-1},
$$
R^+_{j}(r)=R_{j}(r;l)=\left((k+2)r+j\right)\left((k+2)r+j+(2l+1)\right),\qquad j\geq 0.
$$
If $j=0$, then since $\Re (2l+1) \geq 0$ it follows that $r=0$ is the only non-negative integer root $R^+_{0}(r)$. Therefore the equation
$A^+_{0}g_0=0,\; g_0 \in \mathcal{O}_{\xi}, \, g_0(0)=1$ admits a unique solution $g_0^+$. The explicit formula (\ref{eq:g0+Bessel}) for $g_0^+$  is obtained from  (\ref{eq:kernelbessel}) and the Frobenius expansion \eqref{Jnu} of the Bessel function.\\
\noindent
If $j\geq 1$, then by direct inspection we notice that the polynomial $R^+_{j}(r)$ does not have any non-negative integer roots; hence 
$$
 \ker A^+_{j}\cap \mc O_\xi=0,\qquad \forall j\geq1.
$$
We now show that $A^+_{j}$ is an invertible operator in $\mc O_\xi$ and that the following estimate holds:
for every $\sigma>0$, there exists a $C(\sigma)$ such that
\begin{equation}\label{eq:normofinverse+}
 \|(A^+_j )^{-1}\|_{\sigma} \leq \frac{C(\sigma)}{j(\Re(2l+1)+j)},\qquad \forall j\geq 1.
\end{equation}
To prove this inequality, we fix  $\sigma>0$  and let $a,b\in\mc{O}_\xi$ be given by
$$
a(\xi)=\sum_{h=0}^{\infty}a_h\sigma^{-h}\xi^h,\qquad b(\xi)=\sum_{h=0}^{\infty}b_h\sigma^{-h}\xi^h.
$$
so that
$$
\| a\|_{\sigma}=  \sup_h|a_h|, \qquad \| b\|_{\sigma}=  \sup_h|b_h| .
$$
In order to estimate  $\|(A^+_j )^{-1}\|_{\sigma}$ we study $A^+_j a=b$ with $b$ such that $\|b\|_{\sigma}=1$. Equation $A^+_j a=b$ is equivalent to the recursion
$$
R^+_j(h)a_h+\sigma a_{h-1}=b_{h},\qquad h\geq 0, \; a_{-1}=0,
$$
whose explicit solution  is given by
$$
a_h=\sum_{i=0}^h\frac{(-1)^{h-i}\sigma^{h-i}b_i}{\prod_{r=i}^hR^+_j(r)},\qquad h\geq 0.
$$
Since  $|R^+_{j}(r)|^{-1}=O(r^{-2})$ as $r\to\infty$ for every $j\geq1$, then from the above formula we deduce that 
\begin{equation*}
 \sup_{\|b\|_{\sigma}=1} \sup_{h} |a_h| = \sup_{\|b\|_{\sigma}=1} \sup_{h} \left|\sum_{i=0}^h\frac{(-1)^{h-i}\sigma^{h-i}b_i}{\prod_{r=i}^hR^+_j(r)}\right|< \infty,\qquad \forall j \geq 1,
\end{equation*}
which in turn implies that
\begin{equation}\label{eq:boundedinverse+}
 \|(A^+_j )^{-1}\|_{\sigma} < +\infty,\qquad \forall \sigma >0, \forall j \geq 1.
\end{equation}
Therefore, $A^+_{j}$ is invertible. We now prove the estimate \eqref{eq:normofinverse+}. Since $|R^+_{j}(r)|>|R^+_{j}(0)|>0$ for $r>0$, then for all $b \neq 0$ we have
\begin{align}\label{eq:estimateahproof}
\frac{|a_h|}{\|b \|_{\sigma}}&\leq \sum_{i=0}^h \frac{\sigma^{h-i}}{\prod_{r=i}^h|R^+_j(r)|}\leq \sum_{i=0}^h \frac{\sigma^{h-i}}{\prod_{r=i}^h|R^+_j(0)|}\notag\\
&=\frac{1}{|R_j^+(0)|}\sum_{i=0}^h\left(\frac{\sigma}{|R_j^+(0)|}\right)^{h-i}=\frac{1}{|R_j^+(0)|}\frac{1-(\sigma/|R^+_j(0)|)^{h+1}}{1-\sigma/|R^+_j(0)|} \notag\\
&\leq \frac{1}{j(j+\Re(2l+1))}\frac{1-(\frac{\sigma}{j(j+\Re(2l+1))})^{h+1}}{1-\frac{\sigma}{j(j+\Re(2l+1))}}.
\end{align}
Now fix $C'\in(0,1)$. Then, there exists a $j_0$ such that  $\frac{\sigma}{j(j+\Re(2l+1))}\leq C'$, $\forall j \geq j_0$, and from (\ref{eq:estimateahproof}) we deduce that there exists a $C''\in (0,+\infty)$ such that
$$
\frac{|a_h|}{\|b \|_{\sigma}}\leq \frac{C''}{j(j+\Re(2l+1))} , \qquad  \forall b \neq0, \,  \forall j \geq j_0,  \, \forall h \geq 0 .
$$
Combining this estimate with \eqref{eq:boundedinverse+} we obtain (\ref{eq:normofinverse+}), and using the latter we prove that the unique solution of the recursion (\ref{230119-6}) satisfies (\ref{eq:recursionconvergence+}). In fact,
for $j \geq 1$, the recursion (\ref{230119-6})
reads
$$
g^+_j=(A^+_j )^{-1}\xi g^+_{j-1}+\sum_{h=0}^{j-1}f_{j-2-h} (A^+_j )^{-1} g^+_h ,
$$
and using (\ref{eq:normofinverse+}), we deduce that
\begin{equation}\label{230130-5}
\|g^+_j\|_{\sigma} \leq \frac{C(\sigma)}{j(\Re(2l+1)+j)} \left( \sigma \|g^+_{j-1}\|_{\sigma} +\sum_{h=0}^{j-1}|f_{j-2-h}| \|g^+_h\|_{\sigma} \right), \qquad j \geq 1.
\end{equation}
In the above inequality we have used the fact that (with respect to the norm $\| \cdot \|_{\sigma}$) the operator of multiplication by $\xi$ has norm $\sigma $. Now let $\rho$ be such that $0<\rho<\min_i|s_i|$. Denoting   $\widetilde{f}_j=\rho^j |f_j|$ and $\widetilde{g}^+_j=\rho^j \|g^+_j\|_{\sigma}$ we can rewrite \eqref{230130-5} as
$$
 \widetilde{g}^+_j  \leq  \frac{C(\sigma)}{j(\Re(2l+1)+j)} \left( \rho \sigma \widetilde{g}^+_{j-1} +  \rho^2\sum_{h=0}^{j-1} \widetilde{f}_{j-2-h} \widetilde{g}^+_h \right) .
$$
Since
$\sup_j\widetilde{f}_j=\sup_j |f_j|\rho^{j} <+\infty$ then
there exists a $C'_{\rho,\sigma}$ such that
$$
\widetilde{g}^+_j   \leq  \frac{C'_{\rho,\sigma} }{\Re(2l+1)+j}  \left(\sup_{l \leq j-1} \widetilde{g}^+_l \right).
$$
It follows that the sequence $\lbrace \sup_{l \leq j} \widetilde{g}_j^+\rbrace_{j \in \mathbb{Z}_{\geq 0}}=\lbrace \sup_{l \leq j} \rho^jg_j^+\rbrace_{j \in \mathbb{Z}_{\geq 0}}$ is bounded, hence (\ref{eq:recursionconvergence+}) holds.

We study now the recursion (\ref{230119-6}) when $\alpha=-l$ and $l> - \frac12$
We introduce the notation
$$
A^-_j=A_{-l,j}=D_{\xi}^2+(2j-(2l+1))D_{\xi}+j(j-(2l+1))+\xi.
$$
which have indicial polynomials
\begin{equation}\label{230130-1}
R^-_{j}(r)=R_{j}(r;l)=\left((k+2)r+j\right)\left((k+2)r+j-(2l+1)\right).
\end{equation}

\noindent
If $j=0$, the equation $A^-_{0}g_0^-=0,\; g_0^- \in \mathcal{O}_{\xi}, \, g_0^-(0)=1$ admits a (a fortiori unique) solution if and only if
the polynomial $R^-_{0}(r)$ does not have any positive integer root. A direct inspection shows that the latter condition is equivalent to
$2l+1 \notin\lbrace (k+2) i, i\geq 1 \rbrace$. Whenever the latter condition is met, the explicit formula (\ref{eq:g0-Bessel}) for $g_0^-$ follows by the same computation as above.

\noindent
If $j\geq 1$ then by looking at the indicial polynomial \eqref{230130-1} we obtain $\ker A^-_{j}=0$ if and only if
$$2l+1\neq(k+2)i+j,\qquad i,j\in\mathbb{Z}_{\geq 0}, i\geq0\,j\geq1.$$
Assuming $2l+1\notin \lbrace (k+2)i+j,\,\, i,j\in\mathbb{Z}_{\geq 0}, i\geq0\,j\geq1 \rbrace$, one obtains the estimate (\ref{eq:recursionconvergence-}) by following the same steps as in the proof of (\ref{eq:recursionconvergence+}) above.
\item Same as ii), replacing $l \to -l-1$.
\end{enumerate}
\end{proof}
The following corollary gives a basis of solutions of \eqref{diffeq:L1s} with special properties at $x=0$ which we also call Frobenius solutions. Recall the vector space $\mc{A}$  introduced in Definition \ref{A}.
\begin{cor}\label{cor:frobenius}
Let $k>-2$,  $\Re l \neq -\frac12$, and let $g^\pm_j$, $j\in\Z_{\geq0}$ be the entire functions defined in Proposition \ref{prop:frobenius}.
\begin{enumerate}[i)]
\item If $\Re l > -\frac12$ there exists $\chi^+\in\mc{A}$ such that
 \begin{equation}\label{eq:chipmx}
  \chi^+(x,\la)= x^{l+1} \,  \sum_{j=0}^{\infty}g^{+}_j(\la^2 x^{k+2})\, x^j, \qquad x\to 0.
 \end{equation}   
If $2l+1 \notin \lbrace (k+2) i+j, (i,j)\in \mathbb{Z}_{\geq 0}^2 \rbrace$, then  there exists $\chi^-\in\mc{A}$ such that
 \begin{equation}\label{eq:chipmx2}
  \chi^-(x,\la)= x^{-l} \,  \sum_{j=0}^{\infty}g^{-}_j(\la^2 x^{k+2})\, x^j, \qquad x\to 0.
 \end{equation}  
\item If $\Re l < -\frac12$ there exists $\chi^-\in\mc{A}$ such that
 \begin{equation}\label{eq:chipmx2bis}
  \chi^-(x,\la)= x^{-l} \,  \sum_{j=0}^{\infty}g^{-}_j(\la^2 x^{k+2})\, x^j, \qquad x\to 0.
 \end{equation} 
If $-(2l+1) \notin \lbrace (k+2) i+j, (i,j)\in \mathbb{Z}_{\geq 0}^2 \rbrace$, then  there exists $\chi^+\in\mc{A}$ such that
 \begin{equation}\label{eq:chipmxbis}
  \chi^+(x,\la)= x^{l+1} \,  \sum_{j=0}^{\infty}g^{+}_j(\la^2 x^{k+2})\, x^j, \qquad x\to 0.
 \end{equation} 
\item The pair $\lbrace \chi^+(x,\la),\chi^-(x,\la) \rbrace$ is a $\mathcal{O}^*_{\la}-$basis of $\mathcal{A}$.
The Wronskian (with respect to $x$) is given by
\begin{equation}\label{eq:chipmwronskian}
 \on{Wr}(\chi^+,\chi^-)=-(2l+1).
\end{equation}
\item
For $\theta \in \R$  the rotated solutions  $\chi^\pm_{[\theta]}$ (see \eqref{psi theta}) are given by
\ben\label{rotated chi}
\chi^\pm_{[\theta]}(x,\la)=e^{\pm(2l+1)\pi i\theta} x^{\frac12\pm(l+\frac12)} \,  \sum_{j=0}^{\infty}g^{\pm}_j(\la^2 x^{k+2})e^{2\pi i\theta j}\, x^j, \qquad x\to 0.
\een
The rotated solutions $\chi^\pm_{[-\frac12]}$ and $\chi^\pm_{[\frac12]}$ satisfy:
\begin{align}
 \on{Wr}(\chi^\pm_{[-\frac12]},\chi^\pm_{[\frac12]})&=0,\label{wrchipmpm}\\
  \on{Wr}(\chi^\pm_{[-\frac12]},\chi^\mp_{[\frac12]})&=\mp(2l+1)e^{\mp(2l+1)\pi i}.\label{wrchipmmp}
\end{align}
\item With respect to the basis $\lbrace \chi^+(x,\la),\chi^-(x,\la) \rbrace$ the monodromy matrix \eqref{monodromy matrix} is diagonal
\ben\label{monodromy matrix chi}
\mathcal{M}
\begin{pmatrix}
\chi^+\\
\chi^-
\end{pmatrix}
=\begin{pmatrix}
e^{(l+1)2\pi i} & 0\\
0 & e^{-l2\pi i}
\end{pmatrix}
\begin{pmatrix}
\chi^+ \\
\chi^- 
\end{pmatrix}.
\een
In other words, $\chi^\pm$ are eigenfunctions of the monodromy operator (\ref{eq:monodromyoperator}). 
\item The functions $\chi^\pm$ are invariant under the map \eqref{e pi i la}, namely 
\ben\label{second map chi}
\chi^{\pm}(x,e^{\pi i}\la)=\chi^\pm(x,\la),
\een 
and satisfy
\ben\label{third map chi}
\chi^{\pm}(x,e^{(k+3)\pi i}\la)=-e^{\mp(2l+1)\pi i}\chi^\pm(e^{2\pi i}x,\la).
\een
\end{enumerate}
\end{cor}
\begin{proof}
\begin{enumerate}[i)]
\item The estimate \eqref{eq:recursionconvergence+} 
implies that the series (\ref{eq:chipmx})  defines a  function $\chi^+$ holomorphic in the domain $\lbrace x \in \widetilde{\C^*}, |x| \leq \min_i |s_i| \rbrace \times \C$.   Since \eqref{diffeq:L1s} has trivial monodromy at the singular points $s_i, i=1,\dots,d_0$,  for every $\la$, it follows that $\chi^+$ can be analytically extended to the domain $\left(\widetilde{\C^*} \setminus \bigcup_{i=1}^{d_0} \Pi^{-1} (s_i) \right) \times \widetilde{\C^*}$. By construction, $\chi^+$ satisfies equation \eqref{diffeq:L1s} and therefore, $\chi^+\in \mc{A}$. The proof for $\chi^-$ is similar, with the only difference that the estimate \eqref{eq:recursionconvergence-} holds under the condition $2l+1 \notin \lbrace (k+2) i+j, (i,j)\in \mathbb{Z}_{\geq 0}^2 \rbrace$.
\item Same as in part i), replacing $l\to -l-1$.
\item The Wronskian (in $x$) of $\chi^+$ and $\chi^-$ is constant (in $x$). Substituting the the asymptotic expansion \eqref{eq:chipmx} into the Wronskian and using $g_0^{\pm}(0)=1$, we get \eqref{eq:chipmwronskian}. In particular, $\chi^+$ and $\chi^-$ are linearly independent for every $\la$.
\item The rotated solutions \eqref{rotated chi} are obtained applying \eqref{psi theta} to \eqref{eq:chipmx} and \eqref{eq:chipmx}. The proof of  \eqref{wrchipmpm}  and  \eqref{wrchipmmp}  is then similar to part ii).
\item The identity \eqref{monodromy matrix chi} follows easily from \eqref{rotated chi} with $\theta=1$.  
\item The identities \eqref{second map chi} follow again from \eqref{eq:chipmx}, \eqref{eq:chipmx2}, and from \eqref{monodromy matrix chi} we obtain \eqref{third map chi}:
$$\chi^\pm(e^{2\pi i}x,\la)=e^{(1\pm(2l+1))\pi i}\chi^\pm(x,e^{(k+2)\pi i}\la)=e^{(1\pm(2l+1))\pi i}\chi^\pm(x,e^{(k+3)\pi i}\la),$$
where in the last equality we used \eqref{second map chi}.
\end{enumerate}
\end{proof}

\subsection{Sibuya solutions at \texorpdfstring{$x=\infty$}{x=infinity}}
Since $k>-2$, the point $x=\infty$ is  an irregular singularity for equation (\ref{diffeq:L1s}). In the asymptotic analysis of equations with irregular singularities a prominent role is played by a primitive of the square root
of the potential
\begin{equation}\label{eq:Vpot}
V(x,\la)=\lambda^2x^k(x-1)+\frac{l(l+1)}{x^2}+\sum_{i=1}^{d_0}\left(\frac{2}{(x-s_i)^2}+\frac{k+s_i/(s_i-1)}{x(x-s_i)}\right),
\end{equation}
 of the operator (\ref{diffeq:L1sop}). In order to understand the asymptotic behaviour of solutions about $x=\infty$,
we need to find an approximate primitive of $\sqrt{V}$, namely a function $F(x,\la)$ that satisfies
the asymptotics
$\partial_x F(x,\la) - \big( V(x,\la) \big)^{\frac12}  = O(x^{-1-\e})$ as $ x \to \infty$, for
for some $\e>0$.
A convenient solution of the above inequality is given as follows.
Denote by $c_j$ ($j\geq 0$) the coefficients of the Taylor expansion of $\sqrt{1-x^{-1}}$ at $x=\infty$, and set  $a=\frac{k+3}{2}$.  Then define
\begin{equation}\label{eq:truncatedaction}
 R(x)=\begin{cases}
 \frac{1}{a}x^{a} \left( 1+ \sum_{j=1}^{\floori{a}}
\frac{c_j\, x^{-j}}{1-\frac{j}{a}}\right),\quad & a \notin \mathbb{Z}_{\geq 0},\\
 \frac{1}{a}  x^{a} \left( 1+ \sum_{j=1}^{a-1} \frac{c_j\, x^{-j}}{1-\frac{j}{a}}\right)+ c_a \log x , \qquad & a \in \mathbb{Z}_{\geq 0}.
\end{cases}
\end{equation}
Clearly $R(x)$ is holomorphic in the domain $\widetilde{\C^*} $
and $$ R'(z)-x^{a-1}\sqrt{1-x^{-1}}=O(x^{a-\floori{a}-2}),\qquad x \to \infty. $$
Since $a>1/2$, as $x\to \infty$, we get 
\begin{align}\label{eq:calS-S}
\la R'-V^{1/2}=\la R'-\la x^{a-1}(\sqrt{1-x^{-1}}+O(x^{-2a}))=\notag \\ O(x^{-2+\{a\}})+O(x^{-a-1})=O(x^{-2+\{a\}}),
\end{align}
where $\{a\}=a-\floori{a}\in [0,1)$ is the fractional part of $a$. This inequality holds uniformly in $\la$ in any compact subset of $\widetilde{\C^*}$. 
The Stokes sectors for equation (\ref{diffeq:L1s}) are obtained by looking at the dominant term in $\la R(x)$, which is given by
$\frac{2 \, \la}{k+3} x^{\frac{k+3}{2}} $.  More precisely,  for given $\lambda\in\widetilde{\C^*}$, define the Stokes rays  by the equation $\Re(\Pi\la x^{\frac{k+3}{2}})=0$, for variable $x\in \widetilde{\C^*}$. Further, define the Stokes sectors as the sectors between two consecutive Stokes rays. The Stokes sectors are parameterized by $j \in\Z$ and are explicitly given by
$$
 \Sigma_j[\la]= \left\{ x \in \widetilde{\C^*}: \left| \arg \la + \frac{k+3}{2}\arg x - \pi j  \right| \leq \frac{\pi}{2} \right\},\qquad j\in\mathbb{Z}.
$$
We notice that for $\theta\in \mathbb{R}$ and $j\in \mathbb{Z}$ the $j$-th Stokes sector satisfy
\begin{equation}\label{eq:twistedStokes}
\Sigma_j[e^{i \pi \theta}\la] = e^{-\frac{2 \pi i \theta}{k+3}} \Sigma_j[\la],
\end{equation}
and
\begin{equation}\label{eq:twistedStokes2}
\Sigma_{j+1}[\la]=e^{\frac{2 \pi i}{k+3}} \Sigma_j[\la]=\Sigma_j[e^{-i \pi}\la].
\end{equation}
Given the asymptotic relation (\ref{eq:calS-S}), we can proceed with the standard complex WKB analysis in order to define Sibuya solutions.
\begin{prop}\label{prop:sibuya}
\hspace{2em}
\begin{enumerate}[i)]
\item For every $j \in \mathbb{Z}$  there exists a unique solution of \eqref{diffeq:L1s}
$\psi^{(j)}\in\mathcal{A}$, which we call the $j$-th Sibuya solution, such that
for every $\la \in \widetilde{\C^*}$:
\begin{align}\label{eq:sibuyatrisectors}
& \lim_{x \to \infty }x^{\frac{k+1}{4}} e^{(-1)^j \la R(x)}
\psi^{(j)}(x,\la)=1   ,\\ \label{eq:sibuyatrisectorsder}
& \lim_{x \to \infty }  x^{-\frac{k+1}{4}} e^{(-1)^j \la R(x)}
\big(\partial_x \psi^{(j)}(x,\la)\big)=(-1)^{j+1}\la,
\end{align}
in any proper subsector of the sector $\Sigma_{j-1}[\la] \cup
\Sigma_j[\la] \cup \Sigma_{j+1}[\la]$, namely on every sector of the
form
$$\left|\arg\la+\frac{k+3}{2}\arg x - \pi j \right| \leq \frac{3
\pi}{2}-\e,\qquad \e>0.$$
The $j$-th Sibuya solution vanishes exponentially fast (it is
subdominant) as $x\to\infty$ in any proper subsector of $\Sigma_j[\la]$, and diverges
exponentially  fast (it is dominant) as $x\to\infty$ in any proper subsector of the
sectors $\Sigma_{j\pm1}[\la]$.
\item For every $j\in\Z$, the pair  $\lbrace \psi^{(j)},\psi^{(j+1)}\rbrace$ is a $\mathcal{O}^*_{\la}-$basis of $\mc{A}$ and
\begin{equation}\label{wronsian j j+1}
 \on{Wr}(\psi^{(j)},\psi^{(j+1)})=(-1)^j 2\lambda.
\end{equation}
\end{enumerate}
\end{prop}
\begin{proof}
\begin{enumerate}[i)]
\item We prove the existence of the solutions $\psi^{(j)}$. Let  $R$ be the function defined in (\ref{eq:truncatedaction}) and let $D_{\la}$ be a punctured disc, whose radius may depend on $\la$, and chosen such
that the function
$$
  y:\widetilde{\C^{*}}\setminus \Pi^{-1}(D_{\la}) \to
\widetilde{\C^{*}}, \qquad y(x)=\la R(x)
$$
 is a biholomorphism (onto its image). We consider the above map as a ($\la$-dependent) local change of variable $x\mapsto y=\la R(x)$.
 The function 
 $$\psi(x,\la)= R'(x)^{-\frac12} \phi(\la R(x),\la)$$ 
 is a solution of
(\ref{diffeq:L1s}) if and only if $\phi$ satisfies the following ODE
\ben\label{eq:pertubedexp}
\partial^2_y\phi(y,\la)=\left(1+ v(y,\la)  \right)  \phi(y,\la) ,
 \een
 where $v(y,\la)$ is given by
 \ben\label{v(y,la)}
 v(\la R(x),\la)= \big(\la R'(x)\big)^{-2} \left( V(x,\la)-\big(\la R'(x)\big)^2 +
\frac{1}{2} \lbrace R,x \rbrace \right) 
 \een
 for every $x\in \widetilde{\C^{*}}\setminus \pi^{-1}(D_{\la})$.  In the above equation $V$ is the potential \eqref{eq:Vpot} and $\lbrace R , x \rbrace$ denotes the Schwarzian derivative. Applying equation (\ref{eq:calS-S}) to \eqref{v(y,la)} we deduce  for every $\la\in\widetilde{\C^*}$ the asymptotic relation
$$
  v(y,\la)=O\big(y^{-1-\frac{\left\{a\right\}+1}{a}}\big),
\qquad  y \to \infty, 
$$
 where $a=\frac{k+3}{2}$  holds uniformly in $\la$ on compact subsets of $\widetilde{\C^*}$.
 It is convenient to write
 \eqref{eq:pertubedexp} in the matrix form
 \begin{equation}
 \partial_y\underline{\phi}(y,\la)= \begin{pmatrix} 0 & 1 \\ 1+ v(y) & 0 \end{pmatrix} \underline{\phi}(y,\la), \label{eq:ymatricial}
  \end{equation}
 with 
 $$ \underline{\phi}(y,\la)=(\phi(y,\la),\partial_y \phi(y,\la))^T,$$
and where in the above equation ${}^T$ denotes the transposition. In fact, the constant gauge transformation 
$$
 \underline{\phi}=
  \begin{pmatrix} 1 & -1 \\ -1 & -1 \end{pmatrix} \underline{\varphi}
$$
reduces \eqref{eq:ymatricial}
to the a linear ODE with a constant diagonal term plus an $L^1$-integrable perturbation
\begin{equation}\label{eq:CGM}
 \partial_y\underline{\varphi}= \left( \begin{pmatrix} -1 & 0 \\ 0 & 1 \end{pmatrix} +
 O\big(y^{-1-\frac{\left\{a\right\}+1}{a}}\big)\right) \underline{\varphi}.
\end{equation}
Equations of the above kind were studied in \cite[Theorem 3.21]{CGM}, where the following result is proved: if we let $\la$ vary in an open subset $A \subset \widetilde{\C^*}$ whose closure is compact, then for any $j \in \mathbb{Z}$ and any $\e>0$ equation \eqref{eq:CGM} admits a (unique) solution satisfying the following asymptotics
$$
  \underline{\varphi}^{(j)}(y,\la)= e^{(-1)^{j+1}y} \left( (1,0)^T + O\big(y^{-\frac{\left\{a\right\}+1}{a}}\big) \right),
$$
as $y \to +\infty$ in the sector 
$$\lbrace |\arg y -j \pi | < \frac{3 \pi}{2}-\e \rbrace.$$ 
Moreover, the remainder is uniform in $\la$, and
$\underline{\varphi}^{(j)}(y,\la)$ is analytic with respect to $y$ and $\la$ on its domain; this can be taken to be in $\widetilde{\C^{*}}\setminus \Pi^{-1} \big(D_{A}\big) \times A$ with $D_{A}$ a punctured disc depending on $A$.
Let 
$$\phi^{(j)}(y,\la)= \underline{\varphi}^{(j)}_1(y,\la)-\underline{\varphi}^{(j)}_2(y,\la),$$ 
where $ \underline{\varphi}^{(j)}_{i}(y)$
is the $i$-th component of the vector solution $\underline{\varphi}^{(j)}$.
By \eqref{eq:CGM}, \eqref{eq:ymatricial}, $\phi^{(j)}(y,\la)$ is a solution of \eqref{eq:pertubedexp}.
Since  
$$y=\la R(x)=\la x^{\frac{k+3}{2}} \left( 1 + O\big( x^{-1}
\big)\right)$$
as $x \to \infty,$
the solution
$$\psi^{(j)}(x,\la)=R'(x)^{-\frac12} \phi^{(j)}(\la R(x),\la)$$
of equation (\ref{diffeq:L1s}) satisfies
the asymptotic estimates \eqref{eq:sibuyatrisectors}, \eqref{eq:sibuyatrisectorsder}.
Thus we have constructed a solution  $\psi^{(j)}(x,\la)$ to equation (\ref{diffeq:L1s}), that satisfies the
asymptotic estimates \eqref{eq:sibuyatrisectors}, \eqref{eq:sibuyatrisectorsder} and that is analytic on a domain which contains the open sets
$  \widetilde{\C^{*}}\setminus \pi^{-1} \big(D_{A}\big) \times A$, where $A$ is an arbitrary open with compact closure and $D_A$ is a punctured disc depending on $A$.
Since for any fixed $\la$ the analytic continuation of $\psi^{(j)}(x,\la)$ is
single-valued, $\psi^{(j)}(x,\la)$ can be extended
to an element of $\mc{A}$.
To prove the uniqueness of the Sibuya solutions $\psi^{(j)}$, it is sufficient to notice that the solution $\psi^{(j)}$ (resp. $\psi^{(j+1)}$) is subdominant in $\Sigma_j[\la]$ (resp. $\Sigma_{j+1}[\la]$)
while $\psi^{(j+1)}$ (resp. $\psi^{(j)}$) is dominant  in $\Sigma_j[\la]$ (resp. $\Sigma_{j+1}[\la]$). Since $\{\psi^{(j)},\psi^{(j+1)}\}$ is a basis, every solution not proportional to $\psi^{(j)}$ (resp. $\psi^{(j+1)}$) is dominant in $\Sigma_j[\la]$ (resp. $\Sigma_{j+1}[\la]$).
\item Equation \eqref{wronsian j j+1} follows from the fact that the Wronskian is constant (with respect to $x$) and that, due to Proposition \ref{prop:sibuya}, the Sibuya solutions $\psi^{(j)}$, $\psi^{(j+1)}$ satisfy \eqref{eq:sibuyatrisectors} and \eqref{eq:sibuyatrisectorsder} in the sector $\Sigma_j[\la] \cup \Sigma_{j+1}[\la]$.
\end{enumerate}
 \end{proof}
 \begin{prop} Let $\psi^{(j)}$ ($j\in\Z$) be the Sibuya solutions described in Proposition \ref{prop:sibuya}.
\begin{enumerate}[i)]
\item  Applying the transformation \eqref{transformationxla1} to $\psi^{(j)}$ we have
\ben\label{rotated psi}
\psi^{(j)}_{[1]}(x,\la)=e^{-\frac{(k+3)\pi i}{2}}\psi^{(j-1)}(x,\la),\qquad j\in\Z,
\een
while under \eqref{e pi i la} we get
\ben\label{second map psi}
\psi^{(j)}(x,e^{\pi i}\la)=\psi^{(j-1)}(x,\la),\qquad j\in \Z.
\een
Moreover, the following identity holds:
\ben\label{rotated la psi}
\psi^{(j)}(x,e^{(k+3) \pi i} \la)= e^{\frac{(k+1)\pi i}{2}} \psi^{(j)}(e^{2\pi i} x,\la),\qquad j\in \Z.
\een
\item Every Sibuya solution $\psi^{(j)}$ is obtained from $\psi^{(0)}$ as
\ben\label{psi j psi 0}
\psi^{(j)}(x,\la)=e^{-\frac{j(k+3)\pi i}{2}}\psi^{(0)}_{[-j]}(x,\la)=\psi^{(0)}(x,e^{-j\pi i}\la),\qquad j\in\Z.
\een
\noindent
\item The rotated solutions $\psi^{(0)}_{[\pm  \frac{1}{2}]}$ form a basis of solutions of $\mc{A}_{\frac{1}{2}}$, and
\ben\label{wr pm 1/2}
\on{Wr}(\psi^{(0)}_{[- \frac{1}{2}]},\psi^{(0)}_{[\frac{1}{2}]})=-2e^{\frac{i\pi}{2}}\lambda.
\een
\end{enumerate}
\end{prop}

\begin{proof}
\hspace{2em}
\begin{enumerate}[i)]
\item To prove \eqref{rotated psi} and \eqref{second map psi}, we first notice that due to Lemma \ref{lemmaLA}, the functions  $\psi^{(j-1)}(x,\la)$, $\psi^{(j)}_{[1]}(x,\la)$ and $\psi^{(j)}(x,e^{\pi i}\la)$ satisfy the same differential equation. Moreover, by definition,  $\psi^{(j-1)}$ vanishes as $x\to \infty$ in the sector $\Sigma_{j-1}[\la]$, while using \eqref{eq:twistedStokes} and \eqref{eq:twistedStokes2} we get that $\psi^{(j)}_{[1]}$ vanishes as $x\to \infty$ in the sector 
$$e^{-2\pi i}\Sigma_j[e^{-(k+2)i \pi}\la] = e^{-2\pi i}e^{\frac{(k+2)2 \pi i }{k+3}} \Sigma_j[\la]=e^{-\frac{2 \pi i }{k+3}} \Sigma_j[\la]=\Sigma_{j-1}[\la],$$ 
and using \eqref{eq:twistedStokes2} we obtain that $\psi^{(j)}(x,e^{\pi i}\la)$ vanishes, as $x\to \infty$ in the sector 
$\Sigma_{j}[e^{i \pi}\la]=\Sigma_{j-1}[\la].$ It then follows that the three functions are proportional.  The proportionality factors are obtained comparing the asymptotic behaviour of the solutions in the sector $\Sigma_{j-1}[\la]$. Identity \eqref{rotated la psi} follows from \eqref{rotated psi} and \eqref{second map psi}.
\item It follows from \eqref{rotated psi} and \eqref{second map psi}.
\item Using \eqref{psi j psi 0} we get that $\psi^{(0)}_{[-\frac{1}{2}]}=\psi^{(0)}_{[-1+\frac{1}{2}]}=e^{\frac{(k+3)\pi i}{2}}\psi^{(1)}_{[\frac{1}{2}]}$, so that
\begin{align*}
\on{Wr}(\psi^{(0)}_{[- \frac{1}{2}]},\psi^{(0)}_{[\frac{1}{2}]})(x,\la)=e^{\frac{(k+3)\pi i}{2}}\on{Wr}(\psi^{(1)}_{ [\frac{1}{2}]},\psi^{(0)}_{[\frac{1}{2}]})(x,\la)\\ =e^{\frac{(k+3)\pi i}{2}}\on{Wr}(\psi^{(1)},\psi^{(0)})(e^{\pi i}x,e^{-\frac{(k+2)\pi i}{2}}\la),
\end{align*}
where in the last equality we used \eqref{wr theta}. The thesis then follows from \eqref{wronsian j j+1}.
\end{enumerate}
 \end{proof}
\begin{cor}
\begin{enumerate}[i)]
\item The matrix of change of basis among $\lbrace \psi^{(j+1)},\psi^{(j)}\rbrace$ and $\lbrace \psi^{(j-1)},\psi^{(j)}\rbrace$ is unipotent:
\begin{equation}\label{stokesmultipliers}
\begin{pmatrix}
 \psi^{(j+1)}\\
  \psi^{(j)}
\end{pmatrix}
= 
\begin{pmatrix}
1 & \sigma_j\\
0 & 1 
\end{pmatrix}
\begin{pmatrix}
 \psi^{(j-1)}\\
  \psi^{(j)}
\end{pmatrix}
\end{equation}
for some $\sigma_j \in \mathcal{O}^*_{\la}$ called the $j$-th Stokes multiplier. 
\item With respect to the basis $\lbrace \psi^{(j+1)}(x,\la),\psi^{(j)}(x,\la) \rbrace$ the monodromy matrix \eqref{monodromy matrix} reads
\ben\label{monodromy matrix psi}
\mathcal{M}
\begin{pmatrix}
\psi^{(j+1)}\\
\psi^{(j)}
\end{pmatrix}
=
e^{-\frac{(k+1)\pi i}{2}}
\begin{pmatrix}
0 & 1\\
1 & -\sigma_j
\end{pmatrix}
\begin{pmatrix}
\psi^{(j+1)} \\
\psi^{(j)}
\end{pmatrix}.
\een
\item The Stokes multipliers $\sigma_j$ satisfy the identity
\ben\label{sigma j invariance}
\sigma_{j}(e^{(k+3)\pi i}\la)=\sigma_j(\la) ,\qquad j\in\Z.
\een
Moreover, all $\sigma_j$ can be obtained from $\sigma_0$ as
\ben\label{sigma j 2}
\sigma_j(\la)=\sigma_0(e^{-j\pi i}\la),\qquad j\in\Z.
\een
and can be written as
\begin{align}
\sigma_j(\lambda)&=
\frac{e^{-j(k+2)\pi i}}{2\la}\on{Wr}(\psi^{(0)}_{[-(j+1)]},\psi^{(0)}_{[-(j-1)]})(x,\la),\qquad j\in\Z.\label{sigma j wr2}
\end{align}
\end{enumerate}
\end{cor}
\begin{proof}
\begin{enumerate}[i)]
\item To prove \eqref{stokesmultipliers}, let $\psi^{(j+1)}=\sigma_j\psi^{(j)}+\tau_j\psi^{(j-1)}$, for some $\sigma_j,\tau_j\in\mc{O}^\ast_\la$. Since
$$\lim_{\substack{x\to \infty\\x\in\Sigma_j[\la]}}\frac{\psi^{(j)}}{\psi^{(j-1)}}=0,$$
then
$$\lim_{\substack{x\to \infty\\x\in\Sigma_j[\la]}}\frac{\psi^{(j+1)}}{\psi^{(j-1)}}=\tau_j,$$
and due to \eqref{eq:sibuyatrisectors} it follows that $\tau_j=1$.
\item Using \eqref{rotated psi} and \eqref{stokesmultipliers} we obtain:
\begin{align*}
\mathcal{M}
\begin{pmatrix}
\psi^{(j+1)}\\
\psi^{(j)}
\end{pmatrix}
&=
e^{\pi i}
\begin{pmatrix}
\psi_{[1]}^{(j+1)}\\
\psi_{[1]}^{(j)}
\end{pmatrix}
=
e^{-\frac{(k+1)\pi i}{2}}
\begin{pmatrix}
\psi^{(j)}\\
\psi^{(j-1)}
\end{pmatrix}
\\
&=
e^{-\frac{(k+1)\pi i}{2}}
\begin{pmatrix}
0 & 1\\
1 & -\sigma_j
\end{pmatrix}
\begin{pmatrix}
\psi^{(j+1)} \\
\psi^{(j)}
\end{pmatrix}.
\end{align*}
\item We first prove \eqref{sigma j 2}. Due to \eqref{stokesmultipliers} we know that
\ben\label{psi sigma j pf}
\psi^{(j+1)}(x,\la)=\psi^{(j-1)}(x,\la)+\sigma_j(\la) \psi^{(j)}(x,\la),
\een
so that 
$$\psi^{(j+1)}(x,e^{\pi i}\la)=\psi^{(j-1)}(x,e^{\pi i}\la)+\sigma_j(e^{\pi i}\la) \psi^{(j)}(x,e^{\pi i}\la).$$  
Using \eqref{second map psi} the latter is equal to 
$$\psi^{(j)}(x,\la)=\psi^{(j-2)}(x,\la)+\sigma_j(e^{\pi i}\la) \psi^{(j-1)}(x,\la),$$ 
and therefore 
\ben\label{sigma j 2pf}
\sigma_j(e^{\pi i}\la)=\sigma_{j-1}(\la).
\een
Applying this procedure $j$ times we get \eqref{sigma j 2}. In a similar way, from \eqref{psi sigma j pf} we get 
$$\psi_{[1]}^{(j+1)}(x,\la)=\psi_{[1]}^{(j-1)}(x,\la)+\sigma_j(e^{-(k+2)\pi i}\la) \psi^{(j)}_{[1]}(x,\la),$$ 
and using \eqref{rotated psi} we obtain 
$$\psi^{(j)}(x,\la)=\psi^{(j-2)}(x,\la)+\sigma_j(e^{-(k+2)\pi i}\la) \psi^{(j-1)}(x,\la).$$ 
Therefore,
\ben\label{sigma j 1pf}
\sigma_j(e^{-(k+2)\pi i}\la)=\sigma_{j-1}(\la).
\een
Combining \eqref{sigma j 2pf} with \eqref{sigma j 1pf} we obtain \eqref{sigma j invariance}. To prove \eqref{sigma j wr2}, notice that on one hand from \eqref{psi sigma j pf} and using \eqref{wronsian j j+1}, we obtain 
$$\on{Wr}(\psi^{(j+1)},\psi^{(j-1)})=\sigma_j \on{Wr}(\psi^{(j+1)},\psi^{(j)})=(-1)^{j+1}2\lambda \sigma_j,$$ 
on the other hand due to \eqref{psi j psi 0}, we get 
$$\on{Wr}(\psi^{(j+1)},\psi^{(j-1)})=e^{-j(k+3)\pi i}\on{Wr}(\psi^{(0)}_{[-(j+1)]},\psi^{(0)}_{[-(j-1)]}).$$
\end{enumerate}
\end{proof}

\subsection{\texorpdfstring{$QQ$}{QQ}-system and \texorpdfstring{$TQ$}{TQ}-system}
Let $\chi^\pm$ be the Frobenius solutions defined in Corollary \ref{cor:frobenius}, and $\psi^{(j)}$ ($j\in\Z$) be the Sibuya solutions defined in 
Proposition \ref{prop:sibuya}. We define 
\begin{align}
Q^*_\pm(\la)&=e^{-\frac{i\pi}{4}}(l+\tfrac12)^{-\frac12}\la^{-\frac{1}{2}\pm\frac{2l +2r+1}{k+3}}\on{Wr}(\psi^{(0)},\chi^\pm)(x,\la),\label{Q*}\\
T^*_j(\la)&=\frac{e^{\frac{i\pi}{2}}}{2\lambda}\on{Wr}(\psi^{(0)}_{[-\frac{j+1}{2}]},\psi^{(0)}_{[\frac{j+1}{2}]})(x,\la),\qquad j\in\Z,\label{T*}
\end{align}
where $r=d_0-d_1$ and  where the normalization constants are chosen to simplify formulae in the $QQ$-system and $TQ$-system, to be introduced below. Note that due to \eqref{wr pm 1/2} and \eqref{sigma j wr2} (with $j=0$), we have
\begin{align}
T^*_0(\la)&=1,\notag\\
T^*_1(\la)&=e^{\frac{i\pi}{2}}\sigma_0(\la).\label{T1sigma0}
\end{align}
We now prove the $QQ$-system.

\begin{thm}\label{thm:QQsystem}
i)  Let $Q^*_\pm$ be given by \eqref{Q*}. Then, $Q^*_\pm\in \mathcal{O}^*_{\la}$ and the identity
\ben\label{Q* invariance}
Q^*_{\pm}(e^{(k+3)\pi i}\la)=Q^*_{\pm}(\la)
\een
holds for every $\la\in \widetilde{\C^*} $. It follows that the functions
\ben\label{Q Q*}
Q_\pm(\lambda)=Q^*_\pm(\la^{\frac{k+3}{2}}),\qquad \la\in \widetilde{\C^*},
\een
are holomorphic on $\C^*$. \\
ii)  The functions $Q_\pm$ satisfy the following $QQ$-system (or quantum Wronskian condition  \cite{BLZ4}):
\ben\label{QQsystem}
\gamma_{r} Q_+(q\la)Q_-(q^{-1}\la)-\gamma_{r}^{-1}Q_+(q^{-1}\la)Q_-(q\la)=1,\qquad \la\in\C^*,
\een
where
\ben\label{gammaQQ}
q=e^{\frac{k+2}{k+3}\pi i},\qquad \gamma_{r}=e^{2 \frac{l+\frac{1}{2} -r (k+2)}{k+3}\pi i}.
\een
iii) The function $Q_+(\la)$ satisfies the following Bethe equations for Quantum KdV \cite{BLZ4}:
\begin{equation}
\label{eq:BAEQintro}
\gamma_{r}^{-2} \frac{Q_{+}\left(\la^*\, q^{-2}\right)}{Q_{+}\left(\la^* \,q^2 \right)} = -1
\end{equation}
for every $\la^* \neq 0$ such that  $Q_+(\la^*)=0$.
\end{thm}
\begin{proof}
\begin{enumerate}[i)]
\item Using \eqref{third map chi} and \eqref{rotated la psi} we get
$$\on{Wr}(\psi^{(0)},\chi^\pm)(x,e^{(k+3)\pi i}\la)=e^{\frac{(k+3)\pi i}{2}}e^{\pm(2l+1)\pi i}\on{Wr}(\psi^{(0)},\chi^\pm)(x,\la),$$
from which \eqref{Q* invariance} follows, and the functions $Q_\pm$ defined by  \eqref{Q Q*} are single-valued on $\C^\ast$.
\item Expand the function $\psi^{(0)}$ in the basis $\{\chi^+,\chi^-\}$ as
\ben\label{psi0AA}
\psi^{(0)}=A_-\chi^++A_+\chi^- 
\een
for some $A_\pm\in \mathcal{O}^*_{\la}$. Using \eqref{rotated chi} we obtain
\begin{align*}
\psi^{(0)}_{[-\frac12]}(x,\la)&=A_-(e^{\frac{(k+2)\pi i}{2}}\la)\chi^+_{[-\frac12]}(x,\la)+A_+(e^{\frac{(k+2)\pi i}{2}}\la)\chi^-_{[-\frac12]}(x,\la),\\
\psi^{(0)}_{[\frac12]}(x,\la)&=A_-(e^{-\frac{(k+2)\pi i}{2}}\la)\chi^+_{[\frac12]}(x,\la)+A_+(e^{-\frac{(k+2)\pi i}{2}}\la)\chi^-_{[\frac12]}(x,\la).
\end{align*}
Substituting this into \eqref{wr pm 1/2} and using \eqref{wrchipmpm} and \eqref{wrchipmmp} we get
\begin{align}\label{AAsystem}
&e^{-(2l+1)\pi i}A_-(e^{\frac{(k+2)\pi i}{2}}\la)A_+(e^{-\frac{(k+2)\pi i}{2}}\la)-\notag\\ 
& \hspace{30pt} -e^{(2l+1)\pi i}A_+(e^{\frac{(k+2)\pi i}{2}}\la)A_-(e^{-\frac{(k+2)\pi i}{2}}\la)=\frac{2e^{\frac{i\pi}{2}}\lambda}{2l+1}.
\end{align}
From the definition \eqref{Q*} of $Q^*_\pm$ and using \eqref{eq:chipmwronskian} we get the relations
\ben\label{AQ*}
A_\pm(\la)=\pm e^{\frac{i\pi}{4}}\sqrt{2}(2l+1)^{-1/2}\la^{\frac{1}{2}\mp\frac{2l +2r+1}{k+3}}Q^*_\pm(\la),
\een
from which \eqref{AAsystem} takes the form
\begin{align}
\gamma_{r}Q^*_+(e^{\frac{(k+2)\pi i}{2}}\la)Q^*_-(e^{-\frac{(k+2)\pi i}{2}}\la)-\notag\hspace{60pt}\\
-\gamma_{r}^{-1}Q^*_+(e^{-\frac{(k+2)\pi i}{2}}\la)Q^*_-(e^{\frac{(k+2)\pi i}{2}}\la)=1.\label{preQQ}
\end{align}
Finally notice that from \eqref{Q Q*} we get
\begin{equation}\label{QqQ*}
Q^*_\pm(e^{\frac{j(k+2)\pi i}{2}}\la^{\frac{k+3}{2}})=Q_\pm(q^j\la),\qquad j\in\Z.
\end{equation}
Substituting $\la\mapsto\la^{\frac{k+3}{2}}$ in \eqref{preQQ} and using \eqref{QqQ*} we obtain \eqref{QQsystem}.
\item
Let $\la^\ast$ be such that $Q_+(\la^\ast)=0$. Evaluating the $QQ$-system \eqref{QQsystem} at $\la=q^{\pm1}\la^\ast$ we obtain \eqref{eq:BAEQintro}.
\end{enumerate}
\end{proof}
We call functions $Q_{\pm}$ the $Q$-functions or the spectral determinants of the operator $L^G$, see \eqref{diffeq:L1sop}.

We now prove the $TQ$-system.

\begin{thm}\label{thm:TQ}
\hspace{2em}
\begin{enumerate}[i)]
\item  For $j\in\Z$ let $T^*_j$ be given by \eqref{T*}. Then, $T^*_j\in \mathcal{O}^*_{\la}$ and the identity
\ben\label{T* invariance}
T^*_j(e^{(k+3)\pi i}\la)=T^*_j(\la)
\een
holds for every $\la\in \widetilde{\C^*} $.  It follows that the functions
$$T_j(\lambda)=T^*_j(\la^{\frac{k+3}{2}}),\qquad \la\in \widetilde{\C^*},\qquad j\in\Z,$$
are holomorphic on $\C^*$. The functions $T_j$ can be written in terms of the $Q_\pm$ introduced in \eqref{Q Q*} as
\ben\label{TfromQ}
T_{j-1}(\la)=\gamma_{r}^jQ_+(q^j\la)Q_-(q^{-j}\la)-\gamma_{r}^{-j}Q_+(q^{-j}\la)Q_-(q^{j}\la),
\een
with $\la\in\C^*$ and where $q$ and $\gamma$ are given by  \eqref{gammaQQ}.
\item The functions $T_1$ and $Q_\pm$ satisfy the following \emph{$TQ$-system}:
\ben\label{TQsystem}
T_1(\la)Q_\pm(\la)=\gamma_{r}^{\pm1}Q_\pm(q^2\la)+\gamma_{r}^{\mp1}Q_\pm(q^{-2}\la), \qquad \la\in\C^*.
\een
\item The functions $T_j$ satisfy the following \emph{fusion relations}:
\ben\label{fusionrelations}
T_1(\la)T_j(q^{j+1}\la)=T_{j-1}(q^{j+2}\la)+T_{j+1}(q^{j}\la),\qquad \la\in\C^*.
\een
\end{enumerate}
\end{thm}
\begin{proof}
\begin{enumerate}[i)]
\item The proof of \eqref{T* invariance} and the proof of existence of the functions $T_j$ are similar to the proofs of analogues results for $Q^*_\pm$ in Theorem \ref{thm:QQsystem}.  To prove \eqref{TfromQ} one starts from \eqref{T*} and follows the same steps as in the proof of the $QQ$-system \eqref{QQsystem} which is $j=1$ case of \eqref{TfromQ} .
\item Using \eqref{stokesmultipliers} and \eqref{psi j psi 0} one obtains the relation 
$$\sigma_0(\lambda)\psi^{(0)}(x,\la)=\psi^{(0)}(x,e^{-\pi i}\la)-\psi^{(0)}(x,e^{\pi i}\la).$$ 
Expanding $\psi^{(0)}$ as in \eqref{psi0AA} we obtain 
$$\sigma_0(\la)A_\pm(\la)=A_\pm(e^{-\pi i}\la)-A_\pm(e^{\pi i}\la),$$ 
and using \eqref{AQ*} and \eqref{T1sigma0} we get
$$T^*_1(\la)Q^*_\pm(\la)=\gamma_{r}^{\pm 1}Q^*_\pm(e^{-\pi i}\la)+\gamma_{r}^{\mp 1}Q^*_\pm(e^{\pi i}\la), \qquad \la\in\widetilde{\C^*}.$$
This identity is invariant under the map $\la\mapsto e^{(k+3)\pi i}\la$. Due to \eqref{Q* invariance} we have
$Q^*_{\pm}(e^{-j\pi i}\la)=Q^*_{\pm}(e^{j(k+2)\pi i}\la)$ ($j\in\Z$), and using \eqref{QqQ*} we get the $TQ$-system \eqref{TQsystem}.
\item Relations  \eqref{fusionrelations} follow from \eqref{TfromQ} and \eqref{TQsystem}.
\end{enumerate}
\end{proof}

\begin{rem} Recall that the monodromy matrix is given by \eqref{monodromy matrix chi} in the basis $\{\chi^+,\chi^-\}$, and by \eqref{monodromy matrix psi}  (with $j=0$) in the basis $\{\psi^{(1)},\psi^{(0)}\}$. Let $B(\la)$ be the change of basis matrix defined by
$$
\begin{pmatrix}
\psi^{(1)}\\
\psi^{(0)}
\end{pmatrix}
=B(\la)
\begin{pmatrix}
\chi^+\\
\chi^-
\end{pmatrix}.
$$
Due to \eqref{Q*} and \eqref{psi j psi 0}, the coefficients of $B(\la)$ can be written in terms of $Q^*_\pm$ (and therefore in terms of $Q_\pm$).
Under the above change of basis, the monodromy matrix transforms as 
$$B(e^{(-k+2)\pi i}\la)\begin{pmatrix}
e^{(l+1)2\pi i} & 0\\
0 & e^{-l2\pi i}
\end{pmatrix}
=
e^{-\frac{(k+1)\pi i}{2}}
\begin{pmatrix}
0 & 1\\
1 & -\sigma_0
\end{pmatrix}
B(\la).
$$
Writing the above relations explicitly and using \eqref{T1sigma0} one obtains the $TQ$-system  \eqref{TQsystem}. \qed
\end{rem}

\subsection{Action of the reproduction on $Q$-functions}\label{rep sec}
Let $(y_0,y_1)$ be a generic solution to the Gaudin BAE \eqref{Gaudin BAE BLZ}, and let
$$L_1(x,\la;l,y_0)=\partial_x^2-V(x,\la;l,y_0)$$
be the corresponding Schroedinger operator for the form \eqref{L1 BLZ}, where as in Subsection \ref{sec:reproductiononopers} we omit the dependence on $k=2m+2l$, since this quantity is invariant under reproduction. We first consider reproduction in the $0$-th direction. Recall from \eqref{eq:R0L1} that the action on $L_1$ of the reproduction in the $0$-th direction is defined as
$$R_0[L_1(x,\la;l,y_0)]=\partial_x^2-V(x,\la;R_0[l],R_0[y_0]).$$
We want to compare $\ker L_1$ with $\ker R_0[L_1]$. Recall the gauge $B_0$ given by \eqref{eq:B0}, and  define
\begin{equation}
\psi_{B_0}=\left(1+\frac{H_0a_-}{P\lambda^2}\right)\psi-\frac{H_0}{P\lambda^2}\psi',
\end{equation}
where $H_0$ is given in \eqref{eq:H0rep}.
\begin{lem} \label{lem:reproonfunctions}
\begin{enumerate}
\item $\psi\in\ker{L_1}$ if and only if $\psi_{B_0}\in\ker R_0[L_1]$.
\item If $\psi,\phi\in\ker L$ then
\begin{equation}\label{eq:psiB0}
\on{Wr}(\psi_{B_0},\phi_{B_0})=\on{Wr}(\psi,\phi).
\end{equation}
\end{enumerate}
\end{lem}
\begin{proof}
\begin{enumerate}
\item Let $\bar{L}_1$ be given by \eqref{bar L reproduction}, and let
\begin{equation}
\Psi=
\begin{pmatrix}
\psi_1\\\psi_2
\end{pmatrix}
\qquad 
\bar{L}_1\Psi=0.
\end{equation}
Using \eqref{R0bL1} we have that $\Psi\in\bar{L}_1$ if and only if $B_0\Psi\in R_0[\bar{L}_1]$. Now
\begin{equation}
B_0\Psi=
B_0
\begin{pmatrix}
\psi_1\\\psi_2
\end{pmatrix}
=
\begin{pmatrix}
\psi_1+\frac{H_0}{P\lambda}\psi_2\\\psi_2
\end{pmatrix},
\end{equation}
Due to \eqref{eq:L1toscalar} we have that $\psi_1\in \ker{L}$ and that $\lambda\psi_2=-\psi_1'+a_-\psi_1$, from which we obtain that
\begin{equation}
\psi_{B_0,1}=\left(1+\frac{H_0a_-}{P\lambda^2}\right)\psi_1-\frac{H_0}{P\lambda^2}\psi_1'\in\ker R_0[L_1].
\end{equation}
\item A direct computation, using the identity \eqref{H0 eq}.
\end{enumerate}
\end{proof}
We apply the above Lemma to the case of Frobenius solutions $\chi^\pm$ and the Sibuya solution $\psi^{(0)}$. For a generic solution $(y_0,y_1)$ to the Gaudin  BAE with parameters $l$ and $m$, let $\chi^\pm(x,\la;l,y_0)$ be the Frobenius solutions defined in  Corollary \ref{cor:frobenius}, and $\psi^{(0)}(x,\la;l,y_0)$ be the Sibuya solutions defined in  Proposition \ref{prop:sibuya}. 
\begin{prop}\label{prop:frobeniusrepro}
Assume $2l+1,2R_0[l]+1\notin \{(k+2)i+j,(i,j)\in\mathbb{Z}^2_{\geq 0}\}$.
\begin{enumerate}[i)]
\item Let $\chi^\pm_{B_0}\in\ker R_0[L_1]$ be defined as in \eqref{eq:psiB0}. Then:
\begin{equation}\label{frobeniusreproduction}
\chi^\pm_{B_0}(x,\la;l,y_0)=c_\mp\la^{\mp2}\chi^\mp(x,\la;R_0[l],R_0[y_0]),
\end{equation}
where $c_-=(k-2l+1)(2l+1)$ and $c_+=-[(k-2l+1)(k-l+2)]^{-1}$.
\item Let $\psi^{(0)}_{B_0}\in\ker R_0[L_1]$ be defined as in \eqref{eq:psiB0}. Then:
\begin{equation}\label{sibuyareproduction}
\psi^{(0)}_{B_0}(x,\la;l,y_0)=\psi^{(0)}(x,\la;R_0[l],R_0[y_0]).
\end{equation}
\end{enumerate}
\end{prop}

\begin{proof}
\begin{enumerate}[i)]
\item We have:
\begin{align*}
H_0(x)&=\ln'\left(\frac{x^{2m+1}R_0[y_0]}{y_0}\right)=\frac{2m+1}{x}+O(1),\qquad x\mapsto 0,\\
a_-(x)&=-\frac12 \ln'\frac{T_1(x;l)y_0^2}{y_1^2}=-\frac{l}{x}+O(1),\qquad x\mapsto 0,\\
\frac{1}{P(x)\lambda^2}&=\frac{1}{x^k(x-1)\lambda^2}=-\lambda^{-2}x^{-k}(1+O(x)),\qquad x\mapsto 0,
\end{align*}
from which it follows that
\begin{align*}
\chi_{B_0}^+&=\left(1+\frac{ha_-}{P\lambda^2}\right)\chi^+-\frac{h}{P\lambda^2}\partial_x\chi^+\\
&=(2m+1)\left(2l+1\right)\lambda^{-2}x^{l-k-1}(1+O(x))\\
&=c_-\lambda^{-2}x^{-R_0[l]}(1+O(x)).
\end{align*}
The functions $\chi_{B_0}^+(x,\la;l,y_0)$ and $\chi^\pm(x,\la;R_0[l],R_0[y_0])$ belong to $\ker R_0[L_1]$ and (up to the factor $c_-\lambda^{-2}$) they have the same asymptotic behaviour at $x=0$. By uniqueness of Frobenius solutions we obtain \eqref{frobeniusreproduction} for $\chi^+$.
To prove the identities for $\chi_{B_0}^-$, note that by Lemma \ref{lem:reproonfunctions} we have
\begin{equation*}
\on{Wr}(\chi^+_{B_0},\chi^-_{B_0})(x,\lambda)=\on{Wr}(\chi^+,\chi^-)(x,\lambda)=-(2l+1),
\end{equation*}
from which we get
\begin{align*}
\chi^-_{B_0}(x,\lambda;l,y_0)&=\frac{\lambda^2}{(2k-2l+3)(k-2l+1)}x^{k+2-l}(1+O(x))\\
&=c_+\lambda^2 x^{R_0[l]+1}(1+O(x)).
\end{align*}
Reasoning as above, we obtain \eqref{frobeniusreproduction} for $\chi^-$.
\item Both $\psi^{(0)}_{B_0}(x,\la;l,y_0)$ and $\psi^{(0)}(x,\la;R_0[l],R_0[y_0])$ belong to $\ker R_0[L_1]$ and they satisfy the conditions of Proposition \ref{prop:sibuya}. By uniqueness of Sibuya solutions, they coincide.
\end{enumerate}
\end{proof}
We now consider the $Q$-functions. By definition  we have
$$Q^*_\pm(\la;l,y_0)=e^{-\frac{i\pi}{4}}(l+\tfrac12)^{-\frac12}\la^{-\frac{1}{2}
\pm\frac{2l +2r +1}{k+3}}\on{Wr}(\psi^{(0)}(x,\la;l,y_0),\chi^\pm(x,\la;l,y_0)),$$
and
$$Q_\pm(\la;l,y_0)=Q^*_\pm(\la^{\frac{k+3}{2}};l,y_0),$$
see  \eqref{Q*} and \eqref{Q Q*}.

\begin{prop}\label{prop:R0Qpm}
The following identities holds:
\begin{align}
\label{eq:TjTjR0}
& T_j(\la;l,y_0)=T_j(\la;R_0[l],R_0[y_0]), \quad \forall j \in \mathbb{Z}, \\
\label{eq:QpmQmpR0}
& Q_{\pm}(\la;l,y_0)=d_\mp Q_{\mp}(\la;R_0[l],R_0[y_0]),
\end{align}
where
$d_\pm=(2l+1)(2k-2l+3)^{-1}c_\pm$ and $c_\pm$ are defined in Proposition \ref{prop:frobeniusrepro}.
\end{prop}
\begin{proof}
Identity \eqref{eq:TjTjR0} follows from  \eqref{sibuyareproduction}.

We now prove \eqref{eq:QpmQmpR0}. We have
\begin{align*}
\on{Wr}(\psi^{(0)}(x&,\la;l,y_0),\chi^\pm(x,\la;l,y_0))=\on{Wr}(\psi_{B_0}^{(0)}(x,\la;l,y_0),\chi_{B_0}^\pm(x,\la;l,y_0))\\
&=\on{Wr}(\psi^{(0)}(x,\la;R_0[l],R_0[y_0]),c_\mp \la^{\mp 2}\chi^\mp(x,\la;R_0[l],R_0[y_0]))\\
&=c_\mp \la^{\mp 2}\on{Wr}(\psi^{(0)}(x,\la;R_0[l],R_0[y_0]),\chi^\mp(x,\la;R_0[l],R_0[y_0]))\\
&=c_\mp \la^{\mp 2}c[R_0[l]]^{-1}\la^{\frac{1}{2}\pm\frac{2R_0[l] +2 R_0[r]+1}{k+3}}Q^*_\mp(\la;R_0[l],R_0[y_0]).
\end{align*}
Since $R_0[l] + R_0[r]=k+2-l-r$, and letting $c_l=e^{-\frac{i\pi}{4}}(l+\tfrac12)^{-\frac12}$, we get
\begin{align*}
Q^*_\pm(\la;l,y_0)&=c_l\la^{-\frac{1}{2}\pm\frac{2l +2 r +1}{k+3}}\on{Wr}(\psi^{(0)}(x,\la;l,y_0),\chi^\pm(x,\la;l,y_0))\\
&=\frac{c_lc_\mp}{c_{R_0[l]}} Q^*_\mp(\la;R_0[l],R_0[y_0])=d_\pm  Q^*_\mp(\la;R_0[l],R_0[y_0]).
\end{align*}
\end{proof}

For the reproduction in the $1$-st direction the computations are much simpler, we give the results omitting the details.  Recall from \eqref{eq:R1L1} that the action on $L_1$ of the reproduction in the $1$-th direction leaves the oper $L_1$ invariant
$$R_1[L_1(x,\la;l,y_0)]=\partial_x^2-V(x,\la;R_1[l],y_0)=\partial_x^2-V(x,\la;R_1[l],y_0),$$
as it is tantamount to a change of the parameter $l\mapsto R_1[l]=-l-1$. Consequently, for the Frobenius and Sibuya solutions we have
\begin{equation}\label{frobeniusreproductionR1}
\chi^\pm(x,\la;l,y_0)=\chi^\mp(x,\la;R_1[l],y_0),
\end{equation}
\begin{equation}\label{sibuyareproductionR1}
\psi^{(0)}(x,\la;l,y_0)=\psi^{(0)}(x,\la;R_1[l],y_0).
\end{equation}
For the $Q$-functions, we have
\begin{prop}\label{prop:R1Qpm}
The following identities hold:
\begin{align}
\label{eq:TjTjR1}
& T_j(\la;l,y_0)=T_j(\la;R_1[l],y_0), \quad \forall j \in \mathbb{Z}, \\
\label{eq:QpmQmpR1}
& Q_{\pm}(\la;l,y_0)= Q_{\mp}(\la;R_1[l],y_0).
\end{align}
\qed
\end{prop}
Thus, Propositions \ref{prop:R0Qpm} and \ref{prop:R1Qpm} assert that,
essentially, it is sufficient to study $Q$-functions in the $r=0$ sector.

\subsection{\texorpdfstring{$QQ$}{QQ}-system for BLZ opers}
The $QQ$-system \eqref{QQsystem} and the $TQ$-system \eqref{TQsystem} are well known in the literature about
BLZ opers \eqref{BLZ oper}. Indeed, the discovery that these opers provide solutions to the $QQ$ and $TQ$-systems  marked the starting point of the ODE/IM correspondence \cite{DTa, DT, BLZ4}.
For the sake of being able to compare $L^G$opers and BLZ opers, we review briefly below
the $QQ$-system for BLZ opers; the reader must refer to \cite{DTa, DT,BLZ4} or to the review paper \cite{DDT} for all details, including the definitions of the determinants
$\overline{T}_j, j \geq 1$.
One considers the differential equation $\bar{L}\psi=0$, where $\bar{L}$ is given by \eqref{BLZ oper} with $z_i$  solutions to \eqref{BLZ BAE}.
The Sibuya solution $\bar{\psi}^{(0)}(z,\bar{\la})$ and the Frobenius solutions $\bar{\chi}_{\pm}(z,\bar{\la})$ are defined\footnote{The first rigorous definition of such solutions, akin to the one we use in this paper, appeared in \cite{MR1}, well after the original papers.} by the asymptotics properties:
\begin{align*}
 &\bar{\psi}^{(0)}(z,\bar{\la}) = z^{\frac14} e^{- 2 z^{\frac12}+c(z;\bar{\la})} \left(1 + o(1)\right),\qquad &z \to +\infty,\\ 
 &\bar{\chi}_{\pm}(z,\bar{\la}) = z^{\frac12 \pm \left(\bar{l}+\frac12\right)} \left(1 + O(z) \right), \qquad &z \to 0,
\end{align*}
where $c(z;\la)=O(z^{\bar{k}-\frac12})$ is a certain polynomial expression in the variable $z^{\bar{k}-1}$.
The spectral determinants
\begin{align}\label{bar Q wronskian}
 & \overline{Q}_{\pm}(\bar{\la})=\tfrac12 e^{-\frac{i\pi}{4}(\bar{l}+\tfrac12)^{-\frac12}}\on{Wr}(\bar{\psi}^{(0)}(z,\bar{\la}),\bar{\chi}_{\pm}(z,\bar{\la})),
\end{align}
are entire functions of $\bar{\la}$ and they
satisfy the $QQ$-system
\begin{equation}\label{eq:QQbarsystem}
\bar{\gamma} \, \overline{Q}_+(\bar{q}\bar{\la})\,\overline{Q}_-(\bar{q}^{-1}\bar{\la})-\bar{\gamma}^{-1}\,\overline{Q}_+(\bar{q}^{-1}\bar{\la})\,\overline{Q}_-(\bar{q}\bar{\la})=1,
\end{equation}
for  $\bar{\la} \in \C$ and where
\begin{equation}\label{eq:bargammabarq}
\bar{q}=e^{2 \pi i \bar{k}}, \quad \bar{\gamma}=e^{2\pi i (\bar{l}+\frac12)}.
\end{equation}
The above formula can be found e.g. in \cite[Equation (9)]{BLZ4}.
Taking into account the change of parameters \eqref{parameters change}, we get $\bar{q}=q$ and $\bar{\gamma}=\gamma$, and the relation \eqref{eq:QQbarsystem}
coincides with the $QQ$-system (\ref{QQsystem}) for the spectral determinants $Q_{\pm}$. In particular, the function $\overline{Q}_+$ satisfies the Quantum KdV  Bethe equations \eqref{eq:BAEQintro}:
$$\gamma^{-2} \frac{\overline{Q}_+\left(\la^*\, q^{-2}\right)}{\overline{Q}_+\left(\la^* \,q^2 \right)} = -1,$$
for every $\la^* \neq 0$ such that  $\overline{Q}_+(\la^*)=0$.

\begin{rem}
From \cite{BLZ4} we know that $\gamma=e^{2\pi i p}$, where $p$ the quantum KdV momentum. Due to \eqref{eq:bargammabarq} we therefore obtain  $p=\bar{l}+\frac12$ (cf. Remark \ref{rem:blz parameters}). The shift $\bar{l}\mapsto\bar{l}-r\bar{k}$ discussed in Remark \ref{rem: shift} induces the scaling $\gamma\mapsto \gamma q^{-2r}$. \qed
\end{rem}

Recall that if $d_0=0$, $L^G$ opers and BLZ opers are related by the change of variables \eqref{eq:phichange} and
of parameters \eqref{parameters change}.  
Since under the map $x=\varphi(z)$ the opers transform as \eqref{L change of variable}, the corresponding solutions transform as
$$\psi(x)\mapsto \varphi'(z)^{-\frac12}\psi(\varphi(z)).$$
If follows from this that under \eqref{eq:phichange},  \eqref{parameters change},  the Sibuya solutions are mapped to  the Sibuya solutions:
$$\psi^{(0)}(x,\la)\mapsto (1-\bar{k})^{-\frac12}\bar{\la}^{\frac{1}{4(1-\bar{k})}}\bar{\psi}^{(0)}(z,\bar{\la}).$$ 
 Similarly, the Frobenius solutions are mapped to the Frobenius solutions: 
 $$\chi_\pm(x,\la)\mapsto (1-\bar{k})^{-\frac12}\bar{\la}^{\mp\frac{\bar{l}+\frac12}{1-\bar{k}}}\bar{\chi}_\pm(z,\bar{\la}).$$
A straightforward computation shows that
\begin{align}\label{eq:QQbargroundstate}
& \overline{Q}_\pm(\bar{\la})=(\frac{1}{k+3})^{\pm(2\bar{l}+1)}Q_\pm((k+3)^{\frac{2}{k+3}} \, \bar{\la}),
\end{align}
where $Q_{\pm}$ are the spectral determinants \eqref{Q Q*} of $L^G$ with $d_0=0$, and $\bar{Q}_{\pm}$ are the spectral determinants \eqref{bar Q wronskian} of the BLZ oper $L^{BLZ}$  with $d_0=0$. As remarked before, the above transformation does not exist in the case $d_0 >0$.

\section{Studying solutions of BAE}\label{BAE limit chapter}
In this section we study the asymptotics of the solutions of the Gaudin BAE \eqref{Gaudin BAE} as $l\to\infty$. The first order of the asymptotics is controlled by the zeroes of some remarkable polynomials. These zeroes are apparent poles of monodromy-free differential operators \eqref{O osc}, \eqref{osc-1}.
\subsection {Monodromy-free extensions of the harmonic oscillator}
Consider differential operators of the form
\bean\label{O osc}
\partial_x^2-x^2-\sum_{i=1}^{d_0} \frac2{(x-v_i)^2}.
\end{align}
It is known that such operators with trivial monodromy correspond to partitions of $d_0$ as follows.
Let 
\be
H_n(x)=(-1)^n e^{x^2} \frac{d^n}{dx^n} e^{-x^2}
\ee
be the $n$-th Hermite polynomial. For a partition $\mu=(\mu_1,\dots,\mu_a)$, $\mu_1>\dots>\mu_a>0$, of $d_0$ let 
\be
V_\mu(x)=\on{Wr}\big(H_{\mu_a }(x),H_{\mu_{a-1}+1 }(x),\dots, H_{\mu_1+a-1 }(x)\big)
\ee
be the Wronskian of corresponding Hermite polynomials. Then $V_\mu(x)$ is a polynomial of degree $d_0$. Let $v^\mu_1,\dots,v^\mu_{d_0}$ be the roots of $V_\mu$:
\bean\label{Wr zero}
V_\mu=b_\mu\prod_{i=1}^{d_0}(x-v_i^{\mu}), \qquad b_\mu=2^{\sum_{i=1}^m (\mu_i+i-1)}\prod_{i<j}(\mu_i-\mu_j+i-j).
\end{align}
The first few polynomials are:
\bea
V_\emptyset=1, \qquad V_{(1)}=2x, \qquad V_{(2)}=2(2x^2-1),  \qquad V_{(1,1)}=4(2x^2+1),\\
V_{(3)} = 4(2x^3-3x),\qquad V_{(2,1)}=32x^3, \qquad  V_{(1,1,1)}=64x(2x^2+3).
\end{align*}
It is also good to note that an addition of zeroes to partition does not change zeroes of $V_\mu$, as the ratio
\be
\frac{V_{(\mu_1,\dots,\mu_a,0)}}{V_{(\mu_1,\dots,\mu_a)}}
\ee
is always a constant.
It is proved in \cite{O} that an operator of form \eqref{O osc} is monodromy-free if and only if the set of poles coincides with the set of zeroes of $V_\mu(x)$ for some $\mu$. That is  \eqref{O osc} is monodromy-free if and only if  $\{v_j\}_{j=1}^{d_0}=\{v_j^\mu\}_{j=1}^{d_0}$ for some partition $\mu$ of $d_0$.
A partition $\mu$ is called non-degenerate if all $v_j^\mu$, $j=1,\dots,d_0$, are distinct.
The zeroes of $V_\mu(x)$ have been studied in \cite{FHV}. Let
\be
m_\mu=\sum_{j=1}^a(-1)^{\mu_j+m-j},
\ee
then $V_\mu(x)$ has a zero at $x=0$ of order $m_\mu(m_\mu+1)/2$. It is conjectured that if $x=v$ is a multiple root of a $V_\mu(x)$ then $v=0$.
In particular, this implies that a partition $\mu$ is non-degenerate if and only if $V_\mu(x)$ does
not have a multiple root at $x=0$. In other words, one expects $\mu$ to be non-degenerate if and only if $-2\leq m_\mu\leq 1$.
For example, among $19$ partitions of numbers at most $5$ there are $3$ degenerate partitions: $(2,1)$, $(4,1)$, $(2,1,1,1)$, and $m_\mu=2$ for these three partitions.
It is also known that $\{v^\mu_j\}_{j=1}^{d_0}=\{-v^\mu_j\}_{j=1}^{d_0}=\{ i v^{\mu_*}_j\}_{j=1}^{d_0}$, where $\mu_*$ is the partition dual to $\mu$ (the Young diagrams of dual partitions $\mu$ and $\mu_*$ are transposed to each other).
Thus, monodromy-free operators of the form \eqref{O osc} with  distinct $v_i$ correspond to non-degenerate partitions of $d_0$. Then equation \eqref{monodromy eq} of trivial monodromy around $x=v_i$ takes the form
\be
2v_i-\sum_{j=1, j\neq i}^{d_0}\frac{4}{(v_i-v_j)^3}=0, \qquad i=1,\dots,d_0.
\ee
Thus for any complex $c,\tilde c,$ $c\neq 0,$ the algebraic system of $d_0$ equations with $d_0$ variables
\bean\label{O eq}
2c^4 (v_i-\tilde c)-\sum_{j=1, j\neq i}^{d_0}\frac{4}{(v_i-v_j)^3}=0, \qquad i=1,\dots,d_0,
\end{align}
has isolated solutions (up to permutations of $v_i$) given by $\{v_i\}_{i=1}^{d_0}=\{\tilde c+v_i^\mu/c\}_{i=1}^{d_0}$ for non-degenerate partitions $\mu$ of $d_0$.

\subsection{Large \texorpdfstring{$l$}{l} limit, the case of \texorpdfstring{$d_0=d_1$}{d0=d1}}

We start with the case of $d_0=d_1$, that is $r=0$.
For convenience we write $l$ as a square 
\be
l=\ell^2.
\ee

\begin{conj}\label{conj d0=d1}
For generic $l,m$ the number of solutions of the Gaudin BAE \eqref{Gaudin BAE BLZ}, up to permutations of coordinates equals $\on{Part}(d_0)$, the number of partitions of $d_0$, and all solutions are isolated and non-degenerate. For fixed $d_0$, there exists $C\in\R$ such that parameters $l,m$ are generic if $|l|,|m|>C$. 
Moreover, for generic $l, k$ and a partition $\mu$ of $d_0$, there exists a solution $\{s_i^\mu,t_i^\mu\}_{i=1}^{d_0}$  with the asymptotics
\begin{subequations}\label{st asymp}
\bean
&s_i^\mu=\frac{k+2}{k+3}\left(1+\frac{2^{1/4}}{(k+2)^{1/4}(k+3)^{1/4}}\,v^\mu_i\, \ell^{-1}+O(\ell^{-2})\right),\label{s asymp}\\
&t_i^\mu=s_i^\mu(1 -\ell^{-2})+O(\ell^{-4}),  \label{t asymp}
\end{align}
\end{subequations}
$i=1,\dots,d_0$, as $l\to\infty$ and $k$ is fixed. Here $v^\mu_i$ are zeroes of $V_\mu$, see \eqref{Wr zero}.\qed
\end{conj}

For special $l,m$ the number of isolated solutions (counted with multiplicity) is expected to be at most $\on{Part}(d_0)$. In some cases solutions appear in infinite families. 
For example, if $m$ or $l$ are non-negative integers, such families are produced by the reproduction procedure \eqref{rep 1} or \eqref{rep 2}
We prove the existence of solutions predicted by Conjecture \ref{conj d0=d1} for the case when $\mu$ is a non-degenerate partition of $d_0$.
\begin{prop}\label{old prop}
Let $\mu$ be a non-degenerate partition. Then there exists a solution $\{s_i^\mu, t_i^\mu\}_{i=1}^{d_0}$  of 
equations \eqref{Gaudin BAE BLZ}
 with asymptotics \eqref{st asymp}.
\end{prop}
\begin{proof}
Substitute
\ben\label{t sub}
t_j=s_j(1-\ell^{-2})+\tilde t_j\ell^{-4},\qquad j=1,\dots,d_0,
\een 
and then 
\ben\label{s sub}
s_j=\frac{k+2}{k+3}+\frac{v_j}{\ell}, \qquad j=1,\dots,d_0,
\een
to \eqref{Gaudin BAE BLZ}.
Write the resulting equation for $t_i$ as a Laurent series in $\ell$ at $\ell=\infty$. Then the top term is of order $1$  and coincides with the left hand side of:
\ben\label{tilde equation}
\frac{k+3}{k+2}-\frac{(k+3)^2}{(k+2)^2}\, w_i+\sum_{j=1,j\neq i}^{d_0}\frac{(k+2)}{(k+3)(v_i-v_j)^2}=0.
\een
Add the resulting equations for $s_i$ and $t_i$ (with fixed $i$) and write the answer as a Laurent series at $\ell=\infty$. Then the top term is of order $\ell^{-1}$  and coincides with the left hand side of:
\bean\label{order 0 system}
\frac12 (k+3)^3 v_i-\sum_{j=1,j\neq i}^{d_0}\frac{2(k+2)^2}{(k+3)^2(v_i-v_j)^3}=0.
\end{align}
Note that \eqref{order 0 system} is nothing but system \eqref{O eq} with
\ben\label{1/c}
c^{-1}=\frac{k+2}{k+3}\left(\frac{2}{(k+2)(k+3)}\right)^{1/4}, \qquad \tilde c=0.
\een
Thus, Gaudin BAE equations \eqref{Gaudin BAE BLZ}, for variables $v_j$, $\tilde t_j$ are deformations of equations \eqref{O eq} together with equations \eqref{tilde equation}, with respect to parameter $\ell^{-1}$.
Given an isolated solution $\{v_j^{(0)}\}_{j=1}^{d_0}$  of \eqref{O eq} with \eqref{1/c}, we compute $\tilde t_j^{(0)}$ using \eqref{tilde equation}.
Then for large $\ell$, the set $\{v_j^{(0)}, \tilde t_j^{(0)}\}_{j=1}^{d_0}$ can be deformed to $\{v_j, \tilde t_j\}_{j=1}^{d_0}$ such that $\{s_j,t_j\}_{j=1}^{d_0}$ given by \eqref{s sub}, \eqref{t sub} is a solution of 
 \eqref{Gaudin BAE BLZ}.
The  proposition follows. 
\end{proof}
Recall that for every solution $\{s_i\}_{i=1}^{d_0}, \{t_j\}_{j=1}^{d_1}$  of  Gaudin BAE \eqref{Gaudin BAE BLZ}  the operators \eqref{L0 BLZ} and \eqref{L1 BLZ} have trivial monodromy. The operator \eqref{L1 BLZ} with distinct $s_i$ has trivial monodromy if and only if $s_i$
satisfy equation \eqref{la monodromy eq}. Thus, Proposition \ref{old prop} shows the existence of solutions $\{s_i\}_{i=1}^{d_0}$  of 
equations \eqref{la monodromy eq}
 with asymptotics \eqref{s asymp} for all non-degenerate partitions $\mu$. 
Similar asymptotics for solutions $\{z_i\}_{i=1}^{d}$ of the BLZ equations \eqref{BLZ BAE} were proved in \cite{CM2}.

\subsection{Monodromy-free extensions of oscillator with potential \texorpdfstring{$1+\la x$}{1+lambda x}}
Consider the Gaudin BAE system \eqref{Gaudin BAE BLZ}. Substitute
\begin{align*}
s_i=1+ w_i\, l^{-1}, \qquad i=1,\dots,d_0, \\
t_j=1+u_j\,l^{-1}, \qquad j=1,\dots,d_1.
\end{align*}
Then the main term when $l\to \infty$ is of order $l$ and is given by the left hand sides of
\begin{subequations}\label{BAE exp}
\bean
&\frac{1}{w_i}+\sum_{j=1}^{d_1}\frac{2}{w_i-u_j}-\sum_{j=1, j\neq i}^{d_0}\frac{2}{(w_i-w_j)}-2=0,\label{BAE exp 1}\\
&\hspace{25pt}\sum_{i=1}^{d_0}\frac{2}{u_j-w_i}-\sum_{i=1, i\neq j}^{d_1}\frac{2}{(u_j-u_i)}+2=0.\label{BAE exp 2} 
\end{align}
\end{subequations}
These equations are exponential BAE \eqref{Gaudin exp BAE} for 
\be
{\mc L}=L_{1,0}
\ee
with the non-zero but special twist $n=(n_0,n_1)=(-1,1)$. In other words
\be
T_0=xe^{-2x}, \qquad T_1=e^{2x}, \qquad P=T_0T_1=x.
\ee 
Recall that if $n=(0,0)$, as in \eqref{Gaudin BAE} then it is expected that the Gaudin Hamiltonians commute with the action of $\slth$ and, therefore, the solutions of BAE correspond to a basis of singular vectors ${\mc L}^{sing}_{d_0,d_1}$, see Section \ref{periodic sec}. 
For generic $n=(n_1,n_2)$ the symmetry is broken and
the solutions of BAE are expected to  correspond to a basis of ${\mc L}_{d_0,d_1}$, that is all weight vectors, see Section \ref{exp sec}.
For the case $n_1+n_2=0$ (that is when $P(x)$ is a polynomial), we expect a partial symmetry: namely we expect that the Gaudin Hamiltonians commute with the Heisenberg subalgebra of $\slth$. Therefore we expect that the solutions of \eqref{BAE exp} correspond to vectors in ${\mc L}_{d_0,d_1}$ pictured in Figure \ref{basic module pic} by filled circles.
We have the following statement.
\begin{prop}\label{one orbit prop}
If
\bean\label{vir h w}
d_0=r^2, \qquad d_1=r(r-1), \qquad r\in\Z,
\end{align} 
then system \eqref{BAE exp} has a unique generalized solution.
\end{prop}
\begin{proof} 
Given a generalized solution $(y_0,y_1)$ corresponding to some weight $2r$, degree $d_0$ and twist $(n_0,n_1)$ the reproduction procedure \eqref{exp rep} produces a generalized solution $(\tilde y_0,y_1)$
corresponding to weight $-2r+2$ degree $d_0-2r+1$ and twist $(-n_0,n_1+2n_0)$; 
and a solution $(y_0,\tilde y_1)$ corresponding to weight $-2r$ degree $d_0$ and twist $(n_0+2n_1,-n_1)$.
We have $(n_0,n_1)=(-1,1)$, thus 
\be
(-n_0,n_1+2n_0)=(n_0+2n_1,-n_1)=(1,-1)=(-n_0,-n_1).
\ee
If $(y_0(x),y_1(x))$ is a solution 
corresponding to some weight $2r$, degree $d_0$ and twist $(n_0,n_1)$ then $(y_0(-x),y_1(-x))$ is a solution 
corresponding to same weight $2r$, same degree $d_0$ and the opposite twist $(-n_0,-n_1)$.
Thus, we have a bijection between sets of (generalized) solutions with any two different $r\in\Z$. Since for $r=0$, the solution is unique $(y_0,y_1)=(1,1)$, the proposition follows.
\end{proof}
Effectively, all the solutions correspond to the single population of solutions generated from $(y_0,y_1)=(1,1)$. This population corresponds to the orbit of affine Weyl group which is identified with the filled circles in Figure \ref{basic module pic}.
We note that the proof is similar to  the proof of Theorem 3.9 in \cite{MV3} but with some new features due to presence of the twist.
We also expect
\begin{conj}
If \eqref{vir h w} does not hold then system \eqref{BAE exp} has no solutions. \qed
\end{conj}
Let us denote the unique solution corresponding to $r\in\Z$ by $\{w_i^{r}\}_{i=1}^{r^2}$, $\{u_j^{r}\}_{j=1}^{r(r-1)}$. Let us also denote the polynomials with roots $w_i^{r}$  and $u_j^{r}$:
\ben\label{W U}
W_r(x)=2^{r^2} \,\prod_{i=1}^{r^2}(x-w_i^{r}), \qquad U_r(x)=2^{r(r-1)}\prod_{j=1}^{r(r-1)}(x-u_j^{r}).
\een
The recursion (the reproduction procedure) which allows to compute $W_r(x)$ and $U_r(x)$ recursively is explicitly described in the next lemma.
\begin{lem}
For all $r\in\Z$, we have
\be
 W_{-r}(x)=(-1)^r W_r(-x), \qquad U^{-r}(x)=U_{r+1}(-x).
\ee
For $r\geq 0$,
\begin{align*}
 &\on{Wr}\big(U_r,e^{2x}\,U_{-r}(-x)\big)=2 e^{2x}\,W_r^2,\\
 &\on{Wr}\big(W_{-r+1},e^{-2x}\,W_r(-x)\big)= (-1)^{r+1} 4 x e^{-2x}\,U_{-r+1}^2.
 \end{align*} 
\end{lem}
\begin{proof}
The lemma follows directly from the proof of Proposition \ref{one orbit prop}.
\end{proof}
Several first examples are
\begin{align*}
\big(W_0(x/2),U_0(x/2)\big)=\big(1,\ 1\big), \qquad  \big(W_1(x/2),U_1(x/2)\big)=\big(x-1,\ 1\big), \\
\big(W_{-1}(x/2),U_{-1}(x/2)\big)=\big(x+1,\ x^2+4x+5\big),\hspace{45pt}
\\
\big(W_2(x/2),U_2(x/2)\big)=\big(x^4-80x^3+40x^2-70x+35, \ x^2-4x+5\big),\hspace{5pt}
\end{align*}
and
\begin{align*}
&\big(W_{-2}(x/2),U_{-2}(x/2)\big)=\\&\quad \big(x^4-80x^3+40x^2-70x+35,\\
&\qquad \qquad x^6+20x^5+175x^4+840x^3+2275x^2+3220x+1925\big).
\end{align*}
We expect the following additional properties.
\begin{conj}\label{no double root conj}
For all $r$, polynomials $W_r,U_r$ have no multiple roots and no common roots.
For $r\geq 0$,
\begin{align*}
  W_{-r}(x/2),\  U_{-r}(x/2)\in\Z_{\geq 0}[x].
 \end{align*} 
 \qed
\end{conj}
We checked Conjecture \ref{no double root conj} by computer for small values of $|r|$.
The operator $L_1$ corresponding 
to solution of system \eqref{BAE exp} is a specialization of \eqref{L1 first} where $n=0$, $n_1=1$, and $P(x)=x$. It
takes the form
\ben\label{osc-1}
\partial_x^2-1-\sum_{i=1}^{a}\frac{2}{(x-w_i)^2}-\sum_{i=1}^a\frac{1/w_i}{x-w_i}-\la^2 x.
\een
The monodromy-free equation \eqref{monodromy eq} in this case has the form
\bean\label{new system}
\frac{1-4w_i^2}{4w_i^3}-\sum_{j=1,j\neq a}^{a}\frac{1}{(w_i-w_j)^2}\left(\frac{4}{w_i-w_j}+\frac{1}{w_i}+\frac{2}{w_j}\right), \qquad i=1,\dots,a.
\end{align}
Operator of form \eqref{osc-1} where $w_i$ are zeroes of $W_r(x)$ is monodromy-free. In particular zeroes of $W_r(x)$ (if there are no multiple zeroes, see Conjecture \ref{no double root conj}) satisfy \eqref{new system}. 
\begin{conj}\label{new system conj}
Zeroes of $W_r$ are the only solutions of system \eqref{new system}
\qed
\end{conj}
We expect that this conjecture can be proved by methods similar to \cite{O}.

\subsection{Large \texorpdfstring{$l$} limit, the case of \texorpdfstring{$d_0\neq d_1$}{d0 not d1}}
Here is a generalization of Conjecture \ref{conj d0=d1}.
\begin{conj}\label{conj r}
For generic $l,m$ the number of solutions of the Gaudin BAE \eqref{Gaudin BAE BLZ} up to permutation of coordinates equals $\on{Part}(d_0-r^2)$, the number of partitions of $d_0-r^2$. 
All solutions are isolated and non-degenerate. For fixed $d_0,d_1$, there exists $C\in\R$ such that parameters $l,m$ are generic if $|l|,|m|>C$. 
Moreover, for generic $l, k$ and a partition $\mu$ of $d_0-r^2$, there exists a solution $\{t_j^\mu\}_{j=1}^{d_1}$,$\{s_i^\mu\}_{i=1}^{d_0}$  with the asymptotics
\begin{subequations}\label{s r asympt}
\begin{equation}\label{s r asympt 1}
s_i^\mu=1+w^{r}_i\,\ell^{-2}+O(\ell^{-3}),
\end{equation}
for $ i=1,\dots, r^2$,
\begin{equation}\label{s r asympt 2}
s_i^\mu=\frac{k+2}{k+3}\left(1+\frac{2^{1/4}}{(k+2)^{1/4}(k+3)^{1/4}}\,v^\mu_i\, \ell^{-1}+O(l^{-2})\right),
\end{equation}
\end{subequations}
for $i=r^2+1,\dots,d_0$, and 
\begin{subequations}\label{t r asympt}
\bean
&t_i^\mu=1+u^{r}_i\,\ell^{-2}+O(\ell^{-3}), \qquad i=1,\dots, r(r-1),\label{t r asympt 1}
\\
&t_i^\mu=s_{i-r}^\mu(1 -\ell^{-2})+O(\ell^{-4}), \qquad i=r(r-1)+1,\dots, d_1,\label{t r asympt 2}
\end{align}
\end{subequations}
as $l\to\infty$ and $k$ is fixed. Here $\{w^{r}_i\}_{i=1}^{r^2}$ are zeroes of $W_\mu(x)$, $\{u^{r}_i\}_{i=1}^{r(r-1)}$ are zeroes of $U_\mu(x)$ see \eqref{W U} and $\{v^{\mu}_i\}_{i=1}^{d_0-r^2}$ are zeroes of $V_\mu(x)$, given by \eqref{Wr zero}. \qed
\end{conj}
Again, we prove the existence of monodromy-free operators of the form \eqref{L1 BLZ} with the correct asymptotics assuming that zeroes of $W_r(x)$ are simple and given a non-degenerate partition $\mu$ of $d_0$.

\begin{prop}\label{new prop}
Let $\mu$ be a non-degenerate partition of $d_0-r^2$. Let
$\{w^{r}_i\}_{i=1}^{r^2}$ be zeroes of $W_r(x)$. If 
  $\{w^{r}_i\}_{i=1}^{r^2}$ are all distinct, then for large $\ell$ there exists a solution $\{s_i,t_j\}_{i=1}^{d_0}$  of 
equations \eqref{la monodromy eq}
 with asymptotics \eqref{s r asympt}, \eqref{t r asympt}.
\end{prop}
\begin{proof}
Substitute first
\be
t_j=s_{j-r}(1-\ell^{-2})+\tilde t_j\ell^{-4}, \qquad j=r(r-1)+1,\dots,d_1,
\ee
and then
\begin{align*}
&t_j=1+u_j \ell^{-2}, \qquad 
 \hspace{20pt} j=1,\dots, r(r-1), \\
&s_j=1+w_j\ell^{-2}, \qquad   \hspace{20pt } j=1,\dots,r^2,\\ 
&s_j=\frac{k+2}{k+3}+v_j\ell^{-1}, \qquad  j=r^2+1,\dots,d_0,
\end{align*}
to the BAE equations \eqref{Gaudin BAE BLZ}. 
Then as $\ell\to \infty$
the leading terms of equations for $s_i,t_j$, $i=1,\dots,r^2$, $j=1,\dots,r(r-1)$, are of order $\ell^2$ and coincide with left hand sides of the exponential BAE \eqref{BAE exp} (where $d_0=r^2,d_1=r(r-1)$).  
At the same time 
the leading terms of equations for $t_j$,  $j=r(r-1)+1,\dots,d_1$, are of order $1$ and coincide with left hand sides of  \eqref{tilde equation} (where $d_0$ is replaced with $d_0-r^2$). 
Finally, for $i=r(r+1)+1,\dots,d_1$, the leading terms of the sum of equations for $t_i$  and $s_{i-r}$ with fixed $i$ are of order $\ell^{-1}$ and coincide with the left hand side of the system for poles of monodromy-free extensions of the quadratic oscillator, \eqref{order 0 system}  (where $d_0$ is replaced with $d_0-r^2$). 
Therefore, given an isolated solution of \eqref{BAE exp},  \eqref{order 0 system} corresponding to $W_r,U_r,V_\mu$, we deform it by the same argument as in Proposition \ref{old prop}.
\end{proof}
In the case $d_0=r^2$, $d_1=r(r-1)$, one can show that there is a unique generalized solution  of \eqref{Gaudin BAE BLZ} given by trigonometric reproduction procedure \eqref{trig rep}.
Namely, let $(y_0^{(r)}(x;k,l),y_1^{(r)}(x;k,l))$
be pairs of monic polynomials of degrees $(r^2,r(r-1))$ recursively defined by $(y_0^{(0)}(x;k,l),y_1^{(0)}(x;k,l))=(1,1)$ and
\begin{align*}
& \on{Wr}\big(y_0^{(-r+1)}(x;k,l),  x^{2k-2l+1} y_0^{(r)}(x;k, k-l+1)\big)\\
&\qquad \qquad \hspace{100pt} =c_r x^{k-2l}(x-1)(y_1^{(-r+1)}(x;k,l))^2,\\
 &y_1^{(r)}(x;k, k-l+1)=y_0^{(-r+1)}(x;k,l)
\end{align*}
and by
\begin{align*}
&y_0^{(-r)}(x;k,-l-1)=y_1^{(r)}(x;k,l),\\
& \on{Wr}\big(y_1^{(r)}(x;k,l),  x^{2l+1} y_1^{(-r)}(x;k,-l-1)\big)=\tilde c_r x^{2l}(y_0^{(r)}(x;k,l))^2.
\end{align*}
for suitable non-zero constants $c_r,\tilde c_r$.
For example, we have 
\be
(y_0^{(1)}, y_1^{(1)})=(x-\frac{k-2l}{k-2l+1},1),
\ee
and polynomials $y_0^{(-1)}, y_1^{(-1)}$  are described by \eqref{12 explicit 1} and \eqref{12 explicit 2}. 
\begin{lem}
The pair $(y_0^{(r)}(x;k,l),y_1^{(r)}(x;k,l))$ is the only generalized solution of equations
\eqref{Gaudin BAE BLZ} with $d_0=r^2$, $d_1=r(r-1)$.
\end{lem}
\begin{proof}
The lemma follows from the trigonometric reproduction procedure similar to Proposition \ref{one orbit prop}, cf. \cite{MLV}.
\end{proof}

\section{Asymptotic analysis of the spectral functions}\label{compare chapter}
In this section, after having introduced some notions
from the complex WKB method, we address the limit $\la \to 0$
of the spectral determinants $Q_{\pm},T_j, j\geq 1$, and the large $\la$ and large $l$ asymptotics
of the zeroes of the spectral determinant $Q_+$.
For reason of brevity, some proofs are only sketched and their completion
will appear in a separate publication. Whenever the corresponding proof is omitted,
we state the corresponding result as a conjecture.

 \subsection{Preamble on the Complex WKB method}
Our approach is based
on transforming a linear ODE of the second order into an integral equation of Volterra type, following \cite{E}.
We consider a second order scalar linear ODE of the form
\begin{equation}\label{eq:schrodingergen}
 y''(x)=V(x) y(x), \qquad x \in \mathbb{C},
\end{equation}
where $V(x)$ may depend on additional parameters, and a putative approximate solution $Y$
which we suppose to be of such a form that
\begin{equation}\label{eq:t}
 u(x)=\frac{y(x)}{Y(x)}
\end{equation}
is well-defined and approximately $1$ in a certain domain of $\mathbb{C}$ to be specified later.
Defining the forcing term
\begin{equation*}
 F(x)=V(x)-\frac{Y''(x)}{Y(x)}
\end{equation*}
we rewrite equation \eqref{eq:schrodingergen} for $y(x)$, as an equation for the function $u(x)$ defined by \eqref{eq:t}:
\begin{equation}\label{eq:usch}
 \frac{d}{dx}\big( Y^2(x) u'(x) \big)- Y^2(x) F(x) u(x)=0.
\end{equation}

We fix a point $x_0$  in the Riemann sphere, the boundary conditions $u'(x_0)=0$, $u(x_0)=1$, and a piece-wise smooth integration path
$\zeta$ connecting $x_0$ to another point $x \in \mathbb{C}$.
Integrating twice equation \eqref{eq:usch}, we see that $u(x)$ solves the following integral equation of Volterra type
\begin{equation}\label{eq:integral}
 u(x)=1-\int_{x,\zeta}^{x_0} B(x,s) F(s) u(s) d s \,, \qquad
 B(x,s)=\int_{x,\zeta}^t \frac{Y^2(s)}{Y^2(a)} d a \; ,
\end{equation}
provided the above integral converges absolutely.
Conversely, given any continuous solution $u(x)$ of the latter integral equation, the function
$y(x)=u(x) Y(x)$ solves \eqref{eq:schrodingergen} and satisfies the (possibly singular) Cauchy problem
\begin{equation*}
 \lim_{x \to x_0, x \in \zeta} \frac{y(x)}{Y(x)}=1 , \quad \lim_{x \to x_0, x \in \zeta} \frac{y'(x)}{Y'(x)}=1,
\end{equation*}
provided the integral operator $\int_{x,\zeta}^{x_0} B(x,s) F(s) u(s) d s$ is tame enough.
In practice, we often look for an approximate solution of the  form
\begin{equation*}
 Y(x)=e^{-\int^x \sqrt{U(s)}+\frac{U'(s)}{4 U(s)} ds},
\end{equation*}
for an \textit{effective potential} $U$, which does not necessarily coincide with $V$.
In this case, (\ref{eq:integral}) reads
\begin{align*}
 & u(x)=1-B[u](x), \qquad B[u](x)=\int^{x_0}_{x,\zeta} B(x,s) F(s) u(s) dt , \\
 & B(x,s)= \frac{\exp\left\{ -
 2 \int_{x,\zeta}^s \sqrt{V(t)}dt \right\}-1}{2}, \\
 &   F(x)= U^{-\frac12}(x) \left(U(x)-V(x)+\frac{ -4 U''(x) U(x)+5U'^2(x)}{16 U(x)^2}\right).
\end{align*}
We want to solve the above integral equation in the Banach space
$L^{\infty}_c(\zeta)$, the space of continuous complex-valued functions whose domain is the support of $\zeta$, equipped with the sup norm $\| \cdot \|_{\infty}$. A natural class of curves such that the integral operator $B$ is bounded on $L^{\infty}_c(\zeta)$ is the class of admissible (with respect to a given pair $V(x)$, $U(x)$) curves, to be defined below. In what follows, we will use the same letter $\zeta$ to denote a parameterised curve and its oriented image, since all statements are invariant under orientation preserving
reparametrization.

\begin{defn}\label{def:admissiblecurve}
 A parametrised curve $\zeta:\, (0,1)\to\C^*$ is called admissible with respect to a given pair $V(x)$, $U(x)$
if the following four conditions are met
 \begin{enumerate}[i)]
  \item $ U(\zeta(t)) \neq \infty $ and $V(\zeta(t)) \neq \infty$, for all $t\in(0,1)$.
  \item $U(\zeta(t)) \neq 0$ for all $t \in (0,1)$.
  \item For one of the two branches of $\sqrt{U}$ which is defined on the support of $\zeta$ we have
  \begin{equation*}
  \Re \int_{t'}^{t} \sqrt{U(\zeta(s))}\frac{d \zeta(s)}{ds}\,ds \leq 0, \qquad \forall t'\leq t.
\end{equation*}
 \item The forcing term $F$ is $L^1$-integrable along $\zeta$, namely
 \begin{equation*}
  \rho_{\zeta} = \int_0^1|F\big(\zeta(s)\big) \frac{d \zeta(s)}{ds}(s)|ds <+\infty.
 \end{equation*}
 \end{enumerate}
\end{defn}
Whenever $\zeta$ is admissible, the operator $B$ is a bounded linear operator on the Banach
 space $L^{\infty}_c(\zeta)$.
More precisely, the following properties can be proved by standard arguments
 \begin{align}
  & \|B\|_{\infty}\leq \rho_\zeta, \label{eq:Bproperty1}\\ 
  & \|B^{n}\|_{\infty} \leq\ \frac{\rho_{\zeta}^n}{n!},\notag
  \\ 
  & \lim_{t \to 1} B[u]=0, \qquad \forall u \in L_c^{\infty}(\zeta).\notag 
 \end{align}
In turn, the three properties above imply that, whenever $\zeta$ is admissible, then the
integral equation \eqref{eq:integral} admits a unique solution
\begin{equation}\label{eq:udef}
 u=\sum_{n=0}^{\infty} (-1)^n B^n[1],
\end{equation}
which satisfies the properties
$$
 \|u-1\|_{\infty} \leq e^{\rho_\zeta}-1,  \qquad \lim_{t \to 1} u(\zeta(t))=1.
$$
Summing up, one can prove the following proposition (see \cite{E} for details).
\begin{prop}\label{pro:continuation}
Let $D \subset \widetilde{\C^*}$ be an open simply-connected region,
$\zeta:(0,1) \to  \widetilde{\C^*}$ an admissible curve,
 and $t_0 \in (0,1)$.
 There exists a solution $y(x):\, D \to \C$ of \eqref{eq:schrodingergen} which, restricted to $\zeta\big( (0,1) \big) $,
 admits the representation
 \begin{equation*}
  y(\zeta(t))=e^{-\int_{t_0}^t \left( \sqrt{V(\zeta(t))}+ \frac{V'(\zeta(s))}{4 V(\zeta(s))}\right) \dot{\zeta}(s) ds } u\big(\zeta(t)\big),
 \end{equation*}
where $u$ is as in \eqref{eq:udef}.
\qed
\end{prop}

\subsection{The limit \texorpdfstring{$\la \to 0$}{lambda to 0}}

Here we study the $\la \to 0$ limit for the spectral determinants $Q_{\pm},T_j, j\geq 1$.
Such a study is necessary as this is a singular limit for Sibuya solutions: since
the irregular singularity becomes regular when $\la=0$, the limit
$\lim_{\la \to 0}\psi^{(0)}(x,\la)$ does not exist, not even point-wise.
Using the complex WKB method, we prove however that, with a proper (re)normalisation,
$\lim_{\la \to 0}\psi^{(0)}(x,\la)$ exists, provided
$x$ is restricted to the ray
$x= a e^{- i\frac{2\arg \la}{k+3}}$ with $a\geq a_0$ for some $a_0>0$ (that will be later determined).
This result in turn allows us to compute the limit $\la \to 0$ of the functions $Q_{\pm}(\la), T_j(\la), j \geq 1$.
Since we are interested in the limit $\la \to 0$, in this subsection we restrict $\la$ to the following subset
\begin{equation}\label{def:tD1}
\widetilde{D}_1=\lbrace \la \in \widetilde{\C^*},\; |\Pi(\la^{\frac{2}{k+3}})| \leq 1\rbrace .
\end{equation}
Moreover, in order to simplify the notation, we assume $l+r \neq 0$.
We introduce the effective potential\footnote{\ In the case $l +r=0$, the correct choice of the effective potential is not given by
equation (\ref{eq:reducedpot}), but by the expression $U(x,\la)=\la^2 x^k(x-1)$. With such a modification,
the proof below holds also for the case $l=0$. The details are left to the reader.}
\begin{equation}\label{eq:reducedpot}
U(x,\la,l,r)= \la^2 x^k(x-1)+ \frac{(l+r+\frac12)^2}{x^2},
\end{equation}
which we will denote simply by $U(x,\la)$.
Following the complex WKB method, we define
a family of integral operators on the space of continuous bounded functions with support $[a_0,\infty)$:
\begin{equation}\label{eq:Bladef}
B_{\la}[u](a)= e^{- i\frac{2\arg \la}{k+3}}  \int_{a}^{\infty} \frac{e^{-2S(t,\la)+2S(a,\la)} -1}{2} F(e^{-\frac{ i 2 \arg \la}{k+3}}t,\la) u(t) dt
\end{equation}
with
\begin{equation}\label{eq:Sdef}
S(a,\la)= e^{- i\frac{2\arg \la}{k+3}} \int_{a_0}^{a} \sqrt{U(e^{- i\frac{2\arg \la}{k+3}}t,\la)} dt,
\end{equation}
and
\begin{align}\label{eq:F}
F(x,\la)= U(x,\la)^{-\frac12}& \bigg(V(x,\la)-U(x,\la)\notag\\
&\left.+\frac{ -4 U''(x,\la) U(x,\la)+5U'^2(x,\la)}{16 U(x,\la)^2}\right),
\end{align}
and where, by equation \eqref{diffeq:L1s},
\begin{align*}
V(x,\la)=&U(x,\la)\\
&+\sum_{i=1}^{d_0}\left(\frac{2}{(x-s_i)^2}+\frac{k+s_i/(s_i-1)}{x(x-s_i)}\right)
         -\frac{r(r+1)+2 l\, r}{x^2} -\frac{1}{4x^2}.
\end{align*}
In equations (\ref{eq:Sdef}) and (\ref{eq:F}), the branch of the square root is chosen so that
\begin{equation}\label{eq:Sdefinitionbranch}
 \lim_{a\to +\infty} \Re S(a,\la)=+\infty .
\end{equation}
Before we proceed further with the analysis of the integral equation $u=1-B_{\la}[u]$, we collect some properties
of the functions $S$ and $F$ in the lemma below.
\begin{lem}\label{lem:uniformadmissibility}
Let $k>-2$, $\Re (l +r) \neq -\frac12$, $l+r\neq0$, and $ \widetilde{D}_1$ as in \eqref{def:tD1}.
If $a_0$ is sufficiently large then
\begin{enumerate}[i)]
 \item There exists a $C>0$ such that
\begin{equation}\label{eq:supResqrt}
  \inf_{\la \in \widetilde{D}_1}
 \Re \frac{\partial S(a,\la)}{\partial a} \geq \frac{C}{a}.
\end{equation}
 \item For $R(x)$  is as in \eqref{eq:truncatedaction},
 \begin{align} \label{eq:normalizationla0}
 \lim_{a\to + \infty}&\big(\la R(e^{- i\frac{2\arg \la}{k+3}} a)-S(a,\la)\big) =
 \frac{2\, \Theta\big( \hat{l}_r\big)}{k+3} \log |\la| \notag\\
 &+ \Theta\big( \hat{l}_r\big) \log(a_0)+ \frac{2\, \Theta\big( \hat{l}_r\big) \left(1- \log 2\, \Theta\big( \hat{l}_r\big)\right) }{k+3}\notag\\
 & +O(|\la|^2), \mbox{ as } \la \to 0.
 \end{align}
In the formula above $\hat{l}_r=l+r+\frac12$ and $\Theta(x)= sgn \left( \Re x \right)\, x$.
 \item If, moreover, $\{s_i\}_{i=1}^{d_0}$, $\{t_j\}_{j=1}^{d_1}$  is a solution of the Gaudin BAE \eqref{Gaudin BAE BLZ} with $d_0-d_1=r$, then the following estimates hold
 \begin{align}\label{eq:rholaell}
 & \sup_{\la \in \widetilde{D}_1} \rho(\la) < +\infty, \; \rho(\la)= \int_{a_0}^{+\infty} \left| F(a e^{- i\frac{2\arg \la}{k+3}},\la) \right|  d a, \\ \label{eq:rholabarell}
  & \sup_{\la \in \widetilde{D}_1} \overline{\rho}(\la) < +\infty, \; \overline{\rho}(\la)= \int_{a_0}^{+\infty} \left| U^{\frac12}(a e^{- i\frac{2\arg \la}{k+3}},\la)F(a e^{- i\frac{2\arg \la}{k+3}},\la) \right|  d a.
\end{align}
\end{enumerate}
\end{lem}
\begin{proof}
In this proof we set $\hat{\ell}_r=l+r+\frac12$. Notice that
\begin{equation}\label{eq:Sdefn}
 S(a,\la)=
\int_{a_0}^a \sqrt{\left| \la \right|^2 \big( t^{k+3}-
t^{k+2} e^{i\frac{2\arg \la}{k+3}}\big)+\hat{\ell}_r^2} \frac{d t}{t} .
\end{equation}
\begin{enumerate}[i)]
\item We show that if $a_0$ is sufficiently large, there exists a $C>0$ such that
\begin{equation*}
 \Re \sqrt{\left| \la \right|^2 \big( a^{k+3}-
a^{k+2} e^{i\frac{2\arg \la}{k+3}}\big)+\hat{\ell}_r^2} \geq C, \qquad \forall a \geq a_0, \forall \la \in A^*.
\end{equation*}
Since $\Re \hat{\ell}_r>0$, then we have that either $0 \leq \arg \hat{\ell}_r^2 < \pi $ or $-\pi < \arg \hat{\ell}_r^2 <0 $.
Assuming first that $0 \leq \arg \hat{\ell}_r^2 < \pi $, we define
$\delta=\pi-\arg \hat{\ell}_r^2$ and choose a $r_0$ large enough so that
$$\left|\arg \left( a^{k+3}-a^{k+2} e^{i\frac{2\arg \la}{k+3}} \right) \right| \leq \frac{\delta}2,$$
$\forall a \geq a_0.$
Simple trigonometric considerations lead to the inequalities
$$\left| | \la |^2 \big( a^{k+3}-
a^{k+2} e^{i\frac{2\arg \la}{k+3}}\big)+\hat{\ell}_r^2\right| \geq \frac{1-\cos \delta}{2}>0,$$
and
$$-\frac{\delta}{2} \leq  \arg\left(| \la |^2 \big( a^{k+3}-
a^{k+2} e^{i\frac{2\arg \la}{k+3}}\big)+\hat{\ell}_r^2 \right) \leq \pi-\delta.$$
Since, by \eqref{eq:Sdefinitionbranch}, the argument of the square root of
$$\left(| \la |^2 \big( a^{k+3}-
a^{k+2} e^{i\frac{2\arg \la}{k+3}}\big)+\hat{\ell}_r^2\right)$$
tends to $0$ as $a\to +\infty$,
the above inequalities imply that
 $$ \Re \sqrt{\left| \la \right|^2 \big( a^{k+3}-
a^{k+2} e^{i\frac{2\arg \la}{k+3}}\big)+\hat{\ell}_r^2} \geq  \left(\frac{1-\cos \delta}{2}\right)^{\frac12}
\cos (\frac{\pi-\delta}{2}) >0.$$
Assuming on the contrary that $-\pi < \arg \hat{\ell}_r^2 <0 $, we define
$\delta=\pi+\arg \hat{\ell}_r^2$ and following the same steps as above we prove the thesis.
\item
A standard computation yields that
\begin{equation}\label{eq:smalltermsU}
\sup_{t \in[a_0,\infty)}\left|\frac{|\la|^2 t^{k+2}}{\la^2 t^{k+3}+\hat{\ell}_r^2}\right| =O\big(|\la|^{\frac{2}{3+k}}\big),
\mbox{ as }
\la \to 0 .
\end{equation}
Using (\ref{eq:Sdefn}) and the above inequality, we have that
\begin{align*} %
S(a,\la)&= \int_{a_0}^a \sqrt{|\la|^2 (t^{k+1}-t^k) + \frac{\hat{\ell}_r^2}{t^2}}\,dt\\& =\int_{a_0}^a \sqrt{|\la|^2 t^{k+1}+\frac{\hat{\ell}_r^2}{t^2}} \sqrt{1-\frac{|\la|^2 t^{k+2}}{|\la|^2 t^{k+3}+\hat{\ell}_r^2}} \, dt  \\ \nonumber
& = \int_{a_0}^a \sqrt{|\la|^2 t^{k+1}+\frac{\hat{\ell}_r^2}{t^2}} \left( 1 + \sum_{m=1}^{\infty} c_m \left(\frac{|\la|^2 t^{k+2}}
{|\la|^2 t^{k+3}+\hat{\ell}_r^2}\right)^m \right)\, dt\\ \nonumber
& = \frac{2}{k+3} \left. \left(\sqrt{|\la|^2t^{k+3}+\hat{\ell}_r^2}-\frac{\hat{\ell}_r}{2} \log\left(\frac{\hat{\ell}_r+\sqrt{|\la|^2t^{k+3}+\hat{\ell}_r^2}}
{-\hat{\ell}_r+\sqrt{|\la|^2t^{k+3}+\hat{\ell}_r^2}}\right)\right) \right|_{t=a_0}^{t=a} \\
& + \sum_{m=1}^{\infty} c_m \left( \frac{|\la|^{2m} \hat{\ell}_r^{1-2m} t^{(k+2)m}}{(k+2)m}{~} \times \right.\\
&\hspace{30pt}\times \prescript{}{2}{F}_1 \left. \left. \left(\frac{2m-1}{2},\frac{(2+k)m}{3+k},
1+\frac{(2+k)m}{3+k},-\frac{|\la|^2 t^{k+3}}{\hat{\ell}_r^2}\right) \right) \right|_{t=a_0}^{t=a} .
\end{align*}
In the above formulae, $_2 F_1 $ is the Gauss hypergeometric function and $c_m$ are the coefficients of the Taylor expansion of $\sqrt{1-x}$ at $x=0$
Moreover, a standard computation shows that\footnote{In the formula below, if $\frac{k+3}{2}-m = 0$, $t^{\frac{3+k}{2}-m}$ must be replaced with $\log t$.}
\begin{align*}
|\la|^{2m} \hat{\ell}_r^{1-2m} t^{(k+2)m} {~}_2 F_1 \left(\frac{2m-1}{2},\frac{(2+k)m}{3+k},
1+\frac{(2+k)m}{3+k},-\frac{|\la|^2 t^{k+3}}{\hat{\ell}_r^2}\right) \\
= |\la| t^{\frac{3+k}{2}-m}+o(1),
\end{align*}
as $ \to +\infty$, and
\begin{align*}
|\la|^{2m} \hat{\ell}_r^{1-2m} a_0^{(k+2)m} {~}_2 F_1 \left(\frac{2m-1}{2},\frac{(2+k)m}{3+k},
1+\frac{(2+k)m}{3+k},-\frac{|\la|^2 a_0^{k+3}}{\hat{\ell}_r^2}\right)\\ =O\big(|\la|^{2m}\big),
\end{align*}
as $ \la \to 0$. Comparing the latter computations with \eqref{eq:truncatedaction}, we get
\begin{align*}
 \lim_{a\to + \infty}\big(\la R(e^{- i\frac{2\arg \la}{k+3}} a)&-S(a,\la \big) =
 \frac{2 \sqrt{|\la|^2a_0^{k+3}+\hat{\ell}_r^2} }{k+3} + \\
 &- \frac{\hat{\ell}_r}{k+3} \log\left(\frac{\hat{\ell}_r+\sqrt{|\la|^2a_0^{k+3}+\hat{\ell}_r^2}}
{-\hat{\ell}_r+\sqrt{|\la|^2a_0^{k+3}+\hat{\ell}_r^2}}\right)+O(|\la|^2),
\end{align*}
as $ \la \to 0 $. The thesis follows from expanding the right hand-side of the above formula.
\item By (\ref{at inf eq}),
$\sum_{i=1}^{d_0}\left(\frac{2}{(x-s_i)^2}+\frac{k+s_i/(s_i-1)}{x(x-s_i)}\right) -\frac{r(r+1)+2 l\, r}{x^2}=O(x^{-3})$. It follows that
\begin{equation*}
V(x)-U(x)=-\frac{1}{4x^2} + x^{-3} f(x)
\end{equation*}
for some function $f$ regular at $\infty$, which does not depend on $\la$.
Hence we have
\begin{equation*}%
 F(x,\la)= \frac{1}{\hat{\ell}_r} \left( \frac{g(x,\frac{\la}{\hat{\ell}_r})}{(1+ \big(\frac{\la}{\hat{\ell}_r}\big)^2
 x^{3+k}(1-x^{-1}))^{\frac52}} +\frac{x^{-2} f(x)}{(1+ \big(\frac{\la}{\hat{\ell}_r}\big)^2
 x^{3+k}(1-x^{-1}))^{\frac12}}\right),
\end{equation*}
where
\begin{align*}%
 g(x,c)=  & -4 c^2 (k+3)^2 x^{k+2} \left(1- \frac{(k+2)^2}{(k+3)^2} x^{-1} \right) \\
 & +  c^4 (k+3)^2 x^{5+2k}
 \left(1 -2 \frac{(k+4)(k+1)}{(k+3)^2} x^{-1}+\frac{(k+2)^2}{(k+3)^2} x^{-2}\right).
\end{align*}
Using estimate \eqref{eq:smalltermsU} and the above decomposition, we deduce that
if $a_0$ is sufficiently large then
\begin{align}
  \rho(\la) = & \, O \left( \int_{a_0}^{\infty} \frac{\big|\frac{\la}{\hat{\ell}_r}\big|^2 \, a^{k+2} \, da}{ \left|1+\big(\frac{|\la|}{\hat{\ell}_r}\big)^2 a^{3+k}\right|^{\frac52} } \right) + O \left( \int_{a_0}^{\infty}
  \frac{\big|  \frac{\la}{\hat{\ell}_r}\big|^4 \, a^{5+2k} \, da}{ \left|1+\big(\frac{|\la|}{\hat{\ell}_r}\big)^2
  a^{3+k}\right|^{\frac52} } \right)  \notag\\
 + & \, O \left( \int_{r_0}^{\infty} \frac{a^{-2} \, da}{ \left|1+ \big(\frac{|\la|}{\hat{\ell}_r}\big)^2 a^{3+k}\right|^{\frac12} } \right) , \mbox{ as } \la\to 0.\label{eq:rholaintermediateestimate}
\end{align}
All three contributions are  $O(1)$ as $\la \to 0$.
The integral in the third term is bounded since $\int_{a_0}^{\infty }x^{-2}<+\infty$ and
$$ \left|1+ \big(\frac{|\la|}{\hat{\ell}_r}\big)^2 a^{3+k}\right|^{-\frac12}$$
is bounded. As for the first term, we split the integration in two and we deduce that
\begin{align*}
 &\int_{a_0}^{\infty} \frac{\big|\frac{\la}{\hat{\ell}_r}\big|^2 \, a^{k+2} \, da}{ \left|1+\big(\frac{|\la|}{\hat{\ell}_r}\big)^2 a^{3+k}\right|^{\frac52} }  \\& =
  \int_{a_0}^{a_0 \big|\frac{\la}{\hat{\ell}_r}\big|^{-\frac{2}{k+3}}}
  \frac{\big|\frac{\la}{\hat{\ell}_r}\big|^2 \, a^{k+2} \, da}{ \left|1+\big(\frac{|\la|}{\hat{\ell}_r}\big)^2 a^{3+k}\right|^{\frac52} }
+  \int_{a_0 \big|\frac{\la}{\hat{\ell}_r}\big|^{-\frac{2}{k+3}}}^{\infty} \frac{\big|\frac{\la}{\hat{\ell}_r}\big|^2 \, a^{k+2} \, da}{ \left|1+\big(\frac{|\la|}{\hat{\ell}_r}\big)^2 a^{3+k}\right|^{\frac52} } \\
& \leq C_1 \int_{a_0}^{a_0 \big|\frac{\la}{\hat{\ell}_r}\big|^{-\frac{2}{k+3}}}
 \big|\frac{\la}{\hat{\ell}_r}\big|^2 \, a^{k+2} \, da +
 C_2 \int_{a_0 \big|\frac{\la}{\hat{\ell}_r}\big|^{-\frac{2}{k+3}}}^{\infty} \frac{\big|\frac{\la}{\hat{\ell}_r}\big|^2 \, a^{k+2} \, da}{ \left|\big(\frac{|\la|}{\hat{\ell}_r}\big)^2 a^{3+k}\right|^{\frac52} }
\end{align*}
for some $C_1,C_2>0$. By explicit integration of the two latter integrals, we obtain that first term is bounded.
The same splitting argument, yields that also the second term in \eqref{eq:rholaintermediateestimate} is bounded.
As for $\overline{\rho}(\la)$, reasoning as above we obtain
\begin{align*}
  \overline{\rho}(\la) = & \, O \left( \int_{a_0}^{\infty} \frac{\big|\frac{\la}{\hat{\ell}_r}\big|^2 \, a^{k+1} \, da}{ \left|1+ \big(\frac{|\la|}{\hat{\ell}_r}\big)^2 a^{3+k}\right|^2 } \right) + O \left( \int_{a_0}^{\infty}
  \frac{\big|  \frac{\la}{\hat{\ell}_r}\big|^4 \, a^{4+2k} \, da}{ \left|1+ \big(\frac{|\la|}{\hat{\ell}_r}\big)^2
  a^{3+k}\right|^2 } \right)+  \\
 + & \, O \left( \int_{a_0}^{\infty} a^{-3} \, da \right) , \mbox{ as } \la\to 0.
\end{align*}
The above expression can be analysed in the same way as \eqref{eq:rholaintermediateestimate} to deduce that
$\overline{\rho}(\la)$ is bounded.
\end{enumerate}
 \end{proof}

\begin{prop}\label{prop:psi0tochi+}
Let $\{s_i\}_{i=1}^{d_0}$, $\{t_j\}_{j=1}^{d_1}$ be a solution of the Gaudin BAE \eqref{Gaudin BAE BLZ} with $d_0-d_1=r$,
and $\widetilde{D}_1$ as in \eqref{def:tD1}; assume moreoveor that $\Re (l +r) \neq -\frac12$.
If $a$ is a large enough positive real number then
\begin{align}\label{eq:la0psi0}
  & \sup_{\la \in \widetilde{D}_1}\left|\la^{-\frac{1}{2}+\frac{2\Theta(l+r+\frac12)}{k+3}} \psi^{(0)}(a e^{- i\frac{2\arg \la}{k+3}},\la)\right| <+\infty , \\ \label{eq:la0pis'0}
   & \sup_{\la \in \widetilde{D}_1} \left|\frac{d}{dx} \left( \la^{-\frac{1}{2}+\frac{2\Theta(l+r+\frac12)}{k+3}} \psi^{(0)}(x,\la)\right)_{|x=a e^{- i\frac{2\arg \la}{k+3}}}\right| <+\infty,
 \end{align}
where in the formula above $\Theta(x)= sgn \left( \Re x \right)\, x$.
\end{prop}
\begin{proof}
Let us study the integral equation
\begin{equation}\label{eq:intpsi0u}
 u=1-B_{\la}[u], \; u \in L_c^{\infty}\big([a_0,\infty)\big)
\end{equation}
with $B_{\la}$ as in \eqref{eq:Bladef}.
Estimates \eqref{eq:supResqrt}, \eqref{eq:rholaell} together with (\ref{eq:Bproperty1}) imply that
\begin{equation*}
\sup_{\la \in \widetilde{D}_1} \|B_{\la} \| < \infty.
\end{equation*}
Together with Proposition \ref{pro:continuation}, this implies that the integral equation \eqref{eq:intpsi0u}
admits a unique solution $u_{\la}$ such that
\begin{equation}\label{eq:ulabounded}
 \sup_{\la \in \widetilde{D}_1}  \|u_{\la}\| <\infty,
\end{equation}
and that the function
$$ \psi_u(a,\la)=e^{-S(a,\la)-\frac{1}{4}\log U\big(ae^{-\frac{2 i \arg \la}{k+3}})+
 \frac{1}{4}\log U\big(a_0e^{-\frac{2 i \arg \la}{k+3}}\big)}u_\la(ae^{-\frac{2 i \arg \la}{k+3}})
$$
is the restriction to the ray $x=a\,e^{-\frac{2 i \arg \la}{k+3}}, \, a\geq a_0$ of a solution to \eqref{diffeq:L1s}.
Moreover, due to \eqref{eq:ulabounded}, for any fixed $a\geq a_0$,
\begin{equation}\label{eq:psiubounded}
 \sup_{\la \in \widetilde{D}_1} \psi_u(a e^{- i\frac{2\arg \la}{k+3}},\la)<\infty.
\end{equation}
After \eqref{eq:normalizationla0}, such a solution is subdominant in $\Sigma_0[\la]$ and
hence, by uniqueness of subdominant solutions, proportional to the Sibuya solution $\psi^{(0)}$. More precisely, comparing the asymptotic behaviour of $\psi^{(0)}$ and $\psi_u$ using equation \eqref{eq:reducedpot}, Proposition \ref{prop:sibuya}, and Lemma \ref{lem:uniformadmissibility} iii), we compute that
\begin{align*}\nonumber
 \psi_u&=C(\la) \la^{-\frac12+\frac{2\Theta(l+r+\frac12)}{k+3}} \psi^{(0)}, \\ \label{eq:propfactor}
C(\la)&=
 \lim_{a \to +\infty} \left(\la^{-\frac12+\frac{2\Theta(l+r+\frac12)}{k+3}}\psi_u(a,\la) x^{\frac{k+1}{2}}e^{ \la R(e^{- i\frac{2\arg \la}{k+3}} a)}\right)\\
 &=e^{-\frac{2\, \Theta\big( l+r+\frac12\big) \left(1- \log \big(2\,
 \Theta(l+r+\frac12) \big)\right) }{k+3})+O\big(|\la|^2\big)}.
\end{align*}
Combining the latter identity and \eqref{eq:psiubounded}, \eqref{eq:la0psi0} follows.
In order to prove \eqref{eq:la0pis'0}, it is sufficient to show that
\begin{equation}\label{eq:u'labounded}
 \sup_{\la \in \widetilde{D}_1} \| \frac{d u_\la}{da} \| <\infty.
\end{equation}
In fact, if this holds then, for any fixed $a \geq a_0$,
$$ \sup_{\la \in \widetilde{D}_1}\left| \frac{d}{d x}  \left( \psi_u(x,\la)\right)_{|x=a e^{- i\frac{2\arg \la}{k+3}}}\right|<\infty.$$
Differentiating \eqref{eq:intpsi0u}, we get
\begin{equation*}
 \frac{d u_\la}{d\,a} =  e^{- i\frac{4\arg \la}{k+3}}  \int_{a}^{\infty} e^{-2S(t,\la)+2S(a,\la)}
 U^{\frac12}(e^{-\frac{ i 2 \arg \la}{k+3}}t,\la) F(e^{-\frac{ i 2 \arg \la}{k+3}}t,\la) u(t) dt.
\end{equation*}
Then \eqref{eq:u'labounded}
follows from the above integral representation and \eqref{eq:supResqrt}, \eqref{eq:rholabarell}, \eqref{eq:ulabounded}.
Proposition \ref{prop:extendability} is proved.
 \end{proof}
Recall the definition of the spectral determinants $Q_{\pm}$ and $T_j$, in Theorems \ref{thm:QQsystem} and \ref{thm:TQ}.
From Proposition \ref{prop:psi0tochi+}, we deduce that the function $Q_{+}$ is entire when $\Re \big( l+r+\frac12 \big)>0$, the function $Q_{-}$ is entire when $\Re \big( l+r+\frac12 \big)<0$, and the functions $T_j, j \in \mathbb{Z}_{\geq 0}$ are entire
if $\Re \big( l+r+\frac12 \big) \neq 0$.
\begin{prop}\label{prop:extendability}
Fix $k>-2$, $l$ such that $ \Re (l+r) \neq -\frac12$. Let  $\{s_i\}_{i=1}^{d_0}$, $\{t_j\}_{j=1}^{d_1}$ be a solution of the Gaudin BAE \eqref{Gaudin BAE BLZ} with $d_0-d_1=r$. Let $Q_+,Q_-, T_j, j \geq1$ be the spectral determinants of equation \eqref{diffeq:L1s} defined
in Theorems \ref{thm:QQsystem} and \ref{thm:TQ}.
\begin{enumerate}[i)]
 \item Assume that $\Re (l+r+\frac12) >0$. The holomorphic function $Q_{+}$ has a removable singularity at $\la=0$. If, moreover, we assume that $2l+1 \notin \lbrace (k+2) i+j, (i,j)\in \mathbb{Z}_{\geq 0}^2 \rbrace$, the holomorphic function
 $Q_{-}$, is meromorphically extendable $0$ where it has at most a pole of order
 $\lfloor \Re (2 (l+r)+1) \rfloor$. Furthemore, whenever $Q_{+}(0)\neq 0$, $Q_{-}$ is regular at $0$ and $Q_{-}(0) \neq 0$.
\item Assume that $\Re (l+r+\frac12) <0$. The holomorphic function $Q_{-}$ has a removable singularity at $\la=0$. If, moreover, we assume that $-(2l+1) \notin \lbrace (k+2) i+j, (i,j)\in \mathbb{Z}_{\geq 0}^2 \rbrace$, the holomorphic function
 $Q_{+}$, is meromorphically extendable $0$ where it has at most a pole of order
 $\lfloor -\Re (2 (l+r)+1) \rfloor$. Furthemore, whenever $Q_{-}(0)\neq 0$, $Q_{+}$ is regular at $0$ and $Q_{+}(0) \neq 0$.
 \item The holomorphic functions $T_j:\C^* \to \C$, $j\geq 1$ have a removable singularity at $\la=0$.
 \end{enumerate}
 \end{prop}

 \begin{proof}
\begin{enumerate}[i)]
\item As it is well-known, a holomorphic function on $\C^* \subset \C$ can be holomorphically extended to $0$ if and only if it is bounded on a punctured neighborhood of $0$.
Since by definition  $Q(\la^{\frac{2}{k+3}})=Q^*(\la)$, in order to prove the thesis it is sufficient show that
$Q^*(\la)$ is bounded  on the set $\widetilde{D}_1$ defined
in equation \ref{def:tD1}.
By definition of $Q^*$, equation \eqref{Q*},
\begin{equation*}
 Q^*(\la)=\frac{e^{-\frac{i\pi}{4}}}{\sqrt{2}(2l+1)^{1/2}}\la^{-\frac{1}{2}+\frac{2(l+r)+1}{k+3}}\on{Wr}(\psi^{(0)},\chi^+)(x,\la),
\end{equation*}
where $\chi_+(x,\la)$ is the Frobenius solution defined in Corollary \ref{cor:frobenius}. The thesis is proved if
we show that there exists $a>0$ such that the evaluation of the above Wronskian at the point $x=ae^{- i\frac{2\arg \la}{k+3}}$ is a bounded function of $\la \in \widetilde{D}_1$.
By Proposition \ref{prop:psi0tochi+}, the evaluation of the function $\la^{-\frac{1}{2}+\frac{2(l+r)+1}{k+3}}\psi^{(0)}$
and of its derivative at the point $x=ae^{- i\frac{2\arg \la}{k+3}}$ are bounded functions of $\la \in \widetilde{D}_1$ provided $a$ is large enough (fixed) real positive number.
Moreover, the evaluation of $\chi_+(x,\la)$ and of its derivative at  $x=a e^{- i\frac{2\arg \la}{k+3}}$ is a bounded function of $\la \in \widetilde{D}_1$, provided $a$ is a fixed positive real number such that $a \neq |s_i|$ for all $i=1,\dots,d_0$. This follows from the Frobenius series representation of $\chi_+(x,\la)$, equation \eqref{eq:chipmx}, and a standard perturbation analysis.

Reasoning as above, we deduce that $Q_{-}$ can be meromorphically extended to $0$, where it has at most a pole of order
$\lfloor \Re( 2 l+ 2r +1) \rfloor$.
Now we assume that $Q_{+}(0) =c_1 \neq 0$ and we prove that $Q_{-}$ has a non-zero finite limit at $0$.
Since $Q_{-}$ is meromorphic and it has at most a pole of order $\lfloor \Re( 2 l+ 2r+1) \rfloor$,
there exists a $c_2 \neq 0 $ and an integer $n\leq \lfloor \Re( 2 l+ 2r +1) \rfloor$ such that
$$Q_{-}(\hat{\la})=c_2 \hat{\la}^{-n} \left( 1+ O\big(\hat{\la}\big)\right)$$
as $\hat{\la} \to 0$. We use the $QQ$-system \eqref{QQsystem} to prove that $n=0$. In fact, the $QQ$-system implies that
\begin{equation} \label{eq:QQqgamma}
 c_1\, c_2 \left(\gamma_r q^{n}-\gamma_{r}^{-1} q^{-n}\right) \hat{\la}^{-n} \left( 1+ O\big(\hat{\la}\big)\right) = 1,
 \end{equation}
 with $\gamma_{r}$ and $q$ given by \eqref{gammaQQ}.  The first consequence of the above equation is that $n\geq0$. We prove that $n=0$ by contradiction.
 Assume that there exists an $n \geq 1$ such that
 equation (\ref{eq:QQqgamma}) is satisfied. This implies that
 $\gamma_r q^{n} -\gamma_{r}^{-1} q^{-n}=0$,  whence there exists a positive integer $j$
such that
\begin{equation*}
 \frac{1}{k+3}\big(2l+1-(n-2r)\big)+(n-2r) =j.
\end{equation*}
This implies that $2l+1$ is real and that $j-(n-r)\geq0$, since now $n-r\leq 2l+1$.
Therefore there exist two non-negative integers
$j>0$ and $ i=j-n+r \geq 0$ such that $2l+1=(k+2) i+j$. This is a contradiction.
\item Same as the proof of i).
\item Since $Q_+(\la)$ or $Q_-(\la)$ is analytic at $0$, the $TQ$-system \eqref{TQsystem} implies that the limit $\lim_{\hat{\la} \to 0}T_1(\hat{\la})$ exists, whence $T_1$ has a removable singularity at $\la=0$.
Using the fusion relation \eqref{fusionrelations}, we deduce by induction that all $T_j$ have a removable singularity at $\la=0$.
\end{enumerate}
\end{proof}

A lengthier and more careful analysis would prove
the following stronger version of Proposition \ref{prop:extendability}.
\begin{conj}\label{conj:rneq0}
Fix $k>-2$, $l, \Re \big( l+r\big) \neq -\frac12$. Let  $\{s_i\}_{i=1}^{d_0}$, $\{t_j\}_{j=1}^{d_1}$ be a solution of the Gaudin BAE \eqref{Gaudin BAE BLZ} with arbitrary $r=d_0-d_1$.
If $2l+1 \notin \lbrace (k+2) i+j, (i,j)\in \mathbb{Z}_{\geq 0} \rbrace$, the holomorphic functions $Q_{\pm}: \C^* \to \C$ are holomorphically extendable to $0$ and $Q_{\pm}(0) \neq 0$.
\qed
\end{conj}

\subsection{Asymptotics of the zeroes of the \texorpdfstring{$Q_+$}{Q+}-functions}\label{heu section}
Here we study the asymptotic distribution of the zeroes of the spectral determinant $Q_+$, assuming that $r=d_0-d_1=0$, $l$ is real and $l>-\frac12$.
By definition, the spectral determinant $Q_+(\la)$ is -- after the change of variable $\la\mapsto \la^{\frac{k+3}{2}}$ -- the Wronskian of the subdominant Frobenius solution
$\chi_+$ and of the Sibuya solution $\Psi^{(0)}$, see \eqref{Q Q*} and \eqref{Q*}. This means that the zeroes of $Q_+$ corresponds to those $\la$ for which there exists a non-trivial solution of the boundary value problem
\begin{align}\label{eq:bvgeneral}
 \begin{cases}
 & \psi''(x)=\left(\frac{l(l+1)}{x^2}+\lambda^2x^k(x-1)+
 \sum_{i=1}^{d_0}\left(\frac{2}{(x-s_i)^2}+\frac{k+s_i/(s_i-1)}{x(x-s_i)}\right)\right)\psi(x), \\
 &\lim_{t \to 0^+} t^{l} \psi(e^{-\frac{i2 \arg{\la}}{k+3}} t)=
 \lim_{t\to+\infty} \psi(e^{-\frac{i2 \arg{\la}}{k+3}} t)=0.
 \end{cases}
\end{align}
We note that if $l\geq 0$, the boundary conditions can be written in the simpler form
\begin{equation}\label{eq:bvlg0}
  \lim_{t \to 0^+} \psi(e^{-\frac{i2 \arg{\la}}{k+3}} t)=\lim_{t\to+\infty} \psi(e^{-\frac{i2 \arg{\la}}{k+3}} t)=0.
\end{equation}
 We are interested in two asymptotic regimes, the first is the
asymptotic distribution of the zeroes when $l$ is fixed, the second is the
asymptotic distribution of the zeroes when $l \to +\infty$.
We do not furnish complete proofs of all statements, but only sketches thereof or heuristic arguments.
A thoroughly analysis of the boundary value problem above will be
addressed in a separate publication, dedicated to it.
The asymptotic analysis of the boundary value problem \eqref{eq:bvgeneral} can be obtained via the complex WKB method.
Recall the definition of the effective potential \eqref{eq:reducedpot}. Setting $r=0$, it reads
\begin{equation*}
 U(x,\la)= \la^2 x^{k}(x-1)+\frac{\hat{\ell}^2}{x^2}, \quad  \hat{\ell}=l+\frac12.
\end{equation*}
Here we assume that the parameters $\la^2$ and $\hat{\ell}$ are real and positive.
We are interested in the zeroes of the function $U$ restricted to the positive semiaxis $x>0$.
These are called the turning points of the effective potential and their
location only depends on the ratios $u=\frac{\la}{\hat{\ell}}$.
Clearly the turning points $U$ coincide with the zeroes of $x^2 U$ (restricted to the positive semiaxis). Since the derivative of the function $x^2 U$ has a unique zero $x_*$,
\begin{equation}\label{eq:x*}
 x_*=\frac{k+2}{k+3},
\end{equation}
we deduce that there are three possible configurations of real and positive zeroes of $U$.
Letting
\begin{equation}\label{eq:u*}
 u_*=\frac{\big(k+3)^{\frac{k+3}{2}}}{\big(k+2\big)^{\frac{k+2}{2}}}
\end{equation}
these possible configurations are
\begin{enumerate}
 \item If $0<\la <u_* \, \hat{\ell}$, $U(x)>0$ for all $x>0$.
 \item If $\la=u_* \, \hat{\ell}$, $x=\frac{2+k}{3+k}$ is a double zero and
 $U(x)>0$ for all $x>0$,  $x \neq \frac{2+k}{3+k}$.
 \item If $\la>u_*\, \hat{\ell}$,
  $U$ has two real positive zeroes $x_-\big(\hat{\ell}/\la\big)<x_+\big(\hat{\ell}/\la\big)<1$. Moreover $U(x)>0$ when $x<x_-(\hat{\ell}/\la)$ or $x>x_+(\hat{\ell}/\la)$, while $U(x)<0$ when $x_-(\hat{\ell}/\la)<x<x_+(\hat{\ell}/\la)$.
 \end{enumerate}
We studied the configurations of the turning points, because
the large $\la$ (or large $\hat{\ell}$) asymptotic distribution of the zeroes of the $Q_+$ function
is obtained from the large $\la$ (or large $\hat{\ell}$) asymptotics of the
WKB integral defined for positive real $\la$ and $\hat \ell$ by the formula
$$
 I(\la,\hat{\ell}) = \begin{cases}
             0, \qquad& 0 < \la\leq u_* \hat{\ell},\\
            \frac{2}{\pi} \int_{x_-\left(\hat{\ell}/\la\right)}^{x_+\left(\hat{\ell}/\la\right)} \sqrt{-U(x,\la,\hat{\ell})} dx, \qquad & \la >u_* \, \hat{\ell}.
            \end{cases}
$$
Since
\begin{equation*}
 I(\la,\hat{\ell})=\la \, I(1,\frac{\hat{\ell}}{\la}),\qquad \ I(\la,\hat{\ell})=\hat{\ell} \, I(\frac{\la}{\hat{\ell}},1),
\end{equation*}
it is convenient to introduce the functions
$$
 J_1(u)=I(1,u), \qquad J_2(u)=I(u,1)=u J_1(\frac{1}{u}).
$$
\begin{lem}
The following asymptotic formulae hold
 \begin{align}\label{eq:J10}
  & J_1(u)= \frac{2}{k+2}\left(\frac{1}{\sqrt{\pi}}\frac{\Gamma\big(\frac{k+4}{2} \big)}{\Gamma \big(\frac{k+5}{2}\big)} - u \right)+ O(u^{-2}), \qquad\quad u \to 0^+, \\ \label{eq:J2u*}
  & J_2(u)= \left(\frac{8 (2+k)^{1+k}}{(3+k)^{4+k}} \right)^{\frac12}(u-u_*)+ O\big( (u-u_*)^2 \big),\quad  u - u_* \to 0^+.
 \end{align}
\end{lem}
\begin{proof}
The lemma is equivalent to formulae \cite[equation (2.25)]{CM2} and \cite[equation (2.22)]{CM2}.
More precisely, in
\cite[Proposition 2.7]{CM2}, the following function was studied
$$\tau(\xi;\alpha)=\frac{2(1+\alpha)}{\pi}\int_{w_-}^{w_+}
\sqrt{\frac{\alpha+1}{\alpha} \alpha^{\frac{1}{\alpha+1}}\xi w^2- w^{2\alpha+2}-1}\ \frac{dw}{w},
\quad \xi\geq 1,\ \alpha>0,$$
where $w_{\pm}$ are the two zeros of $t^{\frac{2}{k+3}} w^2- w^{\frac{2(3+k)}{2+k}}-1$.
Using the change of integration variable $v=w^{\frac2{k+2}}u^{-\frac{2}{k+3}}$, a direct computation yields that
\begin{equation}\label{eq:J2tau}
 J_2(u)=\frac{2\alpha}{1+\alpha} \tau\big(\alpha^{-\frac{1}{1+\alpha}}\frac
 {\alpha}{1+\alpha}u^{\frac{2\alpha}{1+\alpha}};\alpha \big), 
\end{equation}
where $ \alpha=\frac{1}{k+2}$. Taking into account (\ref{eq:J2tau}), \eqref{eq:J10} is equivalent to
\cite[equation (2.22)]{CM2}, and  (\ref{eq:J2u*}) is equivalent to \cite[equation (2.25)]{CM2}
\end{proof}

\begin{lem}\label{lem:lajas}
For every $\hat{\ell}>0$ and every $j \in \mathbb{Z}_{\geq 0}$, the implicit equation
\begin{equation*}\label{eq:deflajp}
I(\la,\hat{\ell})=2j+1,
\end{equation*}
with $ \la >0$ has a unique solution $\widetilde{\la}_j(\hat{\ell})$, $j \in \mathbb{Z}_{\geq 0} $.
The functions $\widetilde{\la}_j(\hat{\ell})$ are smooth, and the following asymptotics hold. For fixed  $\hat{\ell}$:
\begin{equation}
 \widetilde{\la}_j(\hat{\ell})=\frac{k+2}{2}\frac{\sqrt{\pi} \Gamma \big(\frac{k+5}{2}\big)}{\Gamma\big(\frac{k+4}{2} \big)}
 \left(2j+1+\frac{2\hat{\ell}}{k+2} \right)+ O(j^{-1}),\quad j\to +\infty  . \label{eq:lajlargej}
 \end{equation}
For fixed $j$:
\begin{equation}
 \widetilde{\la}_j(\hat{\ell})=\frac{\big(k+3)^{\frac{k+3}{2}}}{\big(k+2\big)^{\frac{k+2}{2}}}\, \hat{\ell}+ \left(\frac{8 (2+k)^{1+k}}{(3+k)^{4+k}} \right)^{-\frac12} \left(2j+1\right)+O\big(\hat{\ell}^{-1}\big),\quad  \hat{\ell}\to +\infty . \label{eq:lajlargep}
  \end{equation}
\end{lem}

\begin{proof}
Since $x_-(\hat{\ell}/\la)<x_+(\hat{\ell}/\la)<1$, then
 $$
 \frac{\partial{I}}{\partial \la}= \frac{2 \la}{\pi} \int_{x_-}^{x_+}
 \frac{(1-x)x^k}{\sqrt{-U(x,\la,\hat{\ell})}} dx>0, \qquad \forall \la > u_* \hat{\ell}.
 $$
 Therefore, for fixed $\hat{\ell}$ and $j$, the equation $I(\la,\hat{\ell})=2j+1$ admits a unique solution,
 $\la_j(\hat{\ell})$. Moreover, by the implicit function theorem such a solution is a smooth function
 of $\ell$.
 Equations \eqref{eq:lajlargej}, \eqref{eq:lajlargep} follow directly from \eqref{eq:J10}, \eqref{eq:J2u*}.
\end{proof}

\subsubsection*{The case $d_0=0$}
Using Lemma \ref{lem:lajas} and the complex WKB method, one obtains a very precise characterization
of the zeroes of $Q_+$ for the ground state opers, namely for the solutions of (\ref{eq:bvgeneral})
with $d_0=0$.
\begin{prop}\label{prop:asymptoticground}
Fix $k>-2$.
Let $Q_+(\la;l)$ be the spectral determinant corresponding to the ground state $\la$-oper, namely $d_0=0$,
with $l$ real, $l>-\frac12$.
The zeroes of $Q_+(\la;l)$ are simple, real and positive. They form an increasing and diverging sequence
$\{\la_j\}_{j \in \mathbb{Z}_{\geq 0}}$ which satisfies the following asymptotics
\begin{align}\label{eq:hatlajellfixed}
& \left(\la_j\right)^{\frac{k+3}{2}}=
\frac{k+2}{2}\frac{\sqrt{\pi} \Gamma \big(\frac{k+5}{2}\big)}{\Gamma\big(\frac{k+4}{2} \big)}
 \left(2j+1+\frac{\hat{\ell}}{k+2} \right)+ O(j^{-1}),\qquad j\to +\infty,
 \end{align}
where $\hat{\ell}=l+\frac12$ Moreover, for fixed $j$, the following large $\hat{\ell}$ asymptotics holds
  \begin{align}
\label{eq:lahatjp}
& \left(\la_j\right)^{\frac{k+3}{2}}=
\frac{\big(k+3)^{\frac{k+3}{2}}}{\big(k+2\big)^{\frac{k+2}{2}}} \hat{\ell}+ \left(\frac{8 (2+k)^{1+k}}{(3+k)^{4+k}} \right)^{-\frac12} \left(2j+1\right)+O\big(\hat{\ell}^{-1}\big),\quad  \hat{\ell}\to +\infty.
\end{align}
\end{prop}
\begin{proof}
As we studied in Section \ref{sub:BLZopers}, under the change of variable \eqref{eq:xtowparameters},
the boundary value problem \eqref{eq:bvgeneral} becomes
\begin{equation}\label{eq:doreytateoeq}
\begin{cases}
 \psi''(w)=\left(w^{\frac2{k+2
}}-\left(\frac{2}{k+2}\right)^{-\frac{2}{k+3}}\hat{\la}+\frac{\tilde{l}(\tilde{l}+1)}{w^2} \right)\psi(w),\\
 \lim_{w \to 0^+} w^{\tilde{l}} \psi(w)=
 \lim_{w\to+\infty} \psi(w)=0.
 \end{cases}
\end{equation}
The boundary value problem (\ref{eq:doreytateoeq}) is a self-adjoint Sturm Liouville problem which has been intensively studied  and, in particular, in connection to  the ODE/IM correspondence.
In \cite[Section 2]{DT}, the spectrum is shown to be simple and positive,
and the eigenfunction corresponding to the $j$-th eigenvalue $\hat{\la}_j$ has
has exactly $j$ zeroes (called \textit{nodes}) on the open semi-axis $w>0$.
Formula (\ref{eq:hatlajellfixed}) is also derived (but not proved) using the WKB approximation.
We outline the main steps in the derivation of (\ref{eq:hatlajellfixed}) and (\ref{eq:lahatjp}),
the details will appear elsewhere.
By means of the WKB method, one obtains that
\begin{equation}\label{eq:WKBwr}
 \on{Wr}(\chi_+(x,\la;l),\psi^{(0)}(x,\la;l))=C \left(1- e^{i \pi I(\la,\hat{\ell})} \left(1+ \rho(\la,\hat{\ell}) \right) \right),
\end{equation}where $C \neq 0$ is a normalizing factor, and $\rho$ a correction term such that
\begin{equation}\label{eq:rhowkb}
\rho(\la,\hat{\ell})=O\big((\la^2 +\hat{\ell}^2)^{-\frac12}\big),\qquad  \la+\hat{\ell} \to \infty .
\end{equation}
Therefore, zeroes of $Q_+^*(\la)=0$ are solutions of the  equation
\begin{equation*}
 I(\la_j,\hat{\ell})=2j+1 + O\big((\la^2 +\hat{\ell}^2)^{-\frac12}\big),\qquad j \in \mathbb{Z}_{\geq 0}.
\end{equation*}
This a small perturbation of equation (\ref{eq:deflajp}), which was shown to have a unique solution $\widetilde{\la}_j(\hat{\ell})$ for every $j, \hat{\ell}$; see Lemma \ref{lem:lajas}.
One can prove that if
$\tilde{\la}_j+\hat{\ell}$ is large enough then, in the interval of length  $O\big((\la_j +\hat{\ell})^{-\frac12}\big)$, centered at $\widetilde{\la}_j(\hat{\ell})$ 
there exists a unique zero $\la'_j$ of $Q^*_+$. Moreover, the eigenfunction corresponding to $\la'_j$ has exactly $j$ zeroes in the positive semi-axis.
Hence formulae \eqref{eq:hatlajellfixed}, \eqref{eq:lahatjp} follow from formulae
\eqref{eq:lajlargej}, \eqref{eq:lajlargep}.
\end{proof}

\subsection{Asymptotics of the zeroes, the case of \texorpdfstring{$d_0>0$}{d0>0}}
As it was shown above, in the case of $d_0=0$, the zeroes of $Q_+$ are simple, real, and positive,
and precise asymptotic formulae both in the large $\hat{\la}$ and in the large $l$ regime are available.
In the case $d_0>0$, the situation is rather different, since, in general, 
the boundary value problem (\ref{eq:bvgeneral}) cannot be transformed into a self-adjoint Sturm-Liouville problem. In fact, for general $l$, the zeroes of $Q^*$ may be not all real and positive.
However, precise asymptotic formulae still can be obtained and we present them in Conjecture \ref{conj:largepd0=d1}  which complements Conjectures \ref{conj d0=d1} and \ref{conj r}.
To a partition $\mu=(\mu_1,\dots,\mu_a)$ we associate the  sequence $n^\mu_{0},n^\mu_{1},\dots$ given by
\ben \label{eq:njmu}
n^\mu_j=j-\mu_{j+1}^t, \qquad j\in \Z_{\geq 0},
\een
where $\mu^t$ is the partition dual (transposed) to $\mu$ and $\mu_j^t=0$ for $j >\mu_1$.  Notice that the sequence $n^\mu_j$ for large
$j$ does not depend on $\mu$, namely, $\nu_j^{\mu}=j$ when $j$ is large enough.
\begin{conj}\label{conj:largepd0=d1}
Fix $d_0>0$ and $k>-2$.
\begin{enumerate}[i)]
\item Fix a real $l>-\frac12$ and a solution  $\{s_i\}_{i=1}^{d_0}$, $\{t_j\}_{j=1}^{d_1}$.
There exists a sufficiently large disc $D\subset \C$ and a $j_0 \in \mathbb{Z}_{\geq 0}$ such that the zeroes of the spectral determinant $Q_+$, corresponding to  $\{s_i\}_{i=1}^{d_0}$, restricted to $\C \setminus D$ are simple and form a sequence $\{\la_j\}_{j\geq j_0}$ such that (\ref{eq:hatlajellfixed}) holds.
\item If $l$ is sufficiently large and positive, for every solution a solution  $\{s_i\}_{i=1}^{d_0}$, $\{t_j\}_{j=1}^{d_1}$ of \eqref{Gaudin BAE BLZ} with $d_0=d_1$, the zeroes of the corresponding spectral determinant $Q_+$ are all simple, real, and positive.
They form an increasing and diverging sequence $\{\la_j\}_{ j \in \mathbb{Z}_{\geq 0}}$ which satisfies the large $j$ asymptotic formula (\ref{eq:hatlajellfixed}).
\item If $l$ is sufficiently large and positive, the solutions of the Gaudin BAE \eqref{Gaudin BAE BLZ}  with $d_0=d_1$ are parameterised by partitions of $d_0$ via the asymptotic formula (\ref{s asymp}).
For any partition $\mu$ of $d_0$ and for fixed $j \in \mathbb{Z}_{\geq 0}$ the $j$-th zero
$\la_j$ of the corresponding spectral determinant $Q_+$ satisfies the large $\hat{\ell}$ asymptotic formula,
 \begin{equation}\label{eq:lajpmu}
 \left(\la_j\right)^{\frac{k+2}{3}}=\frac{\big(k+3)^{\frac{k+3}{2}}}{\big(k+2\big)^{\frac{k+2}{2}}}\,\hat{\ell} +
 \left(\frac{8 (2+k)^{1+k}}{(3+k)^{4+k}} \right)^{-\frac12} (2n^{\mu}_j+1)+ O (\hat{\ell}^{-1}), 
\end{equation}
as $\hat{\ell} \to + \infty,$ where  $n_j^{\mu}$ is as in  (\ref{eq:njmu}).
\end{enumerate}
\end{conj}
Below, we present some arguments to support our conjecture.
Conjecture \ref{conj:largepd0=d1} i) can be actually proved by means of the complex WKB method. In fact, when $l$ is fixed and $|\la|$ is large, the term of the potential
$$\sum_{i=1}^{d_0}\left(\frac{2}{(x-s_i)^2}+\frac{k+s_i/(s_i-1)}{x(x-s_i)}\right),$$
is bounded away from $0$ and $\infty$ and therefore it is negligible.
Hence, the very same asymptotic analysis developed in the case $d_0=0$ yields that the asymptotic formulae \eqref{eq:WKBwr}, \eqref{eq:rhowkb} hold when $d_0>0$ too.
We turn now our attention to the large $l$ limit, Conjecture \ref{conj:largepd0=d1} ii) and iii).
We deduce (\ref{eq:lajpmu}) using a heuristic method, that mimics the analysis
of the large momentum limit of the BLZ monster potentials \cite{CM1}. One can develop a rigorous approach similar to \cite{CM1}, but we do not pursue it in the present paper.
The heuristic method is based on the well-known fact that the bottom of the spectrum for
the boundary value problem of an anharmonic oscillator in the large $l$ limit is slightly above
the value of the spectral parameter where the potential has a double zero.
Moreover, the corresponding wave functions
are localized at the double zero of the potential. In practical terms,
one can deduce the large $\hat{\ell}$ asymptotics of the bottom of the spectrum by performing a multiple scaling limit
on equation (\ref{eq:bvgeneral}) with respect to the large parameter $\hat{\ell}=l+\frac12$. Namely, one
assumes that the coordinate $x$, the spectral parameter $\la$, and the poles  $\{s_i\}_{i=1}^{d_0}$, 
scale as powers of $\hat{\ell}$ as $\hat{\ell} \to +\infty$, and fixes the scaling parameters
in such a way that the boundary value problem in the limit is non-trivial.

Let us first illustrate this method in the case of the ground state, to obtain another derivation of the
formula (\ref{eq:lahatjp}).
As we have discussed after the introduction of the quantities $x_*$ ad $u_*$ in equations \eqref{eq:x*},\eqref{eq:u*},
the reduced potential $U$ has a double zero if and only if $\la_*=u_* \hat{\ell} +O(\hat{\ell}^{-1})$; moreover, the double zero is $x_*$. Therefore, we make the following linear change of variable
\begin{equation}\label{eq:doublescaling}
 x=\frac{k+2}{k+3}+ \left( \frac{4+2k}{3+k} \right)^\frac{1}{4} t \, \hat{\ell}^{-\frac12}, \quad \la=\frac{\big(k+3)^{\frac{k+3}{2}}}{\big(k+2\big)^{\frac{k+2}{2}}}\,\hat{\ell} + \left(\frac{8 (2+k)^{1+k}}{(3+k)^{4+k}} \right)^{-\frac12} \nu.
\end{equation}
After (\ref{eq:doublescaling}), the Schr{\"o}dinger equation
 $L^G\Psi(x)=0$ with $d_0=0$ reads 
 $$ \phi''(t)=\left(t^2- \nu+ O (\hat{\ell}^{-\frac12}) \right)\phi(t)$$ 
 and the boundary conditions (\ref{eq:bvlg0}) become 
 $$\lim_{t \to +\infty}\phi(t)=\lim_{t \to - a \hat{\ell}^{\frac12}}\phi(t)=0,\qquad a=\frac{k+2}{k+3} \left( \frac{4+2k}{3+k} \right)^{-\frac{1}{4}}.$$
Hence, up to $O(\hat{\ell}^{-\frac12})$ correction the boundary value problem (\ref{eq:bvgeneral}) reduces to
\begin{align*}
\begin{cases}
  &\phi''(t)=\left(t^2- \nu\right) \phi(t),\\
  &\lim_{t\to\pm \infty}\phi(t)=0.
\end{cases}
\end{align*}
This boundary value problem describes the spectrum of the quantum harmonic oscillator,
which is given by the sequence $\{\nu_j=2j+1\}_{j\in \mathbb{Z}_{\geq 0}}$.
Hence, using (\ref{eq:doublescaling}), we deduce that
$$\la_j=\frac{\big(k+3\big)^{\frac{k+3}{2}}}{\big(k+2\big)^{\frac{k+2}{2}}}\,\hat{\ell} + \left(\frac{8 (2+k)^{1+k}}{(3+k)^{4+k}} \right)^{-\frac12} (2j+1)+O(\hat{\ell}^{-\frac12}),$$
which is a slightly weaker\footnote{The $O(\hat{\ell}^{-\frac12})$ correction actually vanishes because the $O(\hat{\ell}^{-\frac12})$ term of the potential is odd.}
variant of (\ref{eq:lahatjp}).
We turn now our attention to the case of a non-trivial solution of Gaudin BAE satisfying the asymptotic formulae
(\ref{s asymp}). Fixing a partition $\mu=(\mu_1,\dots ,\mu_a)$ of $d_0$, we
define
\begin{align}\label{eq:triplescaling}
 x&=\frac{k+2}{k+3}+ \left( \frac{4+2k}{3+k} \right)^\frac{1}{4} t \, \hat{\ell}^{-\frac12},\notag\\
 \la&=\frac{\big(k+3)^{\frac{k+3}{2}}}{\big(k+2\big)^{\frac{k+2}{2}}}\,\hat{\ell} + \left(\frac{8 (2+k)^{1+k}}{(3+k)^{4+k}} \right)^{-\frac12}
 \nu, \notag \\
 s_i&=\frac{k+2}{k+3}+ \frac{2^{\frac{1}{4}} (k+2)^{^\frac{3}{4}}}{(k+3)^\frac{5}{4}} v_i^{\mu} \, \hat{\ell}^{-\frac12}+ \tau_i \hat{\ell}^{-1},
\end{align}
where $v_i^{\mu}$ is, as in \eqref{s asymp}, the $i$-th zero of the Wronskian of the Hermite polynomials associated
to the partition $\mu$.
Since in the new variables, the Schr{\"o}dinger equation $\psi''(x)=V(x)\psi(x)$ reads
\begin{equation*}
 \phi''(t)=\left(t^2- \nu+ \sum_{i=1}^{d_0} \frac{2}{(t-v^{\mu}_i)^2}+ O (\hat{\ell}^{-\frac12}) \right)\phi(t),
\end{equation*}
up to corrections of order $O(\hat{\ell}^{-\frac12})$, the boundary value problem (\ref{eq:bvlg0}) reduces to
\begin{align*}
\begin{cases}
  &\phi''(t)=\left(t^2+ \sum_{i=1}^{d_0} \frac{2}{(t-v^{\mu}_i)^2})- \nu \right) \phi(t),\\
  &\lim_{t\to\pm \infty}\phi(t)=0.
\end{cases}
\end{align*}
Non-trivial solutions of the above BVP are also explicitly known, see \cite[equation (3.2)]{CM1}.
The spectral points are $\nu_j^{\mu}=2 n_j^{\mu}+1$ with $n_j^{\mu}$ as in (\ref{eq:njmu}).
From this, using (\ref{eq:triplescaling}), we deduce the asymptotic formula (\ref{eq:lajpmu}) for the bottom of the spectrum in the large $l$ limit.

\subsection{Asymptotics of zeroes of the spectral determinants of BLZ opers}
The large momentum analysis of BLZ opers was developed in \cite{CM1}, where the focus was slightly different from the present paper.
In order to make the comparison simpler, we summarise the content of \cite{BLZ4} and \cite{CM1} in a single conjecture
similar to Conjecture \ref{conj:largepd0=d1}.
\begin{conj}\label{conj:largepblz}
 Fix $d\geq 1$ and $k>-2$ and let $l>-\frac12$.  For every $k$ and $l$, let $\bar{k}$ and $\bar{l}$ be given by \eqref{parameters change}.
\begin{enumerate}[i)]
\item \label{conj:largepblz1} For any $\bar{l}$, the number of solutions to the trivial monodromy system \eqref{BLZ BAE} is  at most the number of partition of $N$. Moreover, the system has exactly partition of $N$ solutions if $\bar{l}$ is generic.
\item\label{conj:largepblz2} If $\bar{l}$ is large enough, the zeroes of the spectral determinant $\overline{Q}_+$ are all simple, real, and positive.
 Denoting by $\{\bar{\la}_j\}_{ j \in \mathbb{Z}_{\geq 0}}$ the increasing sequence of such zeroes then the sequence
$$\{ \la_j\}_{ j \in \mathbb{Z}_{\geq 0}}=\{(k+3)^{\frac{2}{k+3}}\, \bar{\la}_j\}_{ j \in \mathbb{Z}_{\geq 0}}$$
satisfies the large $j$ asymptotics
 \eqref{eq:lajlargej}.
\item \label{conj:largepblz3} For every partition $\mu$ of $d$ there exists a solution $z_j^{\mu}$ of \eqref{BLZ BAE} such that
  \begin{equation}\label{eq:zjlargelb}
  z_j^{\mu}= \frac{(1-\bar{k})}{\bar{k}}(\bar{l}+\frac12)^2+ \frac{2^{\frac94} (1-\bar{k})^{-\frac94}}{\bar{k}} v_j^{\mu} (\bar{l}+\frac12)^{\frac32}+
  O(\bar{l}+\frac12), \qquad \bar{l} \to \infty,
 \end{equation}
where $v_j^\mu$ are  as in \eqref{s asymp}. Denoting by $\lbrace \bar{\la}^{\mu}_j\}_{j \in \mathbb{Z}_{\geq 0}}$ the sequence of roots of the corresponding spectral determinant $\overline{Q}$, the sequence $\{ \la^{\mu}_j \}_{j \in \mathbb{Z}_{\geq 0}}=\{(k+3)^{\frac{2}{k+3}}\,\bar{\la}^{\mu}_j\}_{j \in \mathbb{Z}_{\geq 0}}$ satisfies the large $l$ asymptotics \eqref{eq:lajlargep}.
 \end{enumerate}
\end{conj}
\begin{rem}
Conjecture \ref{conj:largepblz} \ref{conj:largepblz1})  can be found in \cite[Section 3]{BLZ4}. Conjecture \ref{conj:largepblz} \ref{conj:largepblz2}) is the precise statement of \cite[equation (2.2)]{BLZ4}. Conjecture \ref{conj:largepblz} \ref{conj:largepblz3}) can be found in \cite{CM1}.  In \cite[Theorem 5.5]{CM1} it is proved that any solution of \eqref{BLZ BAE} admits the asymptotic
expansion \eqref{eq:zjlargelb} for some partition $\mu$ of $d_0$; in \cite[Conjecture 5.8]{CM1} it is conjectured that for every partition $\mu$, there exists a unique solution of \eqref{BLZ BAE} satisfying \eqref{eq:zjlargelb};  finally formula \eqref{eq:lajlargep} is equivalent to \cite[formula (5.10)]{CM1}.\qed
\end{rem}

\section{Comparison between \texorpdfstring{$Q$}{Q}-functions}\label{comparison section}
In this section we compare the spectral determinant $Q_+$ and $\overline{Q}_+$ of respectively
$L^G$ opers with $d_1=d_0$ and the BLZ opers with $d=d_0$, assuming Conjectures \ref{conj:largepd0=d1} and \ref{conj:largepblz}.
\subsection{The case of \texorpdfstring{$d_0>0$}{d0>0}}
In the case $d_0=0$, see equation \eqref{eq:QQbargroundstate}, the spectral determinant $Q_+$ and
$\overline{Q}_+$ were shown to coincide using the fact that $L^G$ opers and BLZ opers are related by a change of variables.
We focus now on the case $d_0>0$, for which there is no direct relation between $L^G$ opers and BLZ opers.
Let us fix
$d_0>0$, $k$, and $l$ and denote by 
$$\underline{s}=\{s_i\}_{i=1}^{d_0},\qquad \underline{t}=\{t_j\}_{j=1}^{d_0}$$
solutions to the Gaudin BAE \eqref{Gaudin BAE BLZ} with $d_0=d_1$ and by
$$\underline{z}=\{\bar{z}_j\}_{j=1}^{d_0}$$ 
solutions to the trivial monodromy equations for BLZ opers \eqref{BLZ BAE},
with $\bar{k}$ and
$\bar{l}$ given by \eqref{parameters change}.
The sets of solutions to these two algebraic systems
induce two sets of
of spectral determinants $Q_+(\cdot;\underline{s})$ and $\overline{Q}_+(\cdot;\underline{z})$, which we aim to compare.
So far we have proved or provided evidence for the following facts.

1. Each of these spectral determinants satisfy the same $QQ$-system \eqref{QQsystem}, hence of the same
Bethe equations for Quantum KdV \eqref{eq:BAEQintro}. This is the content of Theorem \ref{thm:QQsystem} in the case of $L^G$ opers, while in the case of BLZ opers was proved in \cite{BLZ4}.

2. If $l$ is generic, the two sets
of spectral determinant have the same cardinality, namely the number of partitions of $d_0$.
Moreover, if  $l$ is large enough, solutions to \eqref{Gaudin BAE BLZ} and to \eqref{BLZ BAE} are explicitly parameterised by partitions of $d_0$ via formulae (\ref{s asymp}) and (\ref{eq:zjlargelb}). We denote them
by $\underline {t}^\mu,\underline{s}^{\mu}$ and $\underline{z}^{\mu}$.
This is the content of Conjecture \ref{conj d0=d1} and Conjecture \ref{conj:largepblz} i), ii).

3. Let
$Q_+(\la,\underline{s}^{\mu})$ and $\overline{Q}_+(\bar{\la},\underline{z}^{\mu})$ be the spectral determinants corresponding to $\underline {t}^\mu,\underline{s}^{\mu}$ and $\underline{z}^{\mu}$.
If $l$ is sufficiently large and positive,
the zeroes of $Q_+(\la,\underline{s}^{\mu})$ and $\overline{Q}_+(\bar{\la},\underline{z}^{\mu})$
are simple, real and positive.
If the zeroes of $Q_+(\la,\underline{s}^{\mu})$ and $\overline{Q}_+(\bar{\la},\underline{z}^{\mu})$ are ordered into the increasing sequence $\lbrace \la_j\rbrace_{j \in \mathbb{Z}_{\geq 0}}$ and $\lbrace \bar{\la}_j\rbrace_{j \in \mathbb{Z}_{\geq 0}}$, then
$\lbrace \la_j\rbrace_{j \in \mathbb{Z}_{\geq 0}}$ and $\lbrace (k+3)^{\frac{2}{k+3}}\,\bar{\la}_j\rbrace_{j \in \mathbb{Z}_{\geq 0}}$ satisfy the large $j$ (\ref{eq:hatlajellfixed}) and large $l$ asymptotics \eqref{eq:lajpmu}. This is the content of Conjecture \ref{conj:largepd0=d1} and Conjecture \ref{conj:largepblz} iii).

We show below that the above points 1,2, and 3, together with \eqref{eq:BAEQintro}, imply that if $l$ is large enough there exists a constant $c \neq 0$ -- possibly depending on $\mu,l,k$ -- such that
$\overline{Q}_+(\la,\underline{z}^{\mu})= c \,  Q_+((k+3)^{\frac{2}{k+3}}\, \la,\underline{s}^{\mu})$.
To this aim we briefly review below the results of \cite{CM2} on the classification of solutions of Bethe equations for Quantum KdV \eqref{eq:BAEQintro}
 in the large momentum limit\footnote{We remark that in \cite{CM2} the parameters are
$\alpha>0$ and $p$, with $ \alpha=\frac{1}{k+2}$ and $p=\frac{l+\frac12}{k+3}.$}.

\subsection{Large momentum analysis of the Bethe equations}
We fix a $k$, $-2<k<-1$ and consider the class of strictly increasing sequence of positive real numbers
$\lbrace\la_j\rbrace_{j \in \mathbb{Z}_{\geq 0}}$, that satisfies the large $j$ asymptotics \eqref{eq:hatlajellfixed}.
By our hypothesis on $k$, $\la_j=O(j^{1+\delta})$ with $\delta=\frac{1-k}{k+3}>0$, see \eqref{eq:hatlajellfixed}. Therefore the
Hadamard product
\begin{align}\label{eq:latoQ}
& Q(\la;\lbrace\la_j\rbrace_{j \in \mathbb{Z}_{\geq 0}})= \prod_{j \in \mathbb{Z}_{\geq 0}} \left(1-\frac{\la}{\la_j} \right),
\end{align}
converges to an entire function (of order of growth $\frac{k+3}{2}<1$). 
We assume that the sequence $\lbrace \la_j\rbrace$ satisfies the following system of equations:
\begin{equation}
\label{eq:BAEQ}
\gamma^{-2} \frac{Q\left(\la_{i} \, q^{-2};\lbrace\la_j\rbrace_{j \in \mathbb{Z}_{\geq 0}} \right)}{Q\left(\la_i \,q^2;\lbrace\la_j\rbrace_{j \in \mathbb{Z}_{\geq 0}}\right)} = -1 ,\; i \in \mathbb{Z}_{\geq 0}, \quad q=e^{\frac{k+2}{k+3}\pi i}, \; \gamma=e^{\frac{2l+1}{k+3}\pi i}.
\end{equation}
The above equations are the Bethe equations \eqref{eq:BAEQintro} and its solutions are called Bethe roots.
Define the  function $z:[0,\infty) \to [-\frac{2l+1}{k+3}, \infty)$ by
\begin{subequations} \label{eq:zfunction}
\begin{align}
  z(\la;\lbrace\la_j\rbrace_{j \in \mathbb{Z}_{\geq 0}})&= -\frac{2l+1}{k+3} + \frac{1}{2\pi i}
  \log \frac{Q\left(\la \, q^{-2};\lbrace\la_j\rbrace_{j \in \mathbb{Z}_{\geq 0}} \right)}  {Q\left(\la \,q^2;\lbrace\la_j\rbrace_{j \in \mathbb{Z}_{\geq 0}} \right)}, \label{eq:zfunction1}\\
 z(0)&=-\frac{2l+1}{k+3}, \label{eq:zfunction2}
\end{align}
\end{subequations}
where we used the analytic continuation from $\la=0$. Then  $z(\la;\lbrace\la_j\rbrace_{j \in \mathbb{Z}_{\geq 0}})$
is real analytic and
strictly increasing \cite[Lemma 1.1]{CM2}. It is known as the \textit{counting function} in the physics literature.
Using \eqref{eq:zfunction} and \eqref{eq:BAEQ}, and since $z$ is strictly monotone,
we deduce that there exists a strictly increasing sequence of integer numbers
$\lbrace n_j^{Q}\rbrace_{j \in \mathbb{Z}_{\geq 0}}$, with $n_0^Q \geq -\frac{2l+1}{k+3}$, such that
\begin{equation*}
 z(\la_i;\lbrace\la_j\rbrace_{j \in \mathbb{Z}_{\geq 0}})=n^Q_i+ \frac12.
\end{equation*}
This is the logarithmic form of the Bethe equations.
The sequence $n_j^Q$ is called the sequence of root-numbers\footnote{If an integer does not belong to the set of root-numbers is called a hole-number.}.
It is proved in \cite[Theorem 4.9]{CM2} that, owing to the large $j$ asymptotics \eqref{eq:lajlargej},
$n_j^Q$ is a stabilising sequence, namely $n_j^Q=j$ if $j$ is large enough.
As it is well-known, stabilising sequences are parameterised by partitions via formula
\eqref{eq:njmu}. It then follows that there exists a partition $\mu$ such that
\begin{equation}\label{eq:zbetheequationsmu}
z(\la_i;\lbrace\la_j\rbrace_{j \in \mathbb{Z}_{\geq 0}})=n^{\mu}_i+ \frac12, \quad i \in \mathbb{Z}_{\geq 0},
\end{equation}
where $n_i^{\mu}$ is as in \eqref{eq:njmu}.
We consider formula \eqref{eq:zbetheequationsmu}, together with \eqref{eq:zfunction} and \eqref{eq:latoQ}, as a system of equations for the sequence $\lbrace \la_j\rbrace_{j \in \mathbb{Z}_{\geq 0}}$. Since the value $z(0)=-\frac{2l+1}{k+3}$ is fixed and $z$ is monotonically increasing,  $l$ must be sufficenlty large that $z(0)-\frac12 < n^{\mu}_0$. The latter inequality is satisfied for every partition $\mu$ of a given integer $d_0$ if and only if
$2l+1> (k+3) (d_0-\frac12)$ \cite{CM2}.
We have the following theorem.
\begin{thm}\label{thm:BetheEquations}\cite[Theorem 1]{CM2}
Fix a $k<-1$ and a $d_0\geq0$.
\begin{enumerate}[i)]
\item If $l$ is large enough, for every partition $\mu$ of $d_0$, the
logarithmic Bethe equations \eqref{eq:zbetheequationsmu}
have a unique solution $\lbrace \la_j^{\mu}\rbrace_{j \in \mathbb{Z}_{\geq 0}}$ in the space of positive monotone sequences satisfying \eqref{eq:hatlajellfixed}.
\item Moreover, the solution $\lbrace \la_j^{\mu}\rbrace_{j \in \mathbb{Z}_{\geq 0}}$ satisfies the large momentum asymptotics \eqref{eq:lajpmu}. \qed
\end{enumerate}
\end{thm}

\subsection{Comparison}
Assuming Conjectures \ref{conj:largepd0=d1} and \ref{conj:largepblz},
the comparison between $Q_+$ and $\overline{Q}_+$ functions for $l$ large and positive and $-2<k<-1$ is a direct corollary of the above facts.
Indeed,
due to the large $j$ asymptotics \eqref{eq:hatlajellfixed}, the sequences
$\lbrace \la_j\rbrace_{j \in \mathbb{Z}_{\geq 0}}$ and $\lbrace (k+3)^{\frac{2}{k+3}}\,\bar{\la}_j\rbrace_{j \in \mathbb{Z}_{\geq 0}}$ satisfy the logarithmic Bethe equations \eqref{eq:zbetheequationsmu} for some partitions $\mu_1$ and $\mu_2$, respectively.
Moreover, due to the large $l$ asymptotics \eqref{eq:lajpmu} and Theorem \ref{thm:BetheEquations} ii),
the partitions  $\mu_1$ and $\mu_2$ coincide with $\mu$.
Hence, due to Theorem \ref{thm:BetheEquations} i), $\la_j= (k+3)^{\frac{2}{k+3}}\,\bar{\la}_j$ for all $j \in \mathbb{Z}_{\geq 0}$.
Therefore $Q((k+3)^{\frac{2}{k+3}}\, \la,\underline{s}^{\mu})$ and $ \overline{Q}(\la,\underline{z}^{\mu})$ coincide up to a multiplicative constant: there exists a constant $c=c(k,l,\mu) \neq 0$ such that
 \begin{equation}\label{eq:Q=barQ}
  \overline{Q}_+(\la,\underline{z}^{\mu})= c \; Q_+((k+3)^{\frac{2}{k+3}}\, \la,\underline{s}^{\mu}),\qquad  \forall \la \in \C.
 \end{equation}

We think that our argument can be completed to a rigorous proof  of \eqref{eq:Q=barQ}.  We recall here the technical difficulties on the way. First, our approach is based on a partially conjectural classification of the opers \eqref{L Gaudin-1} in terms of solutions to the Gaudin Bethe Ansatz equations, and on the partially conjectural classification of the latter solutions in the limit $l \to +\infty$, see Conjecture \ref{conj r}. This conjecture is established for non-degenerate solutions, see Proposition \ref{old prop}.

Second, the study of the local behaviour at $\la=0$ of the functions $Q_{\pm}$ in the case of $L^G$ is much more involved
than in the case of $L^{BLZ}$. We gave a proof that $Q_{+}$ is entire, without computing the order of zero, and we provided an upper bound for the (possible) pole of $Q_{-}$ at $\la=0$, see Proposition \ref{prop:extendability}. Third,  asymptotics  \eqref{eq:hatlajellfixed}, \eqref{eq:lajpmu} of zeroes of $Q_{+}$ is difficult to compute. For monster potential operators $L^{BLZ}$ it is done in \cite{CM2} (with some restrictions), a similar method could be used for $L^{G}$. We give a partial argument in Section \ref{heu section}.
Finally, as usual, it is not clear how to extend such large $l$ considerations to all values.

\medskip

We discussed the correspondence \eqref{eq:Q=barQ} for the case of $r=0$ (that is $d_0=d_1$), $-2<k<-1$ and $l$ large. Due Propositions \ref{prop:R0Qpm} and \ref{prop:R1Qpm} a
similar formula should hold for all $r$, $-2<k<-1$, generic $l$,  and  $d_0\geq r^2$. 

\bigskip

{\bf Acknowledgements.} DM and EM would like to thank the Simons Center for Geometry and Physics where this collaboration has started. EM and AR would like to thank the Group of Mathematical Physics of Lisbon University where this project started for hospitality. EM would like to that the University of Bergamo where the project was developed for the hospitality and support. EM is partially supported by the Simons foundation grant number \#709444. AR is partially supported by funds of INFN (Istituto Nazionale di Fisica Nucleare) by IS-CSN4 Mathematical Methods of Nonlinear Physics and  by the INdAM–GNFM Project CUP-E53C22001930001. DM is supported by the FCT projects UIDB/00208/2020,
2021.00091.CEE  CIND, and 2022.03702.PTDC (GENIDE), and he is a member of the COST Action CA21109 CaLISTA.


\begin{thebibliography}{GLVW22}


\bibitem[BLZ96]{BLZ1}
V.~Bazhanov, S.~Lukyanov, and A.~Zamolodchikov,
{\it Integrable structure of conformal field theory, quantum {K}d{V}
  theory and thermodynamic {B}ethe Ansatz},
 Comm. Math. Physics, {\bf 177}, no. 2 (1996), 381-398

\bibitem[BLZ97]{BLZ2}
V.~Bazhanov, S.~Lukyanov, and A.~Zamolodchikov,
{\it Integrable structure of conformal field theory {II}. {Q}-operator and
  {DDV} equation}, Comm. Math. Physics, {\bf 190}, no. 2 (1997), 247-278

\bibitem[BLZ01]{BLZ5}
V.~Bazhanov, S.~Lukyanov, and A.~Zamolodchikov,
{\it Spectral determinants for Schr{\"o}dinger equation and Q operators of
  conformal field theory}, J.Statist.Physics, {\bf 102}, no. 3-4 (2001), 567-576

\bibitem[BLZ04]{BLZ4}
V.~Bazhanov, S.~Lukyanov, and A.~Zamolodchikov,
{\it Higher-level eigenvalues of {Q}-operators and {S}chroedinger
  equation},
Adv. Theor. Math. Physics, {\bf 7}, no. 4 (2004), 711-725




\bibitem[CGM23]{CGM}
G.~Cotti, D.~Guzzetti and D.~Masoero, {\it Asymptotic solutions for linear ODEs with not-necessarily meromorphic coefficients: a Levinson type theorem on complex domains, and applications},
arXiv preprint, {arXiv:2310.19739}, (2023).

\bibitem[CM21]{CM1}
R.~Conti and D.~Masoero, {\it Counting monster potentials}, J. High Energy Phys, no. 02 (2021), Paper No. 059, 59pp

\bibitem[CM23]{CM2}
R.~Conti and D.~Masoero, {\it On solutions of the bethe Ansatz for the quantum KdV model},
Comm. Math. Physics, {\bf 402}, no.1 (2023), 335-390

\bibitem[DDT07]{DDT}
P.~Dorey, C.~Dunning, and R.~Tateo,
{\it The ODE/IM correspondence},
J. Physics A, {\bf 40}, no. 32 (2007), R205-R283

\bibitem[DT99a] {DTa} P. Dorey and R. Tateo, {\it Anharmonic oscillators, the thermodynamic Bethe Ansatz, and nonlinear
integral equations}, J. Phys. A, {\bf{32}} (1999) L419–L425

\bibitem[DT99b]{DT}
P.~Dorey and R.~Tateo,
{\it  On the relation between {S}tokes multipliers and the {T}-{Q} systems
  of conformal field theory}, Nuclear Physics B, {\bf 563}, no. 3 (1999), 573-602



\bibitem[EMOT53]{EMOT}
A.~Erd{\'e}lyi, W.~Magnus, F.~Oberhettinger, and F.~Tricomi, Higher transcendental functions, {V}ol. {II}, McGraw-Hill Book Company, Inc., New York-Toronto-London, 1953

\bibitem[Erd56]{E}
A.~Erd{\'e}lyi, Asymptotic Expansions, Dover Books on Mathematics, Dover Publications, 1956

\bibitem[ESV]{ESV}
S.~Ekhammar, H.~Shu, and D.~Volin,
{\it Extended systems of Baxter Q-functions and fused flags I:
  simply-laced case}, arXiv preprint, {arXiv: 2008.10597}, (2020)


\bibitem[FF11]{FF}
B.~Feigin and E.~Frenkel,
{\it Quantization of soliton systems and {L}anglands duality}, Adv. Stud. Pure Math., {\bf 61} (2011), 
   Math. Soc. Japan, Tokyo, 185-274

\bibitem[FH18]{FH}
E.~Frenkel and D.~Hernandez,
{\it Spectra of quantum {K}dv hamiltonians, {L}anglands duality, and
  affine opers}, Comm. Math. Physics, {\bf 362}, no. 2 (2018), 361-414

\bibitem[FHV12]{FHV}
G.~Felder, A.~Hemery, and A.~Veselov,
{\it  Zeros of Wronskians of Hermite polynomials and Young diagrams},  Physica D, {\bf 241}, no. 23-24 (2012), 2131-2137

\bibitem[FJM17]{FJM}
B.~Feigin, M.~Jimbo, and E.~Mukhin, {\it  Integrals of motion from quantum toroidal algebras}, J. Physics A, 
 {\bf 50}, no. 46 (2017), 464001, 28pp

\bibitem[FJM19]{FJM2}
B.~Feigin, M.~Jimbo, and E.~Mukhin,
{\it The ${(\mathfrak{gl}_{m},\mathfrak{gl}_{n})}$ duality in the quantum
  toroidal setting}, Comm. Math. Physics, {\bf 367}, no. 2 (2019), 455-481

\bibitem[FJMM15]{FJMM0}
B.~Feigin, M.~Jimbo, T.~Miwa, and E.~Mukhin,
{\it Quantum toroidal $\mathfrak{gl}(1)$ and Bethe Ansatz}, J. Physics A {\bf 48}, no. 24 (2015), 244001, 27 pp

\bibitem[FJMM17]{FJMM1}
B.~Feigin, M.~Jimbo, T.~Miwa, and E.~Mukhin,
{\it Finite type modules and Bethe Ansatz for quantum toroidal $\mathfrak{gl}(1)$}, Comm. Math. Physics, {\bf 356}, no. 1 (2017), 285–327

\bibitem[GLVW22]{GLVW}
D.~Gaiotto, J.~Lee, B.~Vicedo, and J.~Wu,
{\it Kondo line defects and affine Gaudin models},
J. High Energy Physics, no.1  (2022), Paper No. 175, 74 pp
 
\bibitem[GKO86]{GKO} P.~Goddard, A.~Kent, and D.~Olive, {\it Unitary representations of the Virasoro and super-Virasoro algebras}
Comm. Math. Physics, {\bf 103} (1986),  105–119

\bibitem[KL21]{KL} G.~Kotousov and S.~Lukyanov, {\it ODE/IQFT correspondence for the generalized affine $\slt(2)$ Gaudin model}, J. High Energy Physics, no. 09 (2021), Paper No. 201


\bibitem[Lit13]{Lit} A. Litvinov, {\it On spectrum of ILW hierarchy in conformal field theory}, J. High Energy Physics, no. 11 (2013), Paper No. 155

\bibitem[LMV16]{MLV}
K.~Lu, E.~Mukhin, and A.~Varchenko,
{\it On the Gaudin model associated to lie algebras of classical types}, J. Math. Phys., {\bf 57}, no. 10 (2016), 101703, 1-22

\bibitem[LVY19]{LVY1}
S.~Lacroix, B.~Vicedo, and C.~Young,
{\it Affine Gaudin models and hypergeometric functions on affine opers},
Adv. Math., {\bf 350} (2019), 486-546

\bibitem[LVY20]{LVY2}
S.~Lacroix, B.~Vicedo, and C.~Young,
{\it  Cubic hypergeometric integrals of motion in affine Gaudin models}, Adv. Theor. Math. Physics, {\bf 24}, no. 1 (2020), 155-187

\bibitem[MR20]{MR1}
D.~Masoero and A.~Raimondo, {\it Opers for {H}igher {S}tates of {Q}uantum {K}d{V} {M}odels}, Comm. Math. Physics, {\bf 378}, no. 1 (2020), 1-74

\bibitem[MR23]{MR2}
D.~Masoero and A.~Raimondo. {\it Feigin-Frenkel-Hernandez opers and the $QQ$-system},
arXiv preprint, {arXiv:2312.01955}, (2023)

\bibitem[MRV16]{MRV}
D.~Masoero, A.~Raimondo, and D.~Valeri, 
{\it Bethe Ansatz and the spectral theory of affine Lie
  algebra-valued connections {I}. {T}he simply-laced case}, Comm. Math. Physics, {\bf 344}, no. 3 (2016), 719-750

\bibitem[MRV17]{MRV2}
D.~Masoero, A.~Raimondo, and D.~Valeri,
{\it Bethe Ansatz and the spectral theory of affine Lie
  algebra-valued connections {II}: {T}he non simply--laced case}, Comm. Math. Phys., {\bf 349}, no. 3 (2017), 1063-1105

\bibitem[MTV06]{MTV1}
E.~Mukhin, V.~Tarasov, and A.~Varchenko,
{\it Bispectral and $(\gl_n,\gl_m)$ dualities},
Funct. Anal. Other Math., {\bf 1}, no.1 (2006), 55-80

\bibitem[MTV08]{MTV2}
E.~Mukhin, V.~Tarasov, and A.~Varchenko,
{\it Bispectral and $(\gl_n, \gl_m)$ dualities, discrete versus differential}, Adv. Math., {\bf 218},  no. 1 (2018), 216-265

\bibitem[MTV09]{MTV3}
E.~Mukhin, V.~Tarasov, and A.~Varchenko,
{\it A Generalization of the Capelli Identity}, Prog. Math., {\bf 270} (2009), 383-398

\bibitem[MV04]{MV1}
E.~Mukhin and A.~Varchenko,
{\it Critical points of master functions and flag varieties}, Comm. in Contemp. Math., {\bf 06}, no. 1 (2004), 111-163

\bibitem[MV08]{MV2}
E.~Mukhin and A.~Varchenko,
{\it Quasi-polynomials and the Bethe Ansatz},
 Geom. Topol. Monographs, {\bf 13} (2008), 385-420

\bibitem[MV14]{MV3}
E.~Mukhin and A.~Varchenko,
{\it On the number of populations of critical points of master functions}, J. Singul., {\bf 8} (2014), 31-38 

\bibitem[Obl99]{O}
A.~Oblomkov, {\it Monodromy-free Schr\"odinger operators with quadratically increasing
  potentials}, Theoret. and Math. Phys., {\bf 121}, no. 3 (1999), 374-386


\end{thebibliography}
\end{document}